\documentclass[12pt,reqno]{amsart}
\usepackage[table]{xcolor}

\usepackage[letterpaper, margin = 0.965truein]{geometry} 

\usepackage{amsmath,amsthm,amssymb}
\usepackage{mathrsfs}
\usepackage{enumitem}
\usepackage{pgf,tikz,pgfplots}
\usetikzlibrary{arrows}
\usepackage[unicode]{hyperref}
\hypersetup{colorlinks=true, linkcolor=blue, citecolor=blue, filecolor=blue, urlcolor=blue}
\usepackage{comment}
\usepackage{chngcntr}
\usepackage{mathabx}
\usepackage{bm}
\usepackage{etoolbox}
\usepackage{algpseudocode}
\usepackage{enumitem}
\usepackage{algorithm}
\usepackage{subcaption}
\usepackage{float}
\usepackage{pkgfile}
\usepackage{bm}
\usepackage{extarrows}
\usepackage{algorithm}
\usepackage{algpseudocode}
\usepackage{enumitem}
\usepackage{booktabs}
\usepackage{multirow}
\usepackage{threeparttable}
\usepackage{caption}
\usepackage{mathtools}
\usepackage{subfloat}
\usepackage{algorithm}
\usepackage{lipsum} 
\usepackage{array}
\usepackage{colortbl}
\usepackage{makecell}
\usepackage{enumitem}
\usepackage{setspace}

\usepackage[comma,sort&compress, numbers]{natbib}

\allowdisplaybreaks

\newtheorem{theorem}{Theorem}[section]
\newtheorem{corollary}[theorem]{Corollary} 
\newtheorem{lemma}[theorem]{Lemma} 
\newtheorem{proposition}[theorem]{Proposition} 

\theoremstyle{definition}
\newtheorem{definition}[theorem]{Definition} 
\newtheorem{remark}[theorem]{Remark} 
\newtheorem{example}[theorem]{Example} 
\newtheorem{assumption}{Assumption}

\newcommand{\bH}{\bm{H}}
\newcommand{\bp}{\bm{p}}
\newcommand{\bw}{\bm{w}}
\newcommand{\bW}{\bm{W}}

\newcommand{\bcW}{\bm{\mathcal{W}}}

\newcommand{\sfK}{\mathsf{K}}
\newcommand{\sfH}{\mathsf{H}}
\newcommand{\one}{\bm{1}}

\newcommand{\sfL}{\mathsf{L}}
\newcommand{\MMD}{\mathrm{MMD}}
\newcommand{\EMMD}{\mathrm{ECMMD}}
\newcommand{\emmdstat}{\EMMD[\sfK,\bcW_n,\sZ_n]}
\newcommand{\emmdstatsq}{\EMMD^{2}[\sfK,\bcW_n,\sZ_n]}
\newcommand{\rmd}{\mathrm{d}}

\title[A Kernel-Based Conditional Two-Sample Test Using Nearest Neighbors]{A Kernel-Based Conditional Two-Sample Test Using Nearest Neighbors \\
\footnotesize (with Applications to Calibration, Regression Curves, and Simulation-Based Inference)}

\author[Chatterjee]{Anirban Chattejee$^\ast$}\thanks{$^\ast$The first two authors contributed equally to the paper.}
\address{Department of Statistics and Data Science, University of Pennsylvania, Philadelphia, USA } 
\email{anirbanc@wharton.upenn.edu} 

\author[Niu]{Ziang Niu$^\ast$} 
\address{Department of Statistics and Data Science, University of Pennsylvania, Philadelphia, USA } 
\email{ziangniu@wharton.upenn.edu }

\author[Bhattacharya]{Bhaswar B. Bhattacharya}
\address{Department of Statistics and Data Science, University of Pennsylvania, Philadelphia, USA } 
\email{bhaswar@wharton.upenn.edu}


\begin{document}

\maketitle

\begin{abstract} 
In this paper we introduce a kernel-based measure for detecting differences between two conditional distributions. Using the `kernel trick' and nearest-neighbor graphs, we propose a consistent estimate of this measure which can be computed in nearly linear time (for a fixed number of nearest neighbors). Moreover, when the two conditional distributions are the same, the estimate has a Gaussian limit and its asymptotic variance has a simple form that can be easily estimated from the data. The resulting test attains precise asymptotic level and is universally consistent for  detecting differences between two conditional distributions. We also provide a resampling based test using our estimate that applies to the conditional goodness-of-fit problem, which controls Type I error in finite samples and is asymptotically consistent with only a finite number of resamples. A method to de-randomize the resampling test is also presented.   The proposed methods can be readily applied to a broad range of problems, ranging from classical nonparametric statistics to modern machine learning. Specifically, we explore three applications: testing model calibration, regression curve evaluation, and validation of emulator models in simulation-based inference. We illustrate the superior performance of our method for these tasks, both in simulations as well as on real data. In particular, we apply our method to (1) assess the calibration of neural network models trained on the CIFAR-10 dataset, (2) compare regression functions for wind power generation across two different turbines, and (3) validate emulator models on benchmark examples with intractable posteriors and for generating synthetic `redshift' associated with galaxy images. 
\end{abstract}

\section{Introduction}\label{sec:introduction} 

The conditional 2-sample problem is to test whether the conditional distributions of two response variables $X$ and $Y$ are the same, given a set of covariates $Z$. More formally, given independent samples $(X_1,Y_1,Z_1), (X_2,Y_2,Z_2),\ldots,(X_n,Y_n,Z_n)$ from a joint distribution $ P_{XYZ}$, we want to test the following null hypothesis: 
\begin{align}\label{eq:H0H1}
    \bH_0: P_{X|Z} =  P_{Y|Z} \text{ almost everywhere } P_Z , 
\end{align}
where $P_{X|Z}$, $P_{Y|Z}$ denote the conditional distributions of $X|Z$ and $Y|Z$, respectively, and $P_Z$ is the marginal distribution of $Z$. This problem appears in different guises in the literature, both classical and contemporary. The following are three examples:

\begin{example} (Calibration tests for predictive models)
Many problems in supervised learning involve estimating a conditional probability distribution $P_{Y|Z}$ of a response $Y$ given a feature vector $Z$. For models which output such predictive distributions, statistical guarantees  beyond accuracy are necessary for quantifying predictive uncertainties. One such guarantee is calibration, which, loosely stated, ensures that almost every prediction matches the conditional distribution of the response given this prediction. Calibration has been used to provide predictive guarantees in meteorological and statistical studies for many years \citep{degroot1983comparison,murphy1977reliability}. 
This notion has seen a resurgence in modern machine learning, following the breakthrough work of \citet{guo2017calibration}, which showed that
common neural network architectures trained on image and text data are often miscalibrated. This prompted a slew of work on different notions of calibration for classification (where $Y$ is categorical) \cite{degroot1983comparison,murphy1977reliability,kull2015novel,kumar2018trainable,guo2017calibration,wenger2020non,widmann2019calibration,lee2023,vaicenavicius2019evaluating} and regression problems (where $Y$ is continuous)
\cite{ho2005calibrated,rueda2007calibration,gneiting2007probabilistic,taillardat2016calibrated,song2019distribution}. Specifically, for a $r$-class classification problem one can define calibration as follows \cite{widmann2019calibration}: Consider a pair of random of variables $(Y, Z)$ with joint distribution $P_{YZ}$ over a space $\cY \times \cZ$, with $\cY = \{1,2,\ldots, r\}$ and a predicted model $f:\cZ\ra \cP_{\cY}$, where 
    \begin{align*}
        \cP_{\cY} = \left\{\bp = (p_{1},p_{2},\ldots, p_{r}):0\leq p_{s}\leq 1, \sum_{s=1}^{r}p_{s}=1\right\} , 
    \end{align*}
    is the set of distributions over $\cY$. Then $f$ is said to be {\it calibrated } if and only if, for $1 \leq s \leq r$, 
    \begin{align}\label{eq:classcalibration}
        \P\left[Y = s \middle| f(Z) = \bp\right] = p_{s} \quad \text{ for all } \bp\in f(\cZ)\subseteq\cP_{\cY} , 
    \end{align} 
where $f(Z)$ is the vector of predicted probabilities. Calibration testing can be formulated as a conditional 2-sample problem by noting that if $X\sim \mathrm{Multinomial}(r, f(Z))$, then by definition of the multinomial distribution: $\P\left[X = s \middle| f(Z) = \bp \right] = p_{s}$, for $1 \leq s \leq r$. 
Hence, from \eqref{eq:classcalibration}, we say the predictive model $f$ is calibrated if and only if,
\begin{align}\label{eq:classification}
     P_{Y|f(Z)} =  P_{X|f(Z)} \text{ almost everywhere } P_Z.
\end{align}
Note that $f$ is a known trained model, hence given data $(Y_1, Z_1), (Y_2, Z_2), \ldots, (Y_n, Z_n)$ from $P_{YZ}$, it is usually easy to draws samples $X_i \sim \mathrm{Multinomial}(r, f(Z_i))$, for $1 \leq i \leq n$. Hence, \eqref{eq:classification} can be implemented as a conditional 2-sample problem based on the samples $(X_1,Y_1,f(Z_1)),\ldots,(X_n,Y_n,f(Z_n))$. For details on the implementation, comparison with existing methods, and extension to regression problems, see Section \ref{sec:classificationregression}. 

\end{example}

\begin{example} (Comparing regression curves) A prototypical example of conditional 2-sample testing is the problem of comparing two regression curves. The curves usually correspond to the mean functions of a control and a treatment outcome given a collection of covariates $Z$. For instance, suppose one has $n$ pairs of independent observations, 
\begin{align}\label{eq:compH0}
    (X_i, Y_i) = (f(Z_i) + \vep_i, g(Z_i) + \delta_i), 
 \end{align}
for $1\leq i\leq n$, where  $\vep_{1},\ldots,\vep_n$ and $\delta_1,\ldots,\delta_n$ are i.i.d. error variables. Then the hypothesis of no treatment effect, that is, testing whether or not $f=g$ is equivalent to \eqref{eq:H0H1}. Note that when $f$ and $g$ are linear and the errors are normally distributed, this is the classical analysis of covariance problem \cite{hocking2013methods}. This problem is also well-studied the nonparametric setting \eqref{eq:compH0}, beginning with the works of \citet{hardle1990semiparametric} and \citet{king1991testing}, where tests based on the kernel density estimates of the regression functions were proposed. Many variations have been considered over the years, focusing primarily on testing equality of conditional moments  (see \cite{dette1998nonparametric,neumeyer2003nonparametric,chaudhuri1999sizer,hardle1990semiparametric,kulasekera1995comparison,varianceequality,fan1998} and the references therein). The framework considered in this paper extends beyond comparing the mean and variance functions, to the entire conditional distribution of the responses given the covariates. For instance, suppose 
 $X = f(Z,\vep)$ and $Y = g(Z,\delta)$, 
where $\vep,\delta$ are independent and identically distributed error variables. Then testing the hypothesis $f=g$ is equivalent to $\bH_0$ in \eqref{eq:H0H1} (see Section \ref{sec:regressionexample} for a specific example).  
\end{example}

\begin{example} (Validation of emulators models in Simulation-Based  Inference)
Likelihood Free Inference (LFI) and Simulation-Based  Inference (SBI) broadly refer to the collection of methods that use simulations to infer the posterior in situations where the likelihood function is intractable. This approach has found widespread success in many scientific domains (see \citet{cranmer2020frontier} and the references therein). In a typical setup, the simulator takes as input a vector of parameters $\theta$ generated from a prior $p(\theta)$, samples a set of latent variable $\eta\sim p(\eta|\theta)$ and generates a data vector $X\sim p(X|\eta,\theta)$. In this case, the likelihood function $p(X|\theta)$ is implicitly defined as 
\begin{align*}
    p(X|\theta) = \int p(X|\eta,\theta)p(\eta)\d\eta , 
\end{align*} 
which is intractable for most real-life simulators. To overcome this issue, several methods like Approximate Bayesian Computation \cite{marin2012approximate} and Synthetic Likelihood \cite{wood2010statistical} have been proposed. These methods make repeated calls to the simulator and use the simulated data to provide an estimate to the posterior distribution of $\theta$.  
However, the accuracy of the estimate depends on the number of simulation calls, leading to increased computational cost, especially for expensive simulators. Additionally, inference chain depends on the choice of hyperparameters and low dimensional summary statistics, which can potentially reduce the quality of inference. Moreover, directly using the simulator presents a lack of amortization, that is, the inference chain has to be restarted every time new data is available. To mitigate these issues, one often trains a faster surrogate or emulator $q(\cdot|\theta)$ for the computationally expensive simulator. Examples of emulator models include Gaussian mixture density networks \cite{lueckmann2019likelihood}, density ratio estimators \cite{hermans2020likelihood,durkan2020contrastive,dinev2018dynamic}, and, more recently, neural conditional density estimators, such as normalizing flows~\cite{trippe2018conditional,glockler2022variational,dirmeier2023simulation}, autoregressive models~\cite{hansen1994autoregressive,papamakarios2017masked,bruinsma2023autoregressive,papamakarios2019sequential},  conditional diffusion models~\cite{batzolis2021conditional,shi2022conditional}, and conditional deep generative networks \cite{zhou2022deep,liu2021wasserstein,song2023wasserstein,alfonso2023generative}, among others. To evaluate the validity of the emulator one needs to understand the extent to which it can imitate the simulator.  The basic diagnostic check towards this is to test the following hypothesis: 
\begin{align}\label{eq:SBItest}
    p(\cdot|\theta) = q(\cdot|\theta)\text{ almost everywhere }p(\theta) , 
\end{align}
This can be implemented as an instance of the conditional 2-sample hypothesis \eqref{eq:H0H1} as follows: Given independent samples $\{\theta_{1}, \theta_{2}, \ldots, \theta_{n}\}$ from the prior $p(\theta)$, generate a sample $X_{i}\sim p(\cdot|\theta_{i})$ from the simulator and a sample $Y_{i}\sim q(\cdot|\theta_{i})$ from the emulator, for $1 \leq i \leq n$. Then we can test the hypothesis \eqref{eq:SBItest} using the samples $(X_{1}, Y_{1}, {\theta}_{1}), (X_{2}, Y_{2}, {\theta}_{2}), \ldots, (X_{n}, Y_{n}, {\theta}_{n})$ (see, for example, \cite{dalmasso2020confidence,dalmasso2020validation, dalmasso2021likelihood,lemos2023sampling}).

In addition to emulating the simulator to provide posterior inferences about the parameter $\theta$, there have been approaches to directly emulate the true posterior from observed samples without additional calls to the simulator. This method offers an amortized way to estimate the posterior distribution. It allows for faster analysis in situations where time is limited or a large amount of data needs to be processed \cite{gonccalves2020training,dax2021real}. Moreover, it facilitates rapid application of diagnostic techniques that rely on obtaining posterior samples for numerous observations \cite{cook2006validation,talts2018validating}. Popular approaches towards approximating the posterior include mixture density network \cite{papamakarios2016fast,lueckmann2017flexible}, masked autoregressive flows \cite{greenberg2019automatic,deistler2022truncated}, normalising flows \cite{papamakarios2021normalizing,rodrigues2021hnpe,wiqvist2021sequential}, generative adversarial networks \cite{ramesh2022gatsbi}, and diffusion models \cite{gloeckler2024all,linhart2024diffusion}, among others. Once again the basic diagnostic check towards understanding the validity of approximate posterior is to test,
\begin{align}\label{eq:sbiposterior}
    p\left(\cdot|X\right) = q\left(\cdot|X\right)\text{ almost everywhere } p(X),
\end{align}
where $p(\cdot|X)$ is the true posterior distribution and $q(\cdot|X)$ is the proposed approximation. Given independent samples $(X_1, \theta_1), \ldots, (X_n, \theta_n)$ from the joint distribution $p(X, \theta)$ the test from \eqref{eq:sbiposterior} can be considered as an instance of \eqref{eq:H0H1} as follows: Generate a sample $\theta_i^\prime\sim q(\cdot|X_i)$ from the emulator for $1\leq i\leq n$. Now, we can test the hypothesis \eqref{eq:sbiposterior} in a similar manner as \eqref{eq:H0H1} using the samples $(\theta_1^\prime, \theta_1,X_1),\ldots(\theta_n^\prime, \theta_n, X_n)$ (see, for example \cite{zhao2021diagnostics,linhart2024c2st}). 
\end{example}

\subsection{Summary of Results}

In this paper we propose a measure of discrepancy between 2 conditional distributions, by adapting the well-known kernel Maximum Mean Discrepancy (MMD) \cite{gretton2012kernel} to the conditional setting. Towards this, we first embed the conditional distributions $P_{X|Z}$ and $P_{Y|Z} $ in a reproducing kernel Hilbert space (RKHS), through their conditional kernel mean embeddings \cite{park2020measure}. Then we quantify the discrepancy between $P_{X|Z}$ and $P_{Y|Z} $ in terms of the norm difference (in the RKHS) between the conditional mean embeddings averaged over the marginal distribution of $Z$. We refer to this measure as the {\it Expected Conditional Mean Embedding} (ECMMD). This measure characterizes the equality of two conditional distributions, that is, the $\EMMD$ measure is zero if and only if  $P_{X|Z} = P_{Y|Z} $ almost surely (Proposition \ref{ppn:EMMD0}). Moreover,  leveraging the reproducing property of the Hilbert space (the well-known `kernel-trick') we can express the $\EMMD$ in terms of the kernel dissimilarities averaged over the respective conditional distributions (Proposition \ref{ppn:EMMDexpr}). This allows us to estimate the $\EMMD$ efficiently based on the observed data. Specifically, we propose a nearest-neighbor graph based estimate of the $\EMMD$ that has the following properties: 

\begin{itemize}

\item The $\EMMD$ estimate has a simple, interpretable form, which does not require any estimation of density or distribution functions. Moreover, it encompasses both categorial and continuous responses and, consequently, can be easily applied to a range of data types. 

\item The estimate can be computed in near-linear time (with a fixed number of nearest neighbors) irrespective of the dimension of the data (see Remark \ref{remark:computation}). 

\item The estimate is consistent for the population $\EMMD$ measure under mild moment conditions on the kernel (Theorem \ref{thm:consistency}). 

\item Under $\bH_0$ as in \eqref{eq:H0H1} the estimate is asymptotically Gaussian (Theorem \ref{thm:CLT}) and its  variance has a simple tractable form (Proposition \ref{ppn:condexpH0}). In particular, we can consistently estimate the variance under $\bH_0$ (Proposition \ref{prop:SnestvarH0}), using which we obtain a universally consistent test for \eqref{eq:H0H1} that requires no nuisance parameter estimation (see Remark \ref{remark:distribution} for details). Moreover, both the consistency and the null distribution hold for any fixed kernel bandwidth, which is in contrast to several results on conditional inference, where density estimation with dimension dependent smoothing bandwidths are required for consistent estimation/testing.

\end{itemize} 

\noindent Next, motivated by the model-$X$ framework for conditional independence testing \cite{candes2018panning}, in Section \ref{sec:finitetest} we propose a resampling based test for the hypothesis \eqref{eq:H0H1} when there is sample access from one of the conditional distributions. This also provides a test for the conditional goodness-of-fit problem that controls Type I error in finite samples and is asymptotically consistent with only a finite number of resamples (see Proposition \ref{ppn:finitetest}). To further improve the stability of the procedure, we propose a de-randomized test which attains precise asymptotic level and is universally consistent (see Section \ref{sec:Derandom}). We provide a numerical comparison our finite-sample and derandomized tests with the test based on Kernel Conditional Stein Discrepancy (KCSD) \cite{jitkrittum2020testing} in Section \ref{sec:simulationsM}. Finally, we return to the examples mentioned before and illustrate how our proposed method performs both in simulations and real-data. Codes to reproduce all the experiments are available in \href{https://github.com/anirbanc96/ECMMD-CondTwoSamp}{https://github.com/anirbanc96/ECMMD-CondTwoSamp}. The following is a summary of our findings:

\begin{itemize}

\item In Section \ref{sec:classification} and Section \ref{sec:regression} we apply both the asymptotic and derandomized $\EMMD$ tests for assessing calibration in classification and regression models. We compare our method with the tests based on squared kernel calibration error (SKCE) proposed recently in \cite{widmann2019calibration,widmann2022calibration}.   In terms of power, the $\EMMD$ and SKCE are comparable for the classification model, but the $\EMMD$ is significantly more powerful for the regression model. Moreover, under the null hypothesis the SKCE statistic has a non-Gaussian asymptotic distribution (an infinite weighted sum of $\chi^2$ distributions) with no closed form expressions for the quantiles, hence the rejection threshold has to be chosen based on bootstrap/permutation methods. In contrast, the null distribution of the $\EMMD$ test statistic is Gaussian, hence, the rejection threshold can be obtained readily. Consequently, the implementation of the $\EMMD$ tests are significantly faster than the SKCE tests (see Appendix \ref{sec:computation} for a time comparison).

\item We also apply our method to test whether convolution neural network (CNN) models for classification are calibrated using the CIFAR 10 dataset. The $\EMMD$ test is able to successfully predict the insufficient calibration of the CNN model for image classification. The $\EMMD$ is also able to detect the significant change in calibration performance when applied to a recalibrated set of predicted probabilities (see Section \ref{sec:imagecalibration}).

\item For comparing regression curves, we apply the $\EMMD$ method on the wind energy dataset \cite{ding2019data,prakash2022gaussian,hwangbo2017production}, to test if the effect wind speed on wind power generation remain the same across different turbines. In contrast to existing methods for comparing regression curves, which  are primarily focused on comparing the mean functions, the $\EMMD$ based test can be used to detect arbitrary differences between the two models (see Section \ref{sec:regressionexample} for details).

\item In Section \ref{sec:posteriorapproximation} we use the $\EMMD$ to test the validity of posterior approximations on benchmark examples arising in simulation-based inference (SBI). Specifically, we apply the $\EMMD$ to test how well methods based on Mixture Density Networks (MDNs) and Neural Spline Flow (NSF) approximate the posterior distributions in the well-known {\tt Two Moons} and {\tt Simple Likelihood Complex Posterior (SLCP)} tasks.

\item In Section \ref{sec:densitysimulator} we apply the $\EMMD$ test to validate an emulator for the conditional density of the synthetic `redshift' associated with photometric galaxy images. Using a Gaussian convolutional mixture density network (ConvMDN) as an emulator for the synthetic `redshift' (as in \cite{zhao2021diagnostics,dey2022calibrated,leroy2021md,schmidt2020evaluation}), we illustrate the efficacy of the $\EMMD$ test in accurately detecting similarities and differences between the emulator and the simulator.

\end{itemize}

\subsection{Related Work}

Although, as discussed above, specific instance of the conditional two-sample problem are abundant, very few rigorous statistical methods have been proposed in the generality considered in this paper.  Very recently, \citet{chen2022paired} proposed a test for the hypothesis \eqref{eq:H0H1hypothesis} using a conditional version of the celebrated energy distance test \cite{szekely2003statistics,wang2015conditional,szekely2004testing,multivariate2004}. Their method uses kernel density estimation techniques, which, as mentioned before, require a careful choice of the bandwidth for the asymptotic properties to hold, that 
can be difficult to control beyond low dimensions. The $\EMMD$ measure has also appeared in the recent paper by \citet{huang2022evaluating}, where it is used as a metric for empirically quantifying the discrepancy between two conditional generative models. However, their estimation method and final aim are very different (see Remark \ref{remark:mmd} are further details). In the context of simulation-based inference, \citet{linhart2024c2st} recently proposed a test for comparing 2 conditional distributions locally at a given observation using the classifier 2-sample test \cite{friedman2003,classifier2022}.

In another variant of the conditional two-sample problem one assumes that the responses $X$ and $Y$ are conditioned on (potentially) different covariates and the response covariate pairs are independently generated from their respective joint distributions. In this setting, \citet{hu2023two} proposed a test using techniques from conformal prediction and a classifier based estimation of the marginal density ratio of the covariates. 
However, this requires sample-splitting and the performance of the test depends on the accuracy of the classifier. A refinement based on de-biased two-sample $U$-statistics has been proposed very recently by  \citet{chen2024biased}. Another test has been proposed by \citet{yan2022nonparametric} based on the integrated conditional energy distance. This requires estimation of the marginal density of the covariates based on kernels with smoothing bandwidths. Consequently, their test statistic has a non-negligible bias for dimensions greater than 4 and the asymptotic variance is also intractable, hence, the rejection threshold is chosen using a bootstrap resampling approach. 

A similarity measure between the mean functions of 2 conditional distributions 
sharing the same set of covariates, based on the minimum mean square error (mMSE) gap, also appears in \citet[Section 2.3.4]{zhang2022inference}. 

\section{Conditional Two Sample Test and Expected Conditional MMD}\label{sec:EMMD}

In this section we introduce the Expected Conditional MMD (ECMMD) measure for the conditional two-sample sample problem 
\eqref{eq:H0H1}. We begin by introducing the necessary formalism. To this end, suppose $\mathcal X$ and $\mathcal Z$ are Polish spaces, that is, complete and separable metric spaces, and $\mathscr B(\mathcal X)$ and $\mathscr B(\mathcal Z)$ be the $\sigma$-algebras generated by the open sets of $\mathcal X$ and $\mathcal Z$, respectively. Denote by $\mathcal P(\mathcal X \times \mathcal X \times \mathcal Z)$ the collection of all probability distributions on $(\mathcal X \times \mathcal X \times \mathcal Z, \mathcal B(\mathcal X) \times \mathcal B(\mathcal X) \times \mathcal B(\mathcal Z))$. Suppose $P_{XYZ}  \in \mathcal P(\mathcal X \times \mathcal X \times \mathcal Z)$ and $(X,Y,Z)\sim P_{XYZ}$ be a random variable with distribution $P_{XYZ}$. Denote by $P_X, P_Y, P_Z, P_{XY}, P_{YZ}, P_{XZ}$ the marginal distributions of $X, Y, Z, (X, Y), (Y, Z), (X, Z)$, respectively. Also, denote by $P_{X|Z}$, $P_{Y|Z}$, and $P_{XY|Z}$, the regular conditional distributions of $X|Z$, $Y|Z$, and $(X, Y)|Z$, respectively, which exist by \cite[Theorem 8.37]{klenke2013probability}.
Now, the conditional two-sample sample hypothesis \eqref{eq:H0H1} can be stated more formally as follows:
\begin{align}\label{eq:H0H1hypothesis}
    \bH_0: P_{Z}\left[P_{X|Z} = P_{Y|Z}\right] = 1 \text{ versus } \bH_1: P_{Z}\left[P_{X|Z} = P_{Y|Z}\right]<1 , 
\end{align} 
where $P_{Z}$ is the marginal distribution of $Z$. Our goal is to test the above hypothesis based on i.i.d. samples $\{(X_1,Y_1,Z_1),\ldots,(X_n,Y_n,Z_n)\}$ from the joint distribution $ P_{XYZ}$. 
 
\begin{remark} 
Note that when $Z$ is almost surely a constant, then \eqref{eq:H0H1hypothesis} reduces to the familiar unconditional two-sample testing problem between the distributions $P_X$ and $P_Y$ which are the marginal distributions of $X$ and $Y$, respectively.
\end{remark} 

To introduce the $\EMMD$ measure we first recall the fundamentals of the kernel MMD from \citet{gretton2012kernel}.

\subsection{Kernel Maximum Mean Discrepancy}  
\label{sec:preliminaries}

Denote by $\mathcal P(\mathcal X)$ the collection of all probability measures on $(\mathcal X, \mathscr B(\mathcal X))$. The {\it maximum mean discrepancy} (MMD) between two probability measures $P_X, P_Y \in \mathcal P(\mathcal X)$ is defined as,
\begin{align}\label{eq:FPQ}
    \MMD\left[\mathcal{F}, P_X, P_Y \right] = \sup_{f\in\mathcal{F}}\left\{ \mathbb{E}_{X \sim P_X } [ f( X ) ] -\mathbb{E}_{Y \sim P_Y } [f( Y ) ] \right\} ,
\end{align} 
where $\mathcal{F}$ is the unit ball of a reproducing kernel Hilbert space (RKHS) $\mathcal H$ defined on $\mathcal X$ \citep{aronszajn1950theory}. Since $\mathcal H$ is an RKHS, by the Riesz representation theorem \cite[Theorem II.4]{reedsimon} there exists a positive definite kernel $\sfK : \mathcal X \times \mathcal X \rightarrow \mathbb R$ such that the feature map $\psi_x\in\cH$ for all $x\in \cX$ satisfies $\sfK(x,\cdot) = \psi_{x}(\cdot)$ and $\sfK (x,y) = \langle \psi_x,\psi_y\rangle_{\mathcal{H}}$. The notion of feature map can now be extended to define the \textit{kernel mean embedding} $\mu_{P}$ for any distribution $P \in \cP(\cX)$ as, 
\begin{align}\label{eq:Hf}
    \langle f,\mu_{P}\rangle_{\mathcal{H}} = \mathbb{E}_{X\sim  P}[f(X)] , 
\end{align}  
for all $f\in \mathcal{H}$. By the canonical form of the feature map it follows that
\begin{align}\label{eq:muP}
    \mu_{P}(t) = \E_{X \sim  P}  [\sfK(X,t)] , 
\end{align} 
for all $t\in \cX$. 

Throughout we will assume the following, which ensures that a RKHS is rich enough to distinguish the two distributions from their corresponding mean embeddings. 

\begin{assumption}\label{assumption:K} {\em The kernel $\sfK : \mathcal X \times \mathcal X \rightarrow \mathbb R$ is positive definite and satisfies the following: 
\begin{itemize} 

\item[$(1)$]  $\mathbb{E}_{X \sim P_X}[\sfK (X,X)^{\frac{1}{2}}]<\infty$ and $\mathbb{E}_{Y \sim P_Y}[\sfK (Y,Y)^{\frac{1}{2}}]<\infty$.

\item[$(2)$] The kernel $\sfK$ is \textit{characteristic}, that is, the mean embedding $\mu: \mathcal{P}(\mathcal{X})\rightarrow\mathcal{H}$ is a one-to-one (injective) function.Moreover, the RKHS $\cH$ generated by $\sfK$ is separable.
\end{itemize} } 
\end{assumption}

Assumption \ref{assumption:K} ensures that $\mu_{P_X}, \mu_{P_Y}\in\mathcal{H}$ and the MMD can then be expressed as the distance in the RKHS between the corresponding kernel mean embeddings (see \cite[Lemma 4]{gretton2012kernel}): 
\begin{align}\label{eq:KPQ}
    \MMD^2\left[\mathcal{F}, P_X, P_Y\right] = \left\|\mu_{ P_X}-\mu_{ P_Y}\right\|_{\mathcal{H}}^2 , 
\end{align}
where $\|\cdot\|_{\mathcal{H}}$ is the norm corresponding to the inner product $\langle\cdot,\cdot\rangle_{\mathcal{H}}$. 
This, in particular, implies that $ \mathrm{MMD}^2\left[\mathcal{F}, P_X, P_Y\right]  = 0$ if and only if $ P_X= P_Y$.

\subsection{Conditional Kernel Mean Embedding and Expected Conditional MMD} 

Using the representation of the kernel mean embedding from \eqref{eq:muP}, one can define the {\it conditional kernel mean embeddings} of $X|Z$ and $Y|Z$ as follows (see \citet[Definition 3.1]{park2020measure}): 
    \begin{align}\label{eq:conditionalmean}
        \mu_{ P_{X|Z}}(t) = \E_{P_{X|Z}}\left[\sfK(X,t)\middle|Z\right] \text{ and } \mu_{ P_{Y|Z}}(t) = \E_{P_{Y|Z}}\left[\sfK(Y,t)\middle|Z\right]  , 
     \end{align} 
    for all $t\in \cX$, where $ P_{X|Z}$ and $P_{Y|Z}$ are the regular conditional distributions of $X|Z$ and $Y|Z$, respectively. By Assumption \ref{assumption:K}, the conditional mean embeddings are well-defined and for all $f \in \cH$, similar to \eqref{eq:Hf} (see \citet[Lemma 3.2]{park2020measure}), 
$$\langle f,\mu_{P_{X|Z}}\rangle_{\mathcal{H}} = \mathbb{E}_{ P_{X|Z}}[f(X)|Z] \text{ and }  \langle f, \mu_{P_{Y|Z}}\rangle_{\mathcal{H}} = \mathbb{E}_{P_{Y|Z}}[f(Y)|Z] ,$$ 
almost surely. Also, note that, unlike the (unconditional) kernel mean embeddings which are fixed elements in $\cH$, the conditional kernel mean embedding $\mu_{ P_{X|Z}}$ and $\mu_{ P_{Y|Z}}$ are random variables taking values in $\cH$. Consequently, $\|\mu_{ P_{X|Z}} - \mu_{ P_{Y|Z}}\|_{\cH}$ is random metric that measures the difference between the distributions $ P_{X|Z}$ and $P_{Y|Z}$, given a specific value of $Z$. Averaging over $Z$ leads to the following definition: 

\begin{definition} 
The {\it Expected Conditional $\MMD$} ($\EMMD$) between the conditional distributions $P_{X|Z}$ and $P_{Y|Z}$ is defined as: 
\begin{align}\label{eq:defEMMD}
    \EMMD^2\left[\cF,  P_{X|Z}, P_{Y|Z}\right] := \E_{Z\sim P_Z}\left[\left\|\mu_{ P_{X|Z}} - \mu_{ P_{Y|Z}}\right\|_{\cH}^2\right] ,
\end{align}  
where the expectation taken over the marginal distribution $P_Z$.  
\end{definition} 

The following result shows that $\EMMD$ characterizes the equality of two conditional distributions (see Appendix \ref{sec:conditionalpf} for the proof). 

\begin{proposition}\label{ppn:EMMD0}
    Suppose the kernel $\sfK$ satisfies Assumption \ref{assumption:K}. Then 
    \begin{align*}
        \EMMD^2\left[\cF,  P_{X|Z}, P_{Y|Z}\right] = 0\text{ if and only if }\P\left[ P_{X|Z} =  P_{Y|Z}\right] = 1.
    \end{align*}
\end{proposition}

The above result shows that the conditional 2-sample hypothesis \eqref{eq:H0H1hypothesis} can be equivalently reformulated as: 
$$\bH_0: \EMMD^2\left[\cF,  P_{X|Z}, P_{Y|Z}\right]= 0 \quad \text{ versus } \quad \bH_1: \EMMD^2\left[\cF,  P_{X|Z}, P_{Y|Z}\right] > 0.$$ 
Hence, to obtain a consistent test for the conditional 2-sample problem it suffices to consistently estimate $\EMMD^2\left[\cF,  P_{X|Z}, P_{Y|Z}\right]$ based on the data. While at first glance this might seem difficult, because the  MMD involves a supremum over functions in the unit ball of a RKHS, the well-known `kernel'-trick \cite{gretton2012kernel} allows us to express the $\EMMD$ in the following more tractable form: 

\begin{proposition}\label{ppn:EMMDexpr}
    Suppose the kernel $\sfK$ satisfies Assumption \ref{assumption:K}. Then,
    \begin{align}\label{eq:EMMDK}
        \EMMD^2\left[\cF,  P_{X|Z}, P_{Y|Z}\right] = \E\left[\sfK(X,X') + \sfK(Y,Y') - \sfK(X,Y') - \sfK(X',Y)\right] , 
    \end{align}
    where $(X,X',Y,Y',Z)$ has the following distribution: Generate $Z\sim P_Z$ and then sample $(X,Y)$ and $(X',Y')$ independently from the conditional distribution $ P_{XY|Z}$.
\end{proposition}

The proof of Proposition \ref{ppn:EMMDexpr} is given in Appendix \ref{sec:Kpf}. In the next section, we use the representation in \eqref{eq:EMMDK} to estimate $\EMMD$ based on nearest-neighbors.

\begin{remark}\label{remark:mmd}
As mentioned in the Introduction, the $\EMMD$ measure has been used in  \cite{huang2022evaluating} to quantify the discrepancy between two conditional generative models (see also earlier related work \cite{conditionalgenerative} where the conditional kernel mean embeddings was used to measure the pointwise difference between 2 conditional distributions). In \cite{huang2022evaluating} authors use multiple rounds of resamples from the generative models to estimate the $\EMMD$ and demonstrate numerically the performance of the estimate in measuring discrepancies between conditional distributions. However, the distributional properties of the estimate was not explored. On the other hand, we use the $\EMMD$ measure to develop a test for the conditional 2-sample problem with asymptotic guarantees. Our estimation is based on a simple nearest-neighbor technique that requires no resampling. We establish the consistency and asymptotic normality of our estimate and, as a consequence, obtain a universally consistent test for the conditional 2-sample problem (see Section \ref{sec:H0}). 
\end{remark}

\section{Estimating $\EMMD$ Using Nearest-Neighbors}

In this section we propose an estimate of the $\EMMD$ measure \eqref{eq:defEMMD} given i.i.d. samples $(X_1,Y_1,Z_1),\ldots,$ 
$(X_n,Y_n,Z_n)$ from the joint distribution $ P_{XYZ}$. Towards this, suppose $\sfK$ is a kernel satisfying Assumption \ref{assumption:K} and for $\bw = (x,y), \bw' = (x',y')\in \cX\times\cX$ define,
\begin{align}\label{eq:defH}
    \sfH(\bw,\bw') = \sfK(x,x') + \sfK(y,y') - \sfK(x,y') - \sfK(x',y) . 
\end{align}
Then using Proposition \ref{ppn:EMMDexpr} the $\EMMD$ can be expressed as: 
\begin{align}\label{eq:HWW}
    \EMMD^2\left[\cF,  P_{X|Z}, P_{Y|Z}\right] = \E\left[\sfH(\bW,\bW')\right] = \E\left[\E\left[\sfH(\bW,\bW')\middle|Z\right]\right] ,  
\end{align}
where $Z\sim  P_Z$ and $\bW = (X,Y)$ and $\bW' = (X',Y')$ are generated independently from $ P_{XY|Z}$. To estimate the RHS of \eqref{eq:HWW}, we first fix $Z=Z_u$, for $1 \leq u \leq n$, and consider the condition expectation $\E\left[\sfH(\bW,\bW')\middle|Z= Z_u\right]$.  To estimate this conditional expectation the idea is to average the discrepancies of the centered kernel $\sfH$ over indices which are `close'  to $Z_u$. One natural way to capture such proximity is through neighbors-neighbor graphs. This motivates the following construction: 

\begin{itemize}

\item Fix $K=K_n \geq 1$ and construct the directed $K$-nearest neighbor ($K$-NN) graph $G(\sZ_n)$ of the data points $\sZ_n := \{Z_1,Z_2,\ldots,Z_n\}$. Unless otherwise specified, we will use the notation $(u, v)$ to denote the directed edge $Z_u\ra Z_v$, for $Z_u, Z_v \in \sZ_n$. 
Also, we will denote the edge set of $G(\sZ_n)$ by $$E(G(\sZ_n)) = \{(u, v) \in  [n]^2: \text{such that } Z_u\ra Z_v \text{ is a directed  edge in } G(\sZ_n) \},$$
where $[n]  := \{1, 2, \ldots, n\}$. 

\item Then the $K$-NN based estimate of \eqref{eq:HWW} is given by: 
\begin{align}\label{eq:estEMMD}
    \EMMD^2[\sfK, \bcW_n, \sZ_n] = \frac{1}{n}\sum_{u=1}^{n}\frac{1}{K}\sum_{v \in N_{G(\sZ_n)}(u)}\sfH(\bW_{u},\bW_{v}) , 
\end{align}
where $\bcW_n:=\left\{\bW_i = (X_i,Y_i):1\leq i\leq n\right\}$ and $N_{G(\sZ_n)}(u) = \{v: 
(u,v) \in E(G(\sZ_n))\}$. 
\end{itemize}

\begin{remark}\label{remark:computation}
Note that the estimate \eqref{eq:estEMMD} can be computed easily in $O(K n \log n)$ time, in any dimensions. This is because $K$-NN graph can be computed in $O(K n \log n)$ time (see, for example, \cite{friedman1977algorithm}) and, given the graph, the sum in \eqref{eq:estEMMD} can be computed in $O(K n)$ time, since the $K$-NN graph has $O(Kn)$ edges. 
\end{remark} 

To prove the consistency of \eqref{eq:estEMMD} we assume the following on the conditioning variable $Z$.

\begin{assumption}\label{assumption:Kn}
The random variable $Z$ takes values in $\cZ = \R^d$, for some $d \in \N$, and $\|Z - Z'\|_2$ has a continuous distribution, where $Z, Z'$ are i.i.d. samples from $P_Z$.
\end{assumption}
 
This assumption ensures that the $K$-NN graph constructed using $\|\cdot\|_2$ norm is well-defined and the degrees of its vertices scales proportional to $K$ (see \cite{jaffe2020randomized,deb2020measuring}). One can easily relax this to include any finite dimensional inner product space over $\R$, by isometrically isometric embedding such spaces into the Euclidean space. It is worth noting here that we do not require the space $\cX$ (on which the random variables $X, Y$ are defined) to be Euclidean or even finite-dimensional. (For example, while calibration testing in classification models the variables $X, Y$ are categorical.) The following theorem establishes the consistency of 
$\EMMD[\sfK,\bcW_n,\sZ_n]$. 

\begin{theorem}\label{thm:consistency}
    Suppose Assumption \ref{assumption:K} and Assumption \ref{assumption:Kn} hold. Moreover, suppose the kernel $\sfK$ satisfies $\int \sfK(x,x)^{2+\delta}\mathrm{d} P_X(x)<\infty$, and  $\int \sfK(x,x)^{2+\delta}\mathrm{d} P_Y(x)<\infty$, for some $\delta>0$. Then with $K = o(n/\log n)$,
    \begin{align*}
        \EMMD^2[\sfK,\bcW_n,\sZ_n]\pto \EMMD^2\left[\cF,  P_{X|Z}, P_{Y|Z}\right] , 
    \end{align*}
    where $\EMMD^2[\sfK,\bcW_n,\sZ_n]$ and $\EMMD^2\left[\cF,  P_{X|Z}, P_{Y|Z}\right]$ are defined in \eqref{eq:estEMMD} and \eqref{eq:defEMMD},  respectively.
\end{theorem}

The proof of Theorem \ref{thm:consistency} is given in Appendix \ref{sec:proofofconsistency}. The proof involves the following steps: 
\begin{itemize}

\item First we show that the expectation of $\EMMD^2[\sfK,\bcW_n,\sZ_n]$ converges to the population $\EMMD^2$ (Lemma \ref{lemma:consistency2}). The main idea here is that averaging the centered kernel $\sfH$ (recall \eqref{eq:defH}) around the nearest neighbors of a point $Z=z$, provides an asymptotically unbiased estimate of the $\MMD^2$ distance between conditional mean embeddings $P_{X|Z=z}$ and $P_{Y|Z=z}$. 

\item Next, using the Efron-Stein inequality \cite{efron1981variance}, we show that the variance of \eqref{eq:estEMMD} converges to zero (Lemma \ref{lemma:consistency1}). This leverages the local dependence of the $K$-NN graph, which is controlled by the condition $K = o(n/\log n)$. 

\end{itemize}

\section{Asymptotic Test Based on $\EMMD$}
\label{sec:H0}

To use the estimate \eqref{eq:estEMMD} for testing the conditional two-sample hypothesis, we need to derive its asymptotic distribution under $\bH_{0}: P_{X|Z} = P_{Y|Z}$. To this end, consider the scaled version of the $\emmdstat$ statistic,
\begin{align}\label{eq:defetan}
    \eta_{n}:=\sqrt{nK}\emmdstatsq = \frac{1}{\sqrt{nK}}\sum_{u=1}^{n}\sum_{v \in N_{G(\sZ_n)}(u)} \sfH(\bW_{u},\bW_{v}).
\end{align} 
Also, denote by $\cF(\sZ_n)$ the $\sigma$-algebra generated by $Z_1,Z_2,\ldots,Z_n$.  To begin with, note that 
    \begin{align}\label{eq:EH0}
        \E_{\bH_{0}}\left[\eta_{n}|\cF(\sZ_n)\right] = \frac{1}{\sqrt{nK}}\sum_{1 \leq u,v \leq n} \E_{\bH_{0}}\left[\sfH(\bW_{u},\bW_{v})|\cF(\sZ_n)\right]\one\left\{(u,v)\in E(G(\sZ_n))\right\} = 0 , 
    \end{align}
    since, for $1 \leq u\neq v \leq n$, recalling \eqref{eq:defH}, 
    \begin{align}
        \E_{\bH_{0}}\left[\sfH(\bW_{u},\bW_{v})|\cF_{n}(\sZ_n)\right]
        =\E_{\bH_{0}}\left[\sfH(\bW_{u},\bW_{v})|Z_{u},Z_{v}\right]
        &
        = 0 , \label{eq:condexpij0}
    \end{align}
    almost surely. 
Next, we compute the conditional variance of $\eta_{n}$ under $\bH_{0}$.   

\begin{proposition}\label{ppn:condexpH0}
Denote by $\sigma^2_n := \Var_{\bH_0}\left[\eta_n\middle|\cF(\sZ_n)\right]$, conditional variance of $\eta_n$ given $\cF(\sZ_n)$  under  $\bH_{0}$. Then
\begin{align}\label{eq:condeta2exp}
       \sigma^2_n = \frac{1}{nK}\sum_{ 1 \leq u, v \leq n} f(Z_u,Z_v) \left( \one\left\{(u,v)\in E(G(\sZ_n))\right\} + \one\left\{(u,v),(v,u)\in E(G(\sZ_n))\right\} \right) , 
\end{align} 
where $f(Z,Z'):= \E_{\bH_0}\left[\sfH^2((X,Y),(X',Y')) |Z,Z'\right]$ with $(X,Y,Z),(X',Y',Z')$ i.i.d. samples from $ P_{XYZ}$.  
\end{proposition}

The proof of Proposition \ref{ppn:condexpH0} is given in Appendix \ref{sec:varpf}. The proof relies on the observation that for $1 \leq u \ne u'  \ne v \ne v' \leq n$, 
$$\mathrm{Cov}_{\bH_0}[ \sfH(\bW_{u},\bW_{v}),  \sfH(\bW_{u'},\bW_{v'}) | \cF(\sZ_n) ] = 0, $$ unless $\{u, v\} = \{ u', v'\}$. 
In other words, under the null hypothesis, the summands in \eqref{eq:defetan} are pairwise conditionally uncorrelated, hence, the only terms that contribute to the conditional variance are when (1) $u=u'$ and $v = v'$ which corresponds to the first term in \eqref{eq:condeta2exp} or (2) $u=v'$ and $v = u'$ which corresponds to the second term in \eqref{eq:condeta2exp}.  Consequently, $\sigma^2_n$ has a simple form as the average of the conditional expectation of $\sfH^2$ over the edges of the $K$-NN graph $G(\sZ_n)$ (counted twice when edges are present in both directions). 
This representation is particularly convenient because just by replacing the (unknown) conditional expectation $f(Z_u, Z_v)$ with kernel values $\sfH^2(\bW_{u},\bW_{v})$, we get following natural estimate of $\sigma_{n}^2$: 
\begin{align}\label{eq:defSn}
    \hat{\sigma}_{n}^2 = \frac{1}{nK}\sum_{u=1}^{n}\sum_{v=1}^{n}\sfH^{2}\left(\bW_{u},\bW_{v}\right)\left(\one\left\{(u,v)\in E(G(\sZ_n))\right\}+\one\left\{(u,v),(v,u)\in E(G(\sZ_n))\right\}\right). 
\end{align} 

In the following we establish the consistency of $\hat{\sigma}_{n}^2$ (see Appendix \ref{sec:estimatepf} for the proof) and subsequently use $\hat{\sigma}_n$ to define a studentized version of $\eta_n$ for constructing the asymptotic test. 

\begin{proposition}\label{prop:SnestvarH0}
    Suppose Assumption \ref{assumption:K} and Assumption \ref{assumption:Kn} holds. Furthermore, for $(X, Y, Z)\sim P_{XYZ}$ assume that $P_{XY}(X\neq Y)>0$ and the kernel $\sfK$ satisfies $\int\sfK(x,x)^{4+\delta}\rmd P_{X}(x)<\infty$ and $\int\sfK(y,y)^{4+\delta}\rmd P_{Y}(y)<\infty$, for some $\delta>0$. Then under $\bH_0$ with $K = o(n/\log n)$, as $n \rightarrow \infty$
    \begin{align}\label{eq:Snapproxvar}
        \left|\frac{\hat{\sigma}_n^2}{\sigma_n^2}-1\right| = o_{p}\left(n^{-\frac{\delta}{32+4\delta}}\right). 
    \end{align}
    
    where $\sigma_n^2$ and $\hat{\sigma}_n^2$ are defined in \eqref{eq:condeta2exp} and \eqref{eq:defSn}, respectively. 
\end{proposition}

We are now ready to state the result about the asymptotic null distribution of $\eta_{n}$. Specifically, in the following theorem we show that $\eta_n$ scaled by $\hat{\sigma}_n$ converges to $N(0, 1)$, under $\bH_0$, in the Kolmogorov distance.

\begin{theorem}\label{thm:CLT}
 Suppose Assumption \ref{assumption:K} and Assumption \ref{assumption:Kn} holds. Furthermore, for $(X, Y, Z)\sim P_{XYZ}$ assume that $P_{XY}(X\neq Y)>0$ and the kernel $\sfK$ satisfies $\int\sfK(x,x)^{4+\delta}\rmd P_{X}(x)<\infty$ and $\int\sfK(y,y)^{4+\delta}\rmd P_{Y}(y)<\infty$, for some $\delta>0$. Then under $\bH_0$, with $K = o(n^{1/44})$, as $n \rightarrow \infty$, 
    \begin{align}\label{eq:distributionH0}
        \sup_{z\in\R}\left|\P_{\bH_0}\left[\frac{\eta_{n}}{\hat{\sigma}_n}\leq z\right] - \Phi(z)\right|\ra 0 , 
    \end{align}
    where $\Phi(\cdot)$ is the CDF of the standard Gaussian distribution. 
\end{theorem}

The proof of Theorem \ref{thm:CLT} is given in Appendix \ref{sec:CLTpf}. The proof entails showing that 
   \begin{align}\label{eq:distributionH0pf}
        \sup_{z\in\R} \left|\P_{\bH_0}\left[\frac{\eta_{n}}{\sqrt{\Var_{\bH_0}[\eta_n|\cF(\sZ_n)]}}\leq z\right] - \Phi(z)\right|\ra 0 . 
    \end{align}
The result in \eqref{eq:distributionH0} follows by combining the above with Proposition \ref{prop:SnestvarH0}. To show \eqref{eq:distributionH0pf} we use the Stein's method based on dependency graphs \cite{chen2004normal}, which allows us to control the 
Kolmogorov distance between $\eta_{n}/\hat{\sigma}_n$ and $N(0, 1)$ in terms 
of maximum degree of the $K$-NN graph. 

To choose the rejection threshold for $\hat \eta_n$ based on Theorem \ref{thm:CLT}, fix $\alpha \in (0, 1)$ and consider the test function:  

\begin{align}\label{eq:defphin}
    \phi_{n} := \bm{1}\left\{\left|\eta_n/\hat{\sigma}_n\right|>z_{\alpha/2}\right\} , 
\end{align}
where $\eta_n$ and $\hat{\sigma}_n$ are as defined in \eqref{eq:defetan} and \eqref{eq:defSn}, respectively. Theorem \ref{thm:CLT} directly implies that $\phi_{n}$ is asymptotically level $\alpha$ for $\bH_0$ as in \eqref{eq:H0H1hypothesis}, that is, 
$$\lim_{n\ra\infty} \P_{\bH_0}\left[\phi_{n} = 1\right] = \alpha.$$ 
Since $\EMMD$ characterizes the equality of the conditional distributions and \eqref{eq:estEMMD} is a consistent estimate of $\EMMD$ (recall Theorem \ref{thm:consistency}), the test $\phi_n$ is consistent for fixed alternatives. This is summarized in the following result (see Appendix \ref{sec:consistencyH0pf} for the proof). 
  
\begin{corollary}\label{cor:asymptest} 
Suppose the assumptions of Theorem \ref{thm:CLT} hold. Then for any $P_{XYZ} \in \bH_1$, $$\lim_{n\ra\infty} \P_{\bH_1}\left[\phi_n = 1\right] = 1.$$ 
\end{corollary}

\begin{remark}\label{remark:distribution}
The ECMMD test has several interesting features which are in contrast to other kernel and nearest-neighbor based nonparametric tests: 
\begin{itemize} 
\item The limiting distribution of the statistic $\eta_{n}$ is normal under $\bH_0$ and the asymptotic variance has a simple form, which can be easily estimated from the data. Consequently, the rescaled test statistic $\eta_{n}/\hat{\sigma}_n$ converges to $N(0, 1)$ under $\bH_0$ and we can obtain the rejection threshold as in \eqref{eq:defphin}, without having to estimate any nuisance parameter. In other words, $\eta_{n}/\hat{\sigma}_n$ is {\it asymptotically distribution-free}, that is, its limiting distribution under  $\bH_0$ does not depend on the unknown distribution of the data. In contrast, the familiar kernel MMD statistic for the (unconditional) two-sample problem has a non-Gaussian (specifically, an infinite mixture of chi-squares) limiting distribution under the null \cite{gretton2012kernel}. Closed form estimates for the quantiles of such distributions are not available, in general,  which necessitates the use permutation/bootstrap resampling techniques or conservative approximations based on concentration inequalities for determining the rejection thresholds \cite{schrab2023mmd,gretton2009fast,chatterjee2023boosting}. We circumvent this issue through the use of nearest neighbor graphs (on the space of the covariate variable), which mitigates the dependence among the summands in \eqref{eq:defetan} in such a way that, although the kernel $\sfH$ is degenerate under $\bH_0$, the asymptotic distribution of $\eta_n$ is normal. 
\item Another important property of the $\EMMD$ test is that the statistic $\eta_n$ is unbiased (has mean zero) under $\bH_0$ (recall \eqref{eq:EH0}). This is different from the recent work on conditional independence testing based on nearest neighbors \cite{azadkia2021simple,shi2021azadkia,huang2022kernel}, where the test statistics have a non-zero bias under the null, and, hence, cannot be directly used for inference, without additional de-biasing. The bias issue also appears in density-estimation based methods, both for the conditional independence \cite{wang2015conditional} and the conditional 2-sample problem \cite{yan2022nonparametric}, which is usually handled by choosing dimension dependent smoothing bandwidths. We, on the other hand, are able to cancel the bias because of the paired nature of the samples (recall that for each $Z_i$ we have paired samples $(X_i, Y_i)$) and through the use of nearest-neighbors. 
\end{itemize} 
\end{remark}

\section{A Resampling Based Conditional Goodness-of-Fit Test Using $\EMMD$}\label{sec:finitetest}

In this section, drawing parallel from the model-$X$ framework for conditional independence testing \cite{candes2018panning}, we design a resampling based test the hypothesis  \eqref{eq:H0H1hypothesis} that controls Type I error in finite samples when it is possible to efficiently sample from one of the conditional distributions $P_{X|Z}$ or $P_{Y|Z}$. This principle applies more broadly to the conditional goodness-of-fit problem which entails testing 
the hypothesis in \eqref{eq:H0H1hypothesis}, when one of the conditional distributions is specified. Specifically, suppose we are given on i.i.d. samples $(Y_1, Z_1),\ldots,(Y_n, Z_n)$ from the joint distribution $ P_{YZ}$ and we wish to test the hypothesis 
\begin{align}\label{eq:H0H1finite}
\bH_0: P_{Z}\left[P_{Y|Z}=P_{X|Z}\right] = 1 \quad \text{ versus } \quad \bH_1: P_{Z}\left[P_{Y|Z} \ne P_{X|Z}\right]>0,
\end{align} 
where $P_{X|Z}$ is a specified conditional distribution.

\begin{algorithm}[!h]
\raggedright \hspace*{\algorithmicindent} \caption{ A finite sample conditional goodness-of-fit test }

\begin{enumerate}[label=\textrm{(\arabic*)}]
    \item For each $1\leq u \leq n$, generate i.i.d. samples $\left(X_{u}^{(1)},\ldots,X_{u}^{(M+1)}\right)$ from the distribution $ P_{X|Z=Z_u}$.

    \item Denote by $\bm W_{u}^{(m)}= (X_u^{(M+1)}, X_u^{(m)})$, for $1 \leq m \leq M$, and $\bm W_{u}^{(M+1)}= (X_u^{(M+1)}, Y_u)$, for $1\leq u \leq n$. Define 
\begin{align}\label{eq:defetanm}
    \eta_{n}^{(m)}=\frac{1}{\sqrt{nK}}\sum_{u=1}^{n}\sum_{v \in N_{G(\sZ_n)}(u)}\sfH(\bW_{u}^{(m)},\bW_{v}^{(m)}), 
    \end{align} 
    for $1 \leq m \leq M+1$, where $\sfH$ is defined in \eqref{eq:defH} and $G(\sZ_n)$ is the $K$-NN graph of the data points $\sZ_n = \{Z_1,\ldots,Z_n\}$.
    
    \item Report the $p$-value 
      \begin{align}\label{eq:Mfinite}
        p_{M}:=\frac{1}{M+1}\left[1+\sum_{m=1}^M \mathbf{1}\bigg\{\left|\eta_{n}^{(m)}\right| \geq \left|\eta_{n}^{(M+1) }\right|\bigg\}\right].
    \end{align} 
   
\end{enumerate}

\label{algorithm:finitesample}
\end{algorithm}

In Algorithm \ref{algorithm:finitesample} we develop a resampling based test for \eqref{eq:H0H1finite}. Note that \eqref{eq:Mfinite} in Algorithm \ref{algorithm:finitesample} is a valid $p$-value because the collection $\{\eta_{n}^{(1)},\eta_{n}^{(2)},\ldots,\eta_{n}^{(M+1)}\}$ is exchangeable conditional on $\sZ_n$ when $ P_{X|Z} =  P_{Y|Z}$ almost surely. 

Consequently, the resulting test controls Type I error in finite samples. 
This is formalized in the following result which also establishes the asymptotic consistency of test with a finite number of resamples (see Appendix \ref{sec:H0hypothesispf} for the proof): 

\begin{proposition}\label{ppn:finitetest}
 Fix $\alpha \in (0, 1)$  and consider the test function $\tilde \phi_{n, M} = \bm{1}\left\{p_{M}\leq \alpha\right\}$, with $p_M$ as in \eqref{eq:Mfinite}. Then the following hold: 
\begin{itemize}
\item[$(1)$] $\P_{\bH_0}[\tilde \phi_{n, M} = 1]\leq \alpha$. 

\item[$(2)$] Suppose the kernel $\sfK$ satisfies $\int \sfK(x,x)^{2+\delta}\mathrm{d} P_X(x)<\infty$, and  $\int \sfK(x,x)^{2+\delta}\mathrm{d} P_Y(x)<\infty$, for some $\delta>0$. Then with $K = o(n/\log n)$ and any $P_{XYZ} \in \bH_1$ (that is, $P_{Y|Z}\ne P_{X|Z}$), $\lim_{n \rightarrow \infty} \P_{\bH_1}[\tilde \phi_{n, M} = 1] = 1$, whenever $M>\frac{1}{\alpha}-1$.
\end{itemize}
\end{proposition} 

In the following we summarize the current state-of-the-art in nonparametric conditional goodness-of-fit testing (Section \ref{sec:survey}), discuss how the resampling based test in Algorithm \ref{algorithm:finitesample} fits into this literature and its relevance in modern machine learning problems (Section \ref{sec:resampling}), and propose a de-randomized version of the resampling test and study its asymptotic properties (Section \ref{sec:Derandom}).

\subsection{Prior Work on Conditional Goodness-of-Fit Testing}
\label{sec:survey}

The conditional goodness-of-fit problem has its roots in the econometrics literature, beginning with the work of Andrews \cite{andrews1997conditional}, which extended the classic Kolmogorov test to the conditional case. Thereafter, other methods for the conditional goodness-of-fit problem have been proposed, however, these tests either involve density estimation \cite{zheng2000consistent}, which require decaying smoothing bandwidths that can be difficult to control, or are designed for specific families of conditional models, such as structural equation models \cite{moreira2003conditional} or generalized linear models \cite{stute2002model}. Recently, \citet{jitkrittum2020testing} proposed a general nonparametric test for the conditional goodness-of-fit problem that does not require any density estimation or knowledge of the normalizing constant of the conditioning distribution. Specifically, the method extends the well-known Kernel Stein Discrepancy (KSD) \cite{liu2016kernelized} to the conditional setting, referred to as the Kernel Conditional Stein Discrepancy (KCSD), and only requires knowledge of the score-function of the conditional distribution $P_{X|Z}$. While the KCSD method circumvents several of the limitations of previous density estimation based methods, it still remains inapplicable in situations where $P_{X|Z}$ is implicitly defined and one only has access to samples from the conditional distribution. In contrast, the resampling based method described above can be readily applied in such situations. We elaborate on this in the next section.

\subsection{When is Resampling Useful?} 
\label{sec:resampling}

Any goodness-of-fit problem can be transformed into a 2-sample problem by repeatedly sampling from the known null distribution. Therefore, it is no surprise that the $\EMMD$ statistic, which is a measure of difference between two conditional distributions, can be calibrated for the conditional goodness-of-fit problem through resampling. Friedman \cite[Section 4]{friedman2003} summarizes this principle succinctly as follows: `Using this additional
information has the potential for increased power at the expense of having to generate many Monte Carlo samples, instead of just one.' We are not advocating that one should always resort to resampling for conditional goodness-of-fit testing, but there are important cases where it might be reasonable to apply the finite sample test in Proposition \ref{ppn:finitetest}, over  existing conditional goodness-of-fit methods. This is indeed the case for two of the examples considered in this paper: (1) calibration testing and (2) validation of emulator models in SBI, as explained below: 

\begin{itemize}

\item[(1)] Calibration testing for classification is an example of a conditional goodness-of-fit problem that can be implemented as a conditional 2-sample problem by sampling from the known distribution $P_{X|f(Z)}$  (recall \eqref{eq:classification}). We illustrate the numerical performance of  this method in Section \ref{sec:classification}. We also compare our method with the test in \cite{widmann2019calibration}, which is based on directly estimating kernel versions of the expected calibration error (ECE) (see Appendix \ref{sec:calibrationfigures} for the definition of ECE). Beyond the classification setting, that is, for continuous response (see Section \ref{sec:regression} for more details), kernel based ECE measures for testing calibration require taking expectations against a specified (usually intractable) probabilistic model and involve kernels on the space of distributions \cite{widmann2022calibration}. Although these issues have been mitigated in the recent work based on KCSD \cite{glaser2023kernelscore}, it still requires the score function of the generative model to be available in closed form. On the other hand, our formulation of the calibration problem as a conditional 2-sample hypothesis only requires sample access from the conditional distribution $P_{X|f(Z)}$, which is usually readily available given the trained model $f(Z)$. Consequently, our proposed method can be applied easily to categorical and continuous responses. In fact, from the numerical experiments in Section \ref{sec:imagecalibration} we will see that even the asymptotic test, which only requires sampling a single $X_i$ for each observation $(Y_i, f(Z_i))$, is powerful in a variety of examples.

\item[(2)] Validation tests of emulator models or approximate posteriors is another instance of a conditional goodness-of-fit problem that can be naturally operationalized as a conditional 2-sample problem. Here, $P_{X|Z}$ in \eqref{eq:H0H1finite} corresponds to $p(\cdot| \theta)$ (the density/score function of simulator) or $p(\cdot|X)$ (the true posterior distribution), and 
$P_{Y|Z}$ corresponds to either the emulator distribution $q(\cdot| \theta)$ as in \eqref{eq:SBItest} or the approximate posterior $q(\cdot|X)$ as in \eqref{eq:sbiposterior}, respectively. In either case, there is usually no tractable form $P_{X|Z}$. Moreover, it is usually easier to sample from $q(\cdot| \theta)$ or $q(\cdot|X)$ than to deal with their actual functional forms. In this situation, following Friedman's aphorism, we can validate the performance of the emulator (test the hypothesis \eqref{eq:SBItest}) or the approximate posterior (test the hypothesis \eqref{eq:sbiposterior}) using the finite-sample test \eqref{eq:Mfinite}, by repeatedly sampling from the emulator, or the asymptotic test \eqref{eq:defphin}, with a single set of samples (recall the discussion after \eqref{eq:SBItest} and \eqref{eq:sbiposterior}). We discuss examples in Section \ref{sec:sbiexperiments}.
\end{itemize} 

\subsection{A Derandomized  Asymptotic Test Based on ECMMD}\label{sec:Derandom}

One issue with the test in Algorithm \ref{algorithm:finitesample} (as well as the test from Section \ref{sec:H0}) is that it is a randomized procedure, that is, different runs of the algorithm produce different $p$-values, which can lead to inconsistent conclusions. In this section we propose a de-randomization method that, instead of calculating the $\EMMD$  statistic for each run of the algorithm as in \eqref{eq:defetanm}, computes a single test statistic by averaging the kernel discrepancies over $M = M_n$ resamples for each given $Z= Z_u$, for $1 \leq u \leq n$, as described below:

\begin{itemize}
    \item[(1)] For each $1\leq u\leq n$ generate i.i.d. samples $(X_{u}^{(1)},\ldots, X_{u}^{(M_n)})$ from the distribution $P_{X|Z = Z_u}$ independent of $Y_u$ given $Z_u$.
    \item[(2)] Denoting $\bW_{u}^{(m)} = (X_u^{(m)}, Y_{u})$, for $1\leq m\leq M_n$ and $1\leq u\leq n$ define the de-randomized test statistic as follows:
        \begin{align}\label{eq:defDn}
            D_n = \frac{1}{nK}\sum_{u=1}^{n}\sum_{v\in N_{G(\sZ_n)}(u)}\frac{1}{M_n}\sum_{m=1}^{M_n}\sfH(\bW_{u}^{(m)},\bW_{v}^{(m)}) . 
        \end{align}
\end{itemize}

Note that for $M_n=1$ the statistic $D_n$ equals the estimate $\emmdstatsq$ defined in \eqref{eq:estEMMD}. The averaging step in \eqref{eq:defDn} is meant to mitigate the sensitivity to the resampling uncertainty. To choose the rejection threshold for $D_n$, we now investigate its asymptotic properties. First, we show that $D_n$ consistently estimates the population $\EMMD$, as $n \rightarrow \infty$, irrespective of the choice of $M_n$ (see Appendix \ref{sec:proofofDnconsistency} for the proof): 

\begin{theorem}\label{thm:consistencyDn}
    Suppose Assumption \ref{assumption:K} and Assumption \ref{assumption:Kn} hold. Moreover, suppose the kernel $\sfK$ satisfies $\int \sfK(x,x)^{2+\delta}\mathrm{d} P_X(x)<\infty$, and  $\int \sfK(x,x)^{2+\delta}\mathrm{d} P_Y(x)<\infty$, for some $\delta>0$. Then with $K = o(n/\log n)$,
    \begin{align*}
        D_n\pto \EMMD^2\left[\cF,P_{X|Z},P_{Y|Z}\right],
    \end{align*}
    where $\EMMD^2\left[\cF,P_{X|Z},P_{Y|Z}\right]$ is defined in  \eqref{eq:defEMMD}.
\end{theorem}

Next, we show that a studentized version of $D_n$ converges to $N(0, 1)$ under the null hypothesis. Towards this, define the studentization factor: 
\begin{align}\label{eq:estimateM}
    \hat{\tau}_n^2 := \frac{1}{nK}\sum_{\substack{1\leq u\leq n\\v\in N_{G(\sZ_n)}(u)}}\left(\frac{1}{M_n}\sum_{m=1}^{M_n}\sfH\left(\bW_{u}^{(m)},\bW_{v}^{(m)}\right)\right)^2 \left(\one\{\cE_{u, v}\} +  \one\{\cE_{u, v}^+\} \right)  , 
\end{align} 
where $\cE_{u, v} := \{(u,v)\in E(G(\sZ_n)) \}$ and  $\cE_{u, v}^+ := \{(u,v),(v,u)\in E(G(\sZ_n)) \}$. 
 
\begin{theorem}\label{thm:DnCLT}
    Suppose Assumption \ref{assumption:K} and Assumption \ref{assumption:Kn} hold. Also, assume that the kernel $\sfK$ satisfies $\int\sfK(x,x)^{4+\delta}\rmd P_{X}(x)<\infty$ and $\int\sfK(y,y)^{4+\delta}\rmd P_{Y}(y)<\infty$, for some $\delta>0$. Then, for $K = o(n^{1/44})$ and $M_n\ra\infty$, as $n\ra\infty$, the following holds under $\bH_0$, 
    \begin{align}\label{eq:distributionH0Dn}
        \sup_{z\in\R}\left|\P_{\bH_0}\left[\frac{\sqrt{nK}D_{n}}{\hat{\tau}_n}\leq z\right] - \Phi(z)\right|\ra 0 , 
    \end{align}
    where $\Phi(\cdot)$ is the CDF of standard Gaussian distribution. 
\end{theorem}

The proof of Theorem \ref{thm:DnCLT} is given in Appendix \ref{sec:proofofDnClT}. Theorem \ref{thm:DnCLT} shows that the test function
\begin{align}\label{eq:defphinM}
    \tilde{\phi}_{n}^* := \bm{1}\left\{\left| \sqrt{nK}D_{n} \right|>z_{\alpha/2} \hat{\tau}_n \right\} , 
\end{align}
is asymptotically level $\alpha$, that is, $\lim_{n \rightarrow \infty} \P_{\bH_0}[\tilde{\phi}_{n}^* =1] = \alpha$. Moreover, since $D_n$ is a consistent estimate of $\EMMD$ (recall Theorem \ref{thm:consistencyDn}), then following the proof of Corollary \ref{cor:asymptest}, the test $\tilde{\phi}_{n}^*$ is also consistent for fixed alternatives, that is, $\lim_{n \rightarrow \infty} \P_{\bH_1}[\tilde{\phi}_{n}^* =1] = 1$. This shows that $\tilde{\phi}_{n}^*$ is a derandomized test for the conditional goodness-of-fit problem \eqref{eq:H0H1finite} which attains precise asymptotic level and is universally consistent.

\section{Applications}

  In this section we apply the $\EMMD$ test to the examples discussed in the Introduction. The section is organized as follows: In Section \ref{sec:classificationregression} we use the $\EMMD$ method for testing calibration in classification and regression models.  We investigate the performance of the finite-sample and the derandomized test for the conditional goodness-of-fit problem in Section \ref{sec:simulationsM}. We  compare regression functions in the wind energy dataset in Section \ref{sec:regressionexample}. In Section \ref{sec:sbiexperiments} we apply the ECMMD measure to validate emulators in benchmark SBI examples and for simulating redshifts associated with galaxy images.

\subsection{Calibration Tests}
   \label{sec:classificationregression}

   In Section \ref{sec:classification} we apply $\EMMD$ based calibration tests in the classification setting and in Section \ref{sec:regression} we apply $\EMMD$ based calibration tests in regression models. Calibration of convolutional neural network models is tested on the CIFAR-10 dataset using the $\EMMD$ measure in Section \ref{sec:imagecalibration}

  \subsubsection{Calibration Tests for Classification} 
  \label{sec:classification}

  For testing calibration in classification we consider the following data 
    generating mechanism (as in \cite{widmann2019calibration}), 
    \begin{align*}
        \bm f(Z) = (f_1(Z), 1-f_1(Z))^\top \sim \textnormal{Dir}(\rho, 1-\rho) \text{ and } X\sim \mathrm{Ber}(f_1(Z)) . 
    \end{align*}
    To examine the Type I error rate and the power we consider the following setups: 
    \begin{itemize}
        \item Null hypothesis:  $ Y\sim \textnormal{Bern}(f_1(Z))$;
        \item Alternative hypothesis: $ Y\sim \textnormal{Bern}(f_1( Z) - f_1(Z)^2)$.
    \end{itemize}
    We implement the $\EMMD$ asymptotic test from \eqref{eq:defphin} and derandomized test from \eqref{eq:defphinM} with the linear kernel $\sfK(x,y)=x\cdot y$, by varying $\rho\in\{0.1, 0.2, 0.3, 0.4, 0.5\}$, the number of nearest-neighbor $K \in \{15, 25\}$, and the sample size $n = 100$. We use $M_n=20$ in derandomized test \eqref{eq:defphinM}. Figure \ref{fig:classcalibration2} shows the empirical Type I error and power over 500 repetitions. Additional simulations with sample $n=75$ are given in Appendix \ref{sec:additionalclassification}.

\begin{figure}[ht]
    \begin{subfigure}{0.5\textwidth}
        \centering 
        \small
        \centering\begingroup\fontsize{11}{20}\selectfont

\begin{tabular}{cccccc}
\multicolumn{6}{c}{Type I Error} \\
\toprule
\multirow{2}{*}{$\rho$} &\multirow{2}{*}{SKCE} & \multicolumn{2}{c}{Asymptotic} & \multicolumn{2}{c}{Derandomized}\\
\cmidrule(l{3pt}r{3pt}){3-6}
& & 15 NN & 25 NN &  15 NN & 25 NN\\
\midrule
\cellcolor{gray!6}{0.1} & \cellcolor{gray!6}{0.064} & \cellcolor{gray!6}{0.050} & \cellcolor{gray!6}{0.048} & \cellcolor{gray!6}{0.042} & \cellcolor{gray!6}{0.052}\\
0.2 & 0.066 & 0.032 & 0.032 & 0.044 & 0.044\\
\cellcolor{gray!6}{0.3} & \cellcolor{gray!6}{0.076} & \cellcolor{gray!6}{0.066} & \cellcolor{gray!6}{0.066} & \cellcolor{gray!6}{0.054} & \cellcolor{gray!6}{0.052}\\
0.4 & 0.076 & 0.062 & 0.054 & 0.040 & 0.046\\
\cellcolor{gray!6}{0.5} & \cellcolor{gray!6}{0.074} & \cellcolor{gray!6}{0.052} & \cellcolor{gray!6}{0.056} & \cellcolor{gray!6}{0.044} & \cellcolor{gray!6}{0.052}\\
\bottomrule
\end{tabular}
\endgroup{}

        \caption*{(a)}
    \end{subfigure}%
    \begin{subfigure}{0.5\textwidth}
        \centering
        \includegraphics[height=6.25cm]{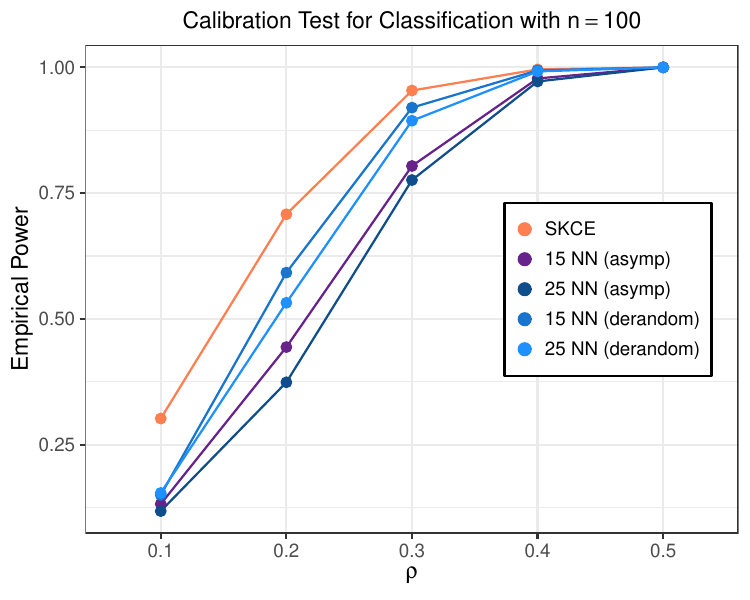} \\ 
            \small{ (b) }
    \end{subfigure}
    \caption{\small Calibration tests for classification: (a) Type I error and (b) empirical power for $n=100$, as a function of the signal strength $\rho$. }
    \label{fig:classcalibration2}
\end{figure}

     Note that, under the alternative hypothesis, as $\rho$ (the signal strength) grows, $f_1^2(Z)$ will tend to be larger, which makes distribution of $Y|Z$ further apart from that of $X|Z$. For comparison we also implement the test based on squared kernel calibration error (SKCE) \cite{widmann2019calibration}. The asymptotic null distribution of the test based on SKCE is an infinite weighted sum of $\chi^2$ distributions \cite{widmann2019calibration}, whose quantiles are intractable.  Hence, to chose the cut-off under the null hypothesis we use a parametric bootstrap procedure as in \cite{widmann2019calibration}. Throughout, the nominal level is to be $0.05$. The results are shown in Figure \ref{fig:classcalibration1} and Figure \ref{fig:classcalibration2}. We oberseve the following: 
    \begin{itemize}
    
   \item Both $\EMMD$ asymptotic test and derandomized test control Type I error for all choices of $K$. The SKCE based test shows marginal Type I error inflation (see the tables in Figure \ref{fig:classcalibration2}(a) and \ref{fig:classcalibration1}(a) in Appendix \ref{sec:additionalclassification}).  
   
   \item In terms of power the SKCE test is slightly better, although the ECMMD tests (both asymptotic and derandomized) catch up as the number of nearest neighbors increase (see Figure \ref{fig:classcalibration2}(b) and Figure \ref{fig:classcalibration1}(b) in Appendix \ref{sec:additionalclassification}). This difference in performance is expected as the SKCE based test uses the functional form of the conditional distribution, whereas the ECMMD based tests only utilize samples from the distribution. It is also worth noticing that the derandomized test improves on the power of asymptotic test by reducing the randomness in sampling. 

\item In Appendix \ref{sec:computation} we compare the time complexities of the SKCE and the ECMMD tests. Figure \ref{fig:computationcalibration}(a) in particular shows that the ECMMD tests (both asymptotic and derandomized) are  significantly faster than the SKCE test. This is expected because the ECMMD tests can be directly implemented without any bootstrap resampling, unlike the SKCE test. 

   \end{itemize} 
   In summary, the SKCE and ECMMD tests have comparable statistical performance, but the ECMMD tests are computationally much more efficient than the SKCE test, for validation of calibration in classification settings. 

  \subsubsection{Calibration Tests for Regression} 
  \label{sec:regression}

In the regression framework calibration is often defined in terms of the quantiles of response distribution \cite{liu2023distribution}. Specifically, suppose $Q(\rho, z)$ is a pre-trained quantile prediction model, which gives the prediction of the $\rho$-th quantile of the conditional distribution of $Y$ given $Z = z$. The quantile prediction model $Q$ is said to be calibrated if and only if,
\begin{align}\label{eq:regcalibration}
    \P\left(Y\leq Q(\rho,Z)|Z = z\right) = \rho , 
\end{align}
for all $0\leq \rho\leq 1$ and $z\in \text{Supp}( P_{Z})$. Notice that for $U\sim \text{Unif}[0,1]$ independent of all previous data, $Q(\cdot, Z)$ is a quantile function for $Q(U,Z)|Z$ and hence, from \eqref{eq:regcalibration}, $Q$ is calibrated if and only if,
\begin{align*}
     P_{Y|Z} =  P_{Q(U,Z)|Z}\text{ almost surely } P_Z,
\end{align*}
reducing the test for calibration of $Q$ to the form \eqref{eq:H0H1}. For more general predictive models, \eqref{eq:classcalibration} can be further extended as follows (see \cite{widmann2022calibration}):  A model $F_Z$ of the conditional distribution $ P_{Y|Z}$ is called calibrated if and only if,
    \begin{align}\label{eq:calibration}
         P_{Y|F_Z} = F_Z\text{ almost surely } P_{F_Z}.
    \end{align}
Note that when $F_Z$ is a classification model, \eqref{eq:calibration} matches with \eqref{eq:classcalibration}. Similar to the classification setting, testing for calibration of a general predictive model $F_Z$ can be framed in terms of  \eqref{eq:H0H1} as follows: For $X\sim F_Z$ notice that $ P_{X|F_Z} = F_Z$ almost surely $ P_{F_Z}$. Then by \eqref{eq:calibration}, the predictive model $F_Z$ is calibrated if and only if,
\begin{align*}
     P_{Y|F_Z} =  P_{X|F_Z}\text{ almost surely } P_{F_Z}.
\end{align*}
Thus, testing for calibration of $F_Z$ is now equivalent to the hypothesis test in \eqref{eq:H0H1} with samples $(X_1,Y_1, F_{Z_1}),\ldots, (X_n,Y_n, F_{Z_n})$, where $X_i\sim F_{Z_i}$, for all $1\leq i\leq n$.

\begin{figure}[!ht]
    \begin{subfigure}{0.5\textwidth}
        \centering
        \includegraphics[height=6.25cm]{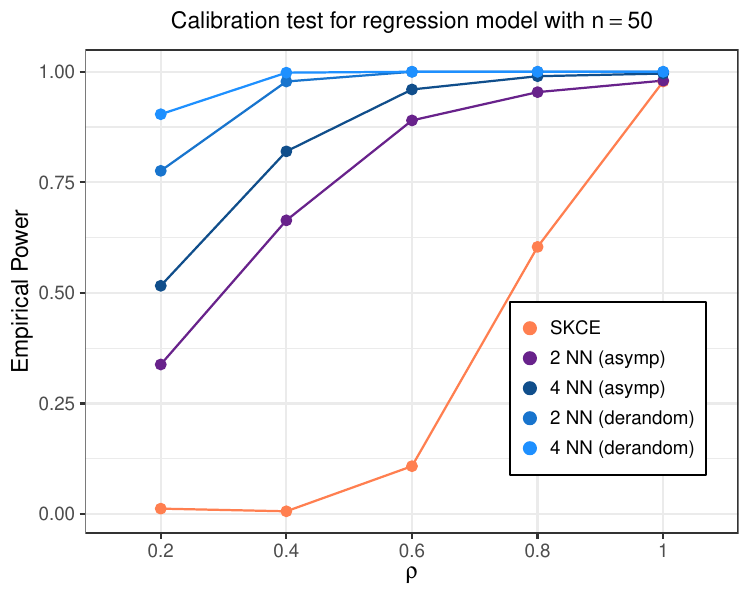} \\ 
        \small{ (a) }
    \end{subfigure}%
    \begin{subfigure}{0.5\textwidth}
        \centering
        \includegraphics[height=6.25cm]{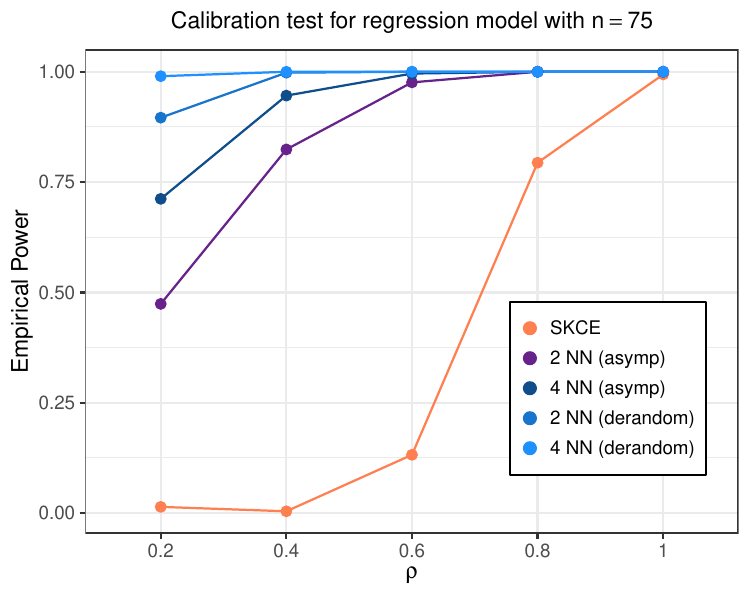} \\ 
        \small{ (b) }
    \end{subfigure}
    \caption{ \small{Empirical power of calibration tests for the regression model \eqref{eq:YZ} as a function of the signal strength $\rho$, with (a) $n=50$ and (b) $n=75$.} } 
    \label{fig:regcalibration}
\end{figure}

\begin{remark}\label{rmk:gaussian_regression_calibration}
    A special case of interest is when $F_Z$ is a Gaussian linear model. In this case, denote the conditional mean and conditional variance as $\E[X|F_Z]$ and $\mathrm{Var}[X|F_Z]$, respectively. Then testing $X|F_Z\overset{D}{=}Y|F_Z$ is the same as:
        $$X|(\E[X|F_Z],\mathrm{Var}[X|F_Z])\overset{D}{=}Y|(\E[X|F_Z],\mathrm{Var}[X|F_Z]),$$ since in the mean and the variance determines a Gaussian distribution. In particular, if one uses the homoscedastic linear model, then it suffices to condition only on the conditional mean.  We will use such a model in the following simulation. 
\end{remark}

To evaluate the performance of the $\EMMD$ test for regression calibration we consider the following model inspired by \citet{widmann2022calibration}: 
    \begin{align}\label{eq:YZ}
        Y=\rho\times\sin(\pi Z)+|1+Z| \varepsilon, \text{ where } \varepsilon\sim N(0,0.15^2) \text{ and } Z\sim \text{Unif}([-1,1]) , 
    \end{align}  
    for $\rho > 0$. Considering a training set of size $n_{\mathrm{train}} = 200$ we fit an ordinary least squares (OLS) to $\{(Y_i, Z_i)\}_{1 \leq i \leq n_{\mathrm{train}}}$ generated i.i.d from \eqref{eq:YZ}, by varying $\rho \in \{0.2, 0.4, 0.6, 0.8, 1\}$. We denote the fitted regression coefficient as $\hat\beta$.  For testing this model we generate data as follows: given $Z_i$ we use \eqref{eq:YZ} to generate $Y_i$ and the conditional normal model to generate $X_i \sim N(Z_i\hat\beta,\hat\sigma^2)$, where $\hat\sigma^2=\sum_{i=1}^{n_{\mathrm{train}}} (Y_i-Z_i \hat\beta)^2/(n_{\mathrm{train}}-1)$, for $1 \leq i \leq n_{\mathrm{test}}$. We will use the samples $\{(X_i, Y_i, Z_i\hat \beta)\}_{1 \leq i \leq n_{\mathrm{test}}}$ to perform the calibration test, since the training model uses the homoscedastic error (recall Remark \ref{rmk:gaussian_regression_calibration}).  
      Specifically, we implement the $\EMMD$ method using asymptotic test \eqref{eq:defphin} and derandomized test \eqref{eq:defphinM} with the Gaussian kernel and the bandwidth chosen as the median of absolute differences $\{|X_i-Y_i|:1\leq i\leq n_{\mathrm{test}}\}$. For the derandomized test \eqref{eq:defphinM} we choose $M_n=20$. The number of nearest neighbors varies over $K \in \{2, 4, 6, 8, 10\}$. We also implement the SKCE test from \cite{widmann2022calibration} for comparison. Figure \ref{fig:regcalibration} shows the empirical power for $n_{\mathrm{test}} \in \{50, 75\}$ over 500 repetitions.  

    Note that due to model misspecification, the OLS model is clearly not calibrated in this case (see \cite[Section A.1]{widmann2022calibration} for further discussion). Moreover, the true model is similar to a heteroscedastic linear model for small values of $\rho$ (since, by a Taylor expansion, $\rho\sin(\pi Z)\approx \rho \pi Z$, when $\rho$ is small), but becomes more non-linear as $\rho$ increases. Hence, the power of the tests are expected to increase as the signal strength $\rho$ increases. This aligns with the results in Figure \ref{fig:regcalibration}. The plots also reveal the following: 
    \begin{itemize}
    
    \item The $\EMMD$ tests (both asymptotic and derandomized)  have better power than the SKCE test and the power of the $\EMMD$ tests increase as $K$ increases.  
    
    \item Even when $\rho$ is small, the $\EMMD$ tests has non-trivial power, whereas the SKCE almost has no power.  This shows that the ECMMD is more sensitive to detecting the heteroscedasticity in the true model (which is approximately linear for small values of $\rho$) than the SKCE.   
    
    \item In Figure \ref{fig:computationcalibration}(b) in Appendix \ref{sec:computation} we provide a comparison of the computation times of the SKCE and the ECMMD tests. As in the case of classification, ECMMD tests  are  significantly faster than the SKCE test.

    \end{itemize}

\subsubsection{Testing Calibration of CNN on Real Data} 
\label{sec:imagecalibration}

In this section, we apply the $\EMMD$ for testing calibration of convolutional neural networks using the CIFAR-10 dataset (\url{https://www.cs.toronto.edu/~kriz/cifar.html}). The dataset 
consists of 10 classes of objects: airplane, automobile, bird, cat, deer, dog, frog, horse, ship, and truck. Each class has 6000 images and the total number of images is 60000. For our experiment, we use the following pairs for binary classification
\begin{align}\label{eq:imagepairs}
    \{\text{bird, cat}\}, \{\text{cat, dog}\}, \{\text{cat, deer}\}, \{\text{cat, frog}\}, \{\text{cat, horse}\}.
\end{align}
In the following we provide the results for $\{\text{bird, cat}\}$ and $\{\text{cat, dog}\}$. The results for the remaining 3 pairs are given in Appendix \ref{sec:imagecalibration3}. The ratio of training data and test data is chosen to be $3:2$. The training data is used to learn a convolutional neural network (CNN) classifier, and the test data will used for assessing calibration based on the $\EMMD$ measure. 

\subsubsection*{Model Setup}

To train the classifier, we use three convolutional layers with 32, 64, and 64 filters, respectively, interspersed with $2\times 2$ max-pooling layers to reduce spatial dimensions. Post convolution, a flattening layer 
transforms the 2D feature maps into a 1D vector, followed by a dense layer with 64 neurons. The output  layer uses a softmax activation function for binary classification.  The model employs the Adam optimizer, categorical cross-entropy loss, and tracks accuracy as its metric.

\subsubsection*{Potential Miscalibration and Recalibration} 

Following the influential paper of \citet{guo2017calibration}, it is now common knowledge that deep-learning models tend to overfit the data which can lead potential miscalibration. To check this we will use the $\EMMD$ measure to test if the convolutional neural network trained as described above is calibrated using the pairs in \eqref{eq:imagepairs}. We then split the test data according to $2:1$ ratio, use the first part of the data to recalibrate the prediction probabilities using isotonic regression (see Appendix \ref{sec:recalibration} for details), and the second part to again test using the ECMMD if  the recalibrated probabilities are indeed calibrated. 

\subsubsection*{Calibration Test}

We randomly choose $1000$ samples from the test data and implement the ECMMD test with the linear kernel, which is a characteristic kernel under the Bernoulli distribution class. To test for calibration we implement the asymptotic test \eqref{eq:defphin} (by generating $1000$ samples from the learned prediction probabilities) as well as the finite-sample test proposed in Algorithm \ref{algorithm:finitesample} with the number of resamples $M=500$. Given the relatively large sample size, we vary the number of nearest neighbors as $K \in \{60, 80, 100, 120\}$. We repeat the above experiment $100$ times to report the proportion of rejection before and after recalibration in Figures \ref{fig:reliability-calibration-1}(b), \ref{fig:reliability-calibration-2}(b) and Figures \ref{fig:reliability-calibration-3}(b), \ref{fig:reliability-calibration-4}(b), \ref{fig:reliability-calibration-5}(b) in Appendix \ref{sec:imagecalibration3}. 

\begin{figure}[ht]
    \begin{subfigure}{0.45\textwidth}
        \centering
        \includegraphics[scale=0.55]{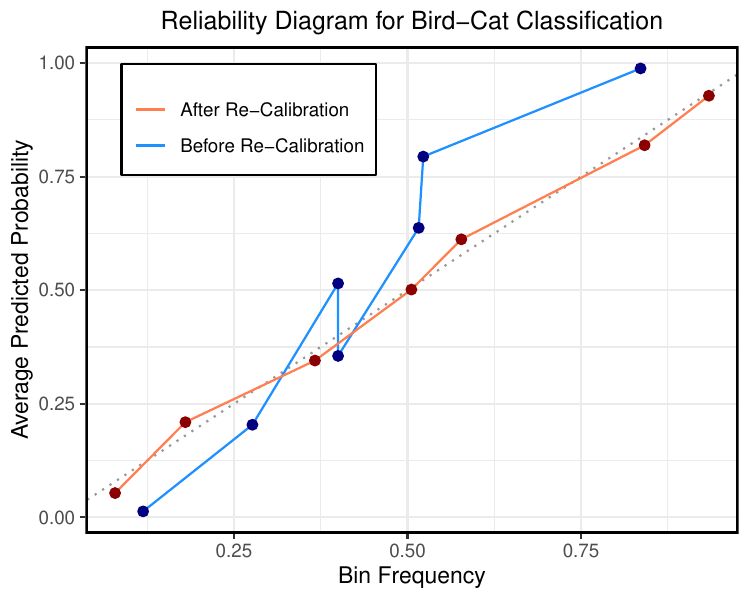}
        \caption*{(a)}
    \end{subfigure}%
    \begin{subfigure}{0.55\textwidth}
        \centering 
        \small
        \begin{tabular}{cccc}
\toprule
\multicolumn{1}{c}{ } & \multicolumn{1}{c}{ } & \multicolumn{2}{c}{Rejection Proportion} \\
\cmidrule(l{3pt}r{3pt}){3-4}
Test & \textit{K} & \multicolumn{1}{p{2.5cm}}{\centering Before\\ Re-Calibration} & \multicolumn{1}{p{2.5cm}}{\centering After\\ Re-Calibration}\\
\midrule
\cellcolor{gray!6}{FS} & \cellcolor{gray!6}{60} & \cellcolor{gray!6}{1.00} & \cellcolor{gray!6}{0.07}\\
FS & 80 & 1.00 & 0.07\\
\cellcolor{gray!6}{FS} & \cellcolor{gray!6}{100} & \cellcolor{gray!6}{1.00} & \cellcolor{gray!6}{0.05}\\
FS & 120 & 1.00 & 0.05\\
\midrule
\cellcolor{gray!6}{Asymp} & \cellcolor{gray!6}{60} & \cellcolor{gray!6}{1.00} & \cellcolor{gray!6}{0.08}\\
Asymp & 80 & 1.00 & 0.08\\
\cellcolor{gray!6}{Asymp} & \cellcolor{gray!6}{100} & \cellcolor{gray!6}{1.00} & \cellcolor{gray!6}{0.05}\\
Asymp & 120 & 1.00 & 0.05\\
\midrule
\cellcolor{gray!6}{ECE} & \cellcolor{gray!6}{} & \cellcolor{gray!6}{0.31} & \cellcolor{gray!6}{0.12}\\
\bottomrule
\end{tabular}

        \caption*{(b)}
    \end{subfigure}
    
    \caption{ \small{ Results for bird-cat classification: (a) reliability plot before recalibration and after recalibration, and (b) $p$-values of the $\EMMD$ test for different values of $K$ and ECE before and after recalibration.
    }}
    \label{fig:reliability-calibration-1} 
\end{figure}

\begin{figure}[ht]
    \begin{subfigure}{0.45\textwidth}
        \centering
        \includegraphics[scale=0.55]{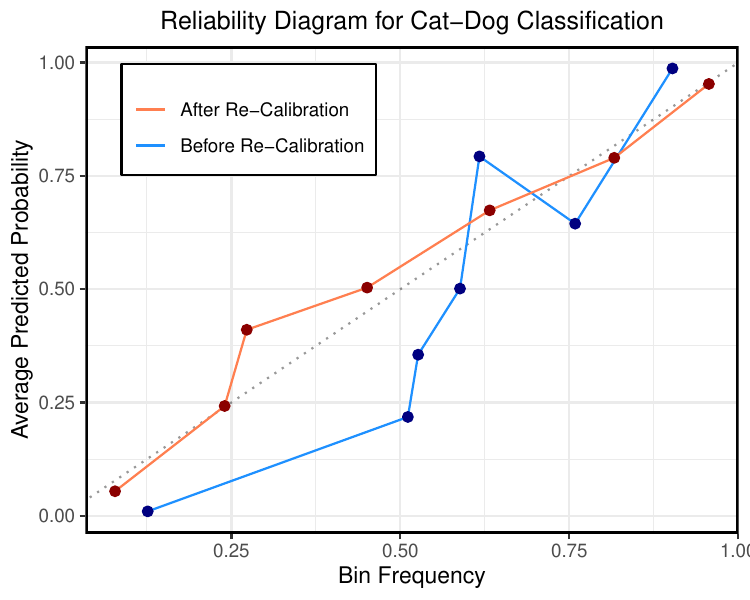}
        \caption*{(a)}
    \end{subfigure}%
    \begin{subfigure}{0.55\textwidth}
        \centering 
        \small
        \begin{tabular}{cccc}
\toprule
\multicolumn{1}{c}{ } & \multicolumn{1}{c}{ } & \multicolumn{2}{c}{Rejection Proportion} \\
\cmidrule(l{3pt}r{3pt}){3-4}
Test & \textit{K} & \multicolumn{1}{p{2.5cm}}{\centering Before\\ Re-Calibration} & \multicolumn{1}{p{2.5cm}}{\centering After\\ Re-Calibration}\\
\midrule
\cellcolor{gray!6}{FS} & \cellcolor{gray!6}{60} & \cellcolor{gray!6}{1.00} & \cellcolor{gray!6}{0.15}\\
FS & 80 & 1.00 & 0.14\\
\cellcolor{gray!6}{FS} & \cellcolor{gray!6}{100} & \cellcolor{gray!6}{1.00} & \cellcolor{gray!6}{0.13}\\
FS & 120 & 1.00 & 0.15\\
\midrule
\cellcolor{gray!6}{Asymp} & \cellcolor{gray!6}{60} & \cellcolor{gray!6}{1.00} & \cellcolor{gray!6}{0.17}\\
Asymp & 80 & 1.00 & 0.17\\
\cellcolor{gray!6}{Asymp} & \cellcolor{gray!6}{100} & \cellcolor{gray!6}{1.00} & \cellcolor{gray!6}{0.15}\\
Asymp & 120 & 1.00 & 0.14\\
\midrule
\cellcolor{gray!6}{ECE} & \cellcolor{gray!6}{} & \cellcolor{gray!6}{0.26} & \cellcolor{gray!6}{0.15}\\
\bottomrule
\end{tabular}

        \caption*{(b)}
    \end{subfigure}
    \caption{ \small{ Results for cat-dog classification: (a) reliability plot before recalibration and after recalibration, and (b) $p$-values of the $\EMMD$ test for different values of $K$ and ECE before and after recalibration. } }
    \label{fig:reliability-calibration-2} 
\end{figure}

\subsubsection*{Results and Interpretation}

To illustrate our results and the effect of recalibration, we show the reliability plots before after recalibration for the 5 pairs in \eqref{eq:imagepairs} are shown in Figures \ref{fig:reliability-calibration-1}(a), \ref{fig:reliability-calibration-2}(a) and Figures \ref{fig:reliability-calibration-3}(a), \ref{fig:reliability-calibration-4}(a), \ref{fig:reliability-calibration-5}(a) in Appendix \ref{sec:imagecalibration3}. A curve aligned with the diagonal line in the reliability diagram indicates that the model is well-calibrated (see Appendix \ref{sec:reliability} for details). The  proportion of rejections of the ECMMD tests, both before and after recalibration, are reported in the tables in Figures \ref{fig:reliability-calibration-1}(b), \ref{fig:reliability-calibration-2}(b), \ref{fig:reliability-calibration-3}(b), \ref{fig:reliability-calibration-4}(b), \ref{fig:reliability-calibration-5}(b). We use $\tt{FS}$ and $\tt{Asymp}$ to denote the finite sample and the asymptotic tests, respectively. The tables also show the estimated values of the expected calibration error (ECE) (averaged over $100$ experiments) computed with the number of bins set to be $100$ (see Appendix \ref{sec:calibrationerror} for details computation of the ECE). The following is a summary of our findings: 

\begin{itemize}
    \item Before recalibration, the reliability diagrams and the proportion of rejection indicate that  CNN classifier is mis-calibrated for all the 5 pairs in \eqref{eq:imagepairs}. 
    
    \item After recalibration, the rejection proportion of the ECMMD test are significantly reduced. The improved calibration can also be seen from the significant reduction of the ECE values. 
\end{itemize} 
The illustrates the efficacy of the $\EMMD$ in testing calibration of deep-learning based classification models in a real-data setting. 

\color{black}

\subsection{Conditional Goodness-of-Fit Test}
\label{sec:simulationsM}

To illustrate the performance of the resampling based conditional goodness-of-fit test using the ECMDD (as in Section \ref{sec:finitetest}), we consider the following heteroscedastic Gaussian model adapted from \cite{jitkrittum2020testing}. For $d\geq 1$, let $Z\sim N_{d}\left(\bm{0}, I_d\right)$ and define,
\begin{align*}
    P_{X|Z} = N \left(Z^\top\bm\one_d, \sigma^2(Z)\right), \text{ where } \sigma^2(Z) = 1 + 10\exp\left(-\frac{\left\| Z-0.8\one_d\right\|_2^2 }{2\times 0.8^2}\right).
\end{align*}
To test the hypothesis from \eqref{eq:H0H1finite} with $P_{X|Z}$ as above, we consider observed samples $(Y_1, Z_1), \ldots,$ $(Y_n, Z_n)$ generated independently from $P_{Y|Z} = N (Z^\top\bm 1_d, 1 )$.\footnote{Here, $\bm 1_d$ denotes the vector all 1 in dimension $d$. Also, $N_d( \mu, \Sigma)$ denotes the $d$-dimensional normal distribution with mean vector $\mu$ and covariance matrix $\Sigma$. In dimension 1, the subscript $d$ is dropped.}  As in \cite{jitkrittum2020testing} we set $d=3$ in the above simulation setup. For comparison we also implement the KCSD test from \cite{jitkrittum2020testing} in the above setup. For the ECMMD based test we use the Gaussian kernel with the bandwidth chosen as the median of absolute differences (as in Section \ref{sec:regression}). 
We implement both finite sample test in Algorithm \ref{algorithm:finitesample} and derandomized test \eqref{eq:defphinM}. The number of nearest-neighbors vary over $K \in \{20, 40\}$ and we choose number of resamples $M=300$ in the finite sample test and $M_n=20$ in the derandomized test. To implement the KCSD test we use Gaussian kernel with the median bandwidth and calibrate the test with $500$ multiplier bootstrap samples. For both tests the nominal level is set to $\alpha = 0.05$.

\begin{figure}[ht]
    \begin{subfigure}{0.5\textwidth}
        \centering 
        \small
        \setlength{\tabcolsep}{4pt}
\centering\begingroup\fontsize{10}{23}\selectfont

\begin{tabular}{cccccc}
\multicolumn{6}{c}{Type I Error} \\
\toprule
\multirow{2}{*}{$n$} &\multirow{2}{*}{KCSD} & \multicolumn{2}{c}{Derandomized} & \multicolumn{2}{c}{Finite-Sample}\\
\cmidrule(l{3pt}r{3pt}){3-6}
& & NN20 & NN40 &  NN20 & NN40\\
\specialrule{\cmidrulewidth}{0pt}{0pt}
50 & 0.038 & 0.054 & 0.060 & 0.048 & 0.062\\
\cellcolor{gray!6}{100} & \cellcolor{gray!6}{0.052} & \cellcolor{gray!6}{0.056} & \cellcolor{gray!6}{0.040} & \cellcolor{gray!6}{0.062} & \cellcolor{gray!6}{0.052}\\
150 & 0.036 & 0.042 & 0.070 & 0.062 & 0.056\\
\cellcolor{gray!6}{200} & \cellcolor{gray!6}{0.040} & \cellcolor{gray!6}{0.052} & \cellcolor{gray!6}{0.052} & \cellcolor{gray!6}{0.070} & \cellcolor{gray!6}{0.062}\\
250 & 0.048 & 0.048 & 0.044 & 0.066 & 0.064\\
\bottomrule
\end{tabular}
\endgroup{}

        \caption*{(a)}
    \end{subfigure}%
    \begin{subfigure}{0.5\textwidth}
        \centering
        \includegraphics[height = 6.25cm]{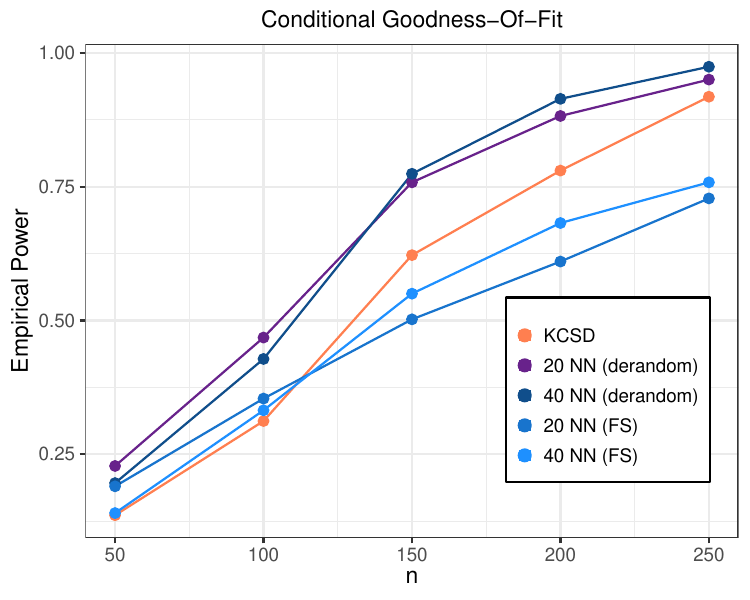} \\ 
        \small{ (b) }
    \end{subfigure}
    \caption{ \small{ Conditional goodness-of-fit test: (a) Type I error and (b) power of the test in \eqref{eq:defphin} and the KCSD test as a function of the sample size. } }
    \label{fig:condgoodfit}
\end{figure}

Figure \ref{fig:condgoodfit} shows the empirical Type I error and power, averaged over $500$ repetitions, of the ECMMD based tests and the KCSD test, as the sample size $n$ varies over $\{50, 100, 150, 200, 250\}$. The KCSD has slightly better power than the finite-sample $\EMMD$ tests. However, the derandomized $\EMMD$ tests outperform the KCSD. 
It is worth noting that the KCSD uses the functional form of the score function of the conditional distribution, whereas the ECMMD test is a completely nonparametric method that `learns' the conditional distribution through sampling. Despite that the ECMDD tests is able to compete and outperform  with the KCSD test, illustrating the strengths of our method.

\color{black}

\subsection{Comparing Regression Curves: Wind Energy Data Application}
\label{sec:regressionexample}

To illustrate the application of the $\EMMD$ based test in comparing two regression curves, we consider the wind energy dataset \cite{ding2019data,prakash2022gaussian,hwangbo2017production}.
The dataset is publicly available at \url{https://aml.engr.tamu.edu/book-dswe/dswe-datasets/} under \texttt{Inland and Offshore Wind Farm Dataset2} folder.
There are four datasets in the folder, each corresponding to one wind turbine. Here, we consider the datasets WT3 and WT4 which correspond to offshore wind turbines. Both these datasets consist of 4 years (2007-2010) of data and the main variable of interest (the outcome) is the wind power generated by different turbines across different time periods. The data also has the following five covariates:  wind speed (V), wind direction (D), air density ($\rho$), turbulence intensity (I), and humidity (S).

The question of interest to researchers in this field is to understand how the wind speed (V) affects the power generation in the presence of other covariates in order to further improve the efficiency of wind turbines. Previous literature considered the problem of testing if the effect of wind speed on power generation is the same across different years. For example, \cite{prakash2022gaussian} adopted a Bayesian nonparameteric approach
and provide a confidence band on the regression functions compared across different years. Here, we investigate if \textit{the effect of wind speed on wind power generation remain the same across different turbines within these two periods: 2007-2008 and 2009-2010?} 
To alleviate the temporal effect, we preprocess the data by averaging the observations within each day. This results in a sample of size $475$ and $626$ for periods $2007$-$2008$ and $2009$-$2010$, respectively. Since the data are recorded over time and the two turbines may be located in adjacent areas, we expect the covariate information to be similar in each recording point. Denote the outcome (averaged within each day as mentioned above) for WT3 and WT4 by $X$ and $Y$, respectively. We then concatenate the covariate vector $Z_1$ for WT3 and the covariate vector $Z_2$ for WT4 as a joint vector $Z = (Z_1,Z_2)$. The combined covariate vector is of size 10 (5 from each turbine). The final preprocessed data can be represented as $\{X_i, Y_i, Z_i\}_{1 \leq i \leq n}$. Based on this data, we now test the hypothesis \eqref{eq:H0H1hypothesis} using the $\EMMD$ statistic \eqref{eq:estEMMD} with the Gaussian kernel. We consider the bandwidth to be the median of absolute difference: $\mathrm{Median}\{|X_i-Y_i|:1\leq i\leq n\}$. We also test the equality of conditional expectations, that is, $\E[Y|Z]=\E[X|Z]$ almost surely, by choosing the linear kernel $\sfK(x, y) = x\cdot y$ in \eqref{eq:estEMMD}. We implement the asymptotic test \eqref{eq:defphin} for this application. Notice whenever the null hypothesis $X|Z\overset{D}{=}Y|Z$ is accepted, we can conclude that the effect of wind speed on wind power generation remain the same across turbines WT3 and WT4, given the other covariates are fixed.

\begin{figure}
    \centering
    \begin{subfigure}{0.5\textwidth}
        \centering
        \includegraphics[height=6.25cm]{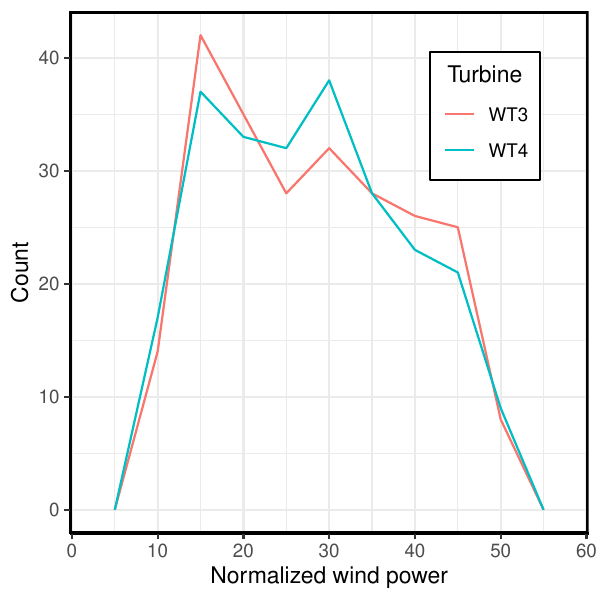}
    \end{subfigure}%
    \begin{subfigure}{0.5\textwidth}
        \centering
        \includegraphics[height=6.25cm]{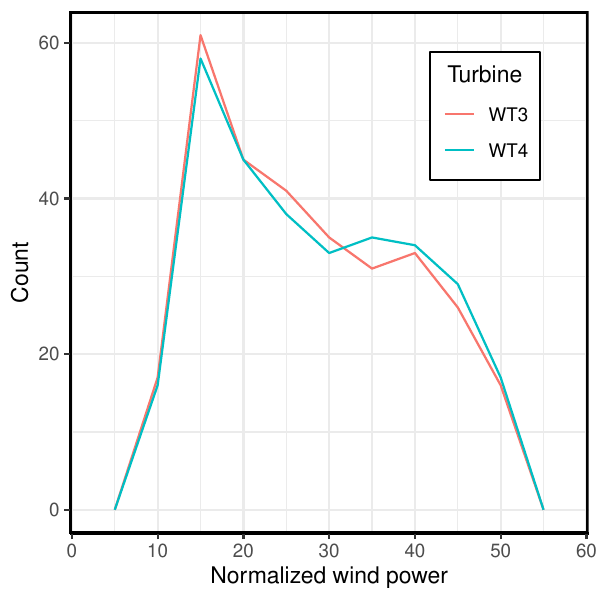}
    \end{subfigure}
    \caption{\small{ Frequency plots of wind power in WT3 ($X$) and WT4 ($Y$) during the 2007-2008 period (left) and 2009-2010 period (right). } }
    \label{fig:real_data_marginal_freq}
\end{figure}

\subsubsection{Results and Interpretation}

We first check the marginal distributions of $X$ and $Y$ by comparing their frequency plots based on the samples $\{X_i\}_{1 \leq i \leq n}$ and $\{Y_i\}_{1\leq i \leq n}$ (see Figure \ref{fig:real_data_marginal_freq}). The frequency plots show that the marginal distributions of $X$ and $Y$ are more different in the 2007-2008 period than the 2009-2010 period (although the overall trend in both cases are very similar). This suggests that there is potentially more heterogeneity between the conditional distributions $X|Z$ and $Y|Z$ in 2007-2008 period compared to the 2009-2010 period. 

\begin{table}[h]
    \centering
     \begin{subtable}[b]{0.45\linewidth}
        \centering
        \begin{tabular}{|c|c|c|}
            \hline
            $K$ & $\E[X|Z]=\E[Y|Z]$ & $\P_{X|Z}=\P_{Y|Z}$ \\
            \hline
            20 & 0.265 & 0.115  \\
            25 & 0.187 & 0.080   \\
            30 & 0.350  & 0.124 \\
            35 & 0.318 & 0.110 \\
            40 & 0.545 & 0.037  \\ 
            45 & 0.419  & 0.046 \\
           \hline 
          \end{tabular} 
           \caption*{(a)}
        \end{subtable}%
     \begin{subtable}[b]{0.45\linewidth}
            \centering
            \begin{tabular}{|c|c|c|}
                \hline
                $K$ & $\E[X|Z]=\E[Y|Z]$ & $\P_{X|Z}=\P_{Y|Z}$ \\
                \hline
                20 & 0.113 & 0.273  \\
                25 & 0.083 & 0.297  \\
                30 & 0.200 &  0.147   \\
                35 & 0.251 &  0.206  \\
                40 & 0.141 & 0.233  \\
                45 & 0.056 & 0.265   \\ 
                \hline 
            \end{tabular} 
            \caption*{(b)}
        \end{subtable}
        \caption{The $p$-values for testing equality of conditional means and  equality of conditional distributions for WT3 versus WT4 in the (a) 2007-2008 and (b) 2009-2010 period, for varying $K$.}
        \label{tab:real_data_turbine_analysis}
\end{table}

To confirm this, in Table \ref{tab:real_data_turbine_analysis} we report the $p$-values of the test \eqref{eq:defphin} with the linear kernel (for testing conditional mean equality) and the Gaussian kernel (for testing conditional distribution equality). We see that during the 2007-2008 period (see Table \ref{tab:real_data_turbine_analysis}(a)), at significance level $0.05$, the hypothesis $\P_{X|Z} = \P_{Y|Z}$ is rejected when $K$ (the number of nearest-neighbors) is large, but the conditional mean hypothesis $\E[X|Z]=\E[Y|Z]$ is not rejected. This indicates there is higher order heterogeneity in the conditional distribution across turbines beyond the conditional mean. However, this heterogeneity is not severe since it takes large number (for example, $K=40$ and $45$) of nearest neighbors to detect the difference. For the time period 2009-2010 (see Table \ref{tab:real_data_turbine_analysis}(b)), both the hypotheses $\P_{X|Z} = \P_{Y|Z}$ and  $\E[X|Z]=\E[Y|Z]$ are accepted, suggesting that during this period both turbines have similar conditional power generation given the wind speed and the other covariates. Similar comparison has been considered in \cite{prakash2022gaussian}. They provide the confidence band for the conditional mean difference of the wind power of WT3 and WT4 over wind speed for each year between 2007-2010, after adjusting the other covariates. Our result on conditional mean testing aligns with the conclusion in \cite[Table 7]{prakash2022gaussian} that there is no significant statistical difference in conditional mean between the wind power of WT3 and WT4 between the periods 2007-2008 and 2009-2010.

\subsection{Validation of Emulators in SBI}\label{sec:sbiexperiments}

In this section we illustrate the applicability of $\EMMD$ measure in testing validity of emulators and approximate posteriors for the true posterior distribution in simulation based inference (SBI).

\subsubsection{Approximate Posteriors in Benchmark Examples}
\label{sec:posteriorapproximation}

\begin{figure}
    \centering
    \begin{subfigure}{.5\textwidth}
      \centering
      \includegraphics[width=0.8\linewidth]{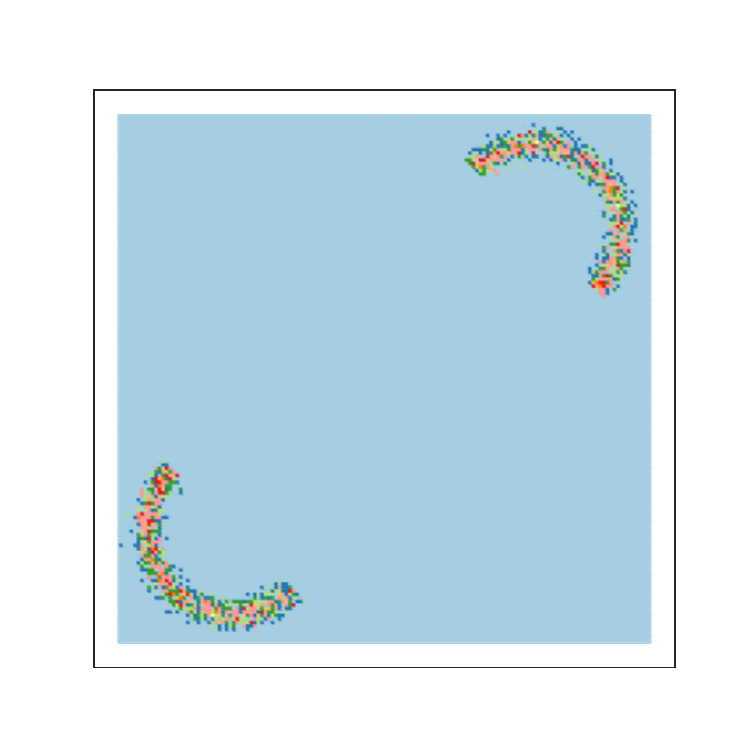}
      \caption*{\small (a)}
    \end{subfigure}%
    \begin{subfigure}{.5\textwidth}
      \centering
      \includegraphics[width=0.8\linewidth]{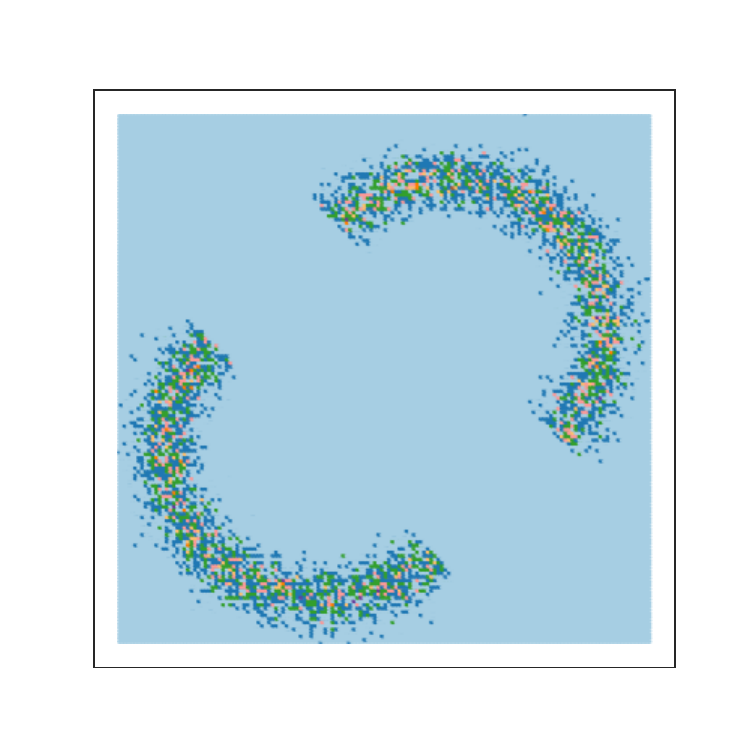}
      \caption*{\small (b)}
    \end{subfigure}
    \caption{ \small{ Examples of posterior distribution in the \textit{Two Moons} task. Samples are generated using the reference posterior from the \texttt{sbibm} \cite{lueckmann2021benchmarking} package with (a) $x = (0.0030109584, -0.6457662)$ and (b) $x = (0.19458583, 1.0400153)$.} }
    \label{fig:twomoonsposterior}
\end{figure} 

We consider the following 2 examples from the literature: 
\begin{itemize}
    \item \textit{Two Moons}: This is a two-dimensional example with a posterior that exhibits both global (bimodality) and local
    (crescent shape) structure (see \cite{greenberg2019automatic,lueckmann2021benchmarking,papamakarios2019sequential,linhart2024c2st,wildberger2024flow,glaser2022maximum}). Here, the data generating process has the following description: 
    \begin{itemize}
        \item Generate $\theta = (\theta_1,\theta_2)^\top$ from the uniform distribution on the unit square $[-1,1]\times [-1,1]$.
        \item Given $\theta$ generate $X$ from the simulator:
        \begin{align*}
            X|\theta = 
            \begin{pmatrix}
                r\cos(\alpha) + 0.25\\
                r\sin(\alpha)
            \end{pmatrix} + 
            \begin{pmatrix}
                -\left|\theta_1+\theta_2\right|/\sqrt{2}\\
                \left(-\theta_1+\theta_2\right)/\sqrt{2}
            \end{pmatrix} , 
        \end{align*}
        where $\alpha\sim \text{Unif}(-\pi/2,\pi/2)$ and $r\sim N(0.1, 0.01^2)$.
    \end{itemize} 
    In Figure \ref{fig:twomoonsposterior} we show the heatmap of the posterior distribution $\theta | X=x$ for 2 different values of $x$ illustrating the bimodality and the cresent shaped structure. 
    \item \textit{Simple Likelihood Complex Posterior (SLCP):} This is an example with a simple likelihood that exhibits complex posterior phenomena. Here, we consider the following data generating process, which is a 
     simplification of the SLCP task considered in \cite{papamakarios2019sequential,greenberg2019automatic,hermans2020likelihood,durkan2020contrastive,lueckmann2021benchmarking}.  

       \begin{itemize}
        \item Generate $\theta = (\theta_1,\theta_2,\theta_3,\theta_4,\theta_5)^\top$ from Uniform distribution on the box $[-3,3]^5$. 
        \item Given $\theta$ generate  $ X|\theta\sim N\left({m}_{\theta}, {S}_{\theta}\right)$ where,
        \begin{align*}
            {m}_{\theta} = 
            \begin{pmatrix}
                \theta_1\\
                \theta_2
            \end{pmatrix}\text{ and }
            {S}_{\theta} = 
            \begin{pmatrix}
                \theta_3^4 & \theta_3^2\theta_4^2\tanh(\theta_5)\\
                \theta_3^2\theta_4^2\tanh(\theta_5) & \theta_4^4
            \end{pmatrix} . 
        \end{align*} 
    \end{itemize}
\end{itemize} 

For both the above examples the posterior distributions are intractable. To approximate the posteriors we consider the following emulators: 
\begin{itemize} 

\item Mixture Density Networks (MDNs) \cite{bishop1994mixture,papamakarios2016fast}: In this case, the emulator $q_\phi(\theta|\bm X)$ takes the form of a mixture of $\ell$ Gaussian components where the mixing coefficients, the means, and the covariance matrices are parametrized by a feed-forward neural network (which we denote by $\phi$) taking $X$ as an input. For the \texttt{Two Moons} example we consider MDNs with $1, 3, 5, 7$ mixtures and for \texttt{SLCP} we consider MDNs with $5, 10, 15, 25$ mixtures.

\item  Conditional density estimates using Neural Spline Flow (NSF) \cite{durkan2019neural}, as in \cite{papamakarios2016fast,lueckmann2021benchmarking}. In our experiments we implement a NSF with $5$ transforms. 

\end{itemize}

\begin{figure}[h]
    \centering
    \begin{subfigure}{.5\textwidth}
      \centering
      \includegraphics[width=0.9\linewidth]{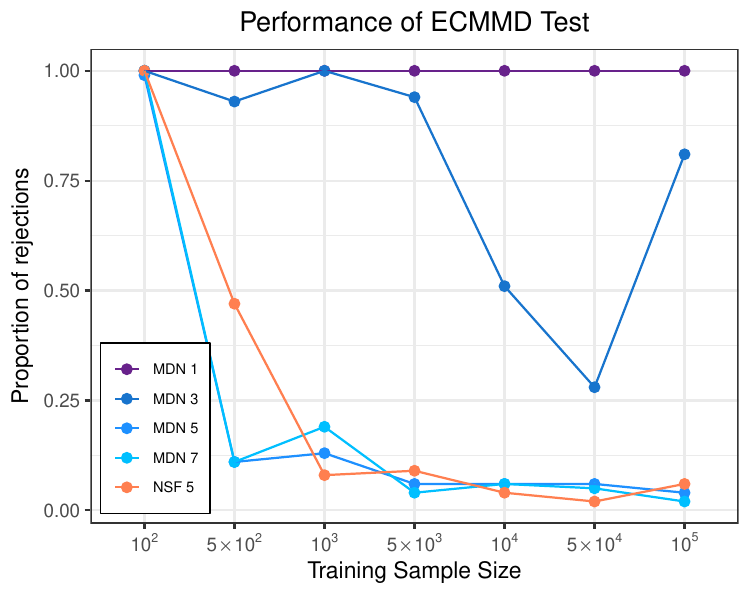}
      \caption*{\small{(a)}} 
    \end{subfigure}%
    \begin{subfigure}{.5\textwidth}
      \centering
      \includegraphics[width=0.9\linewidth]{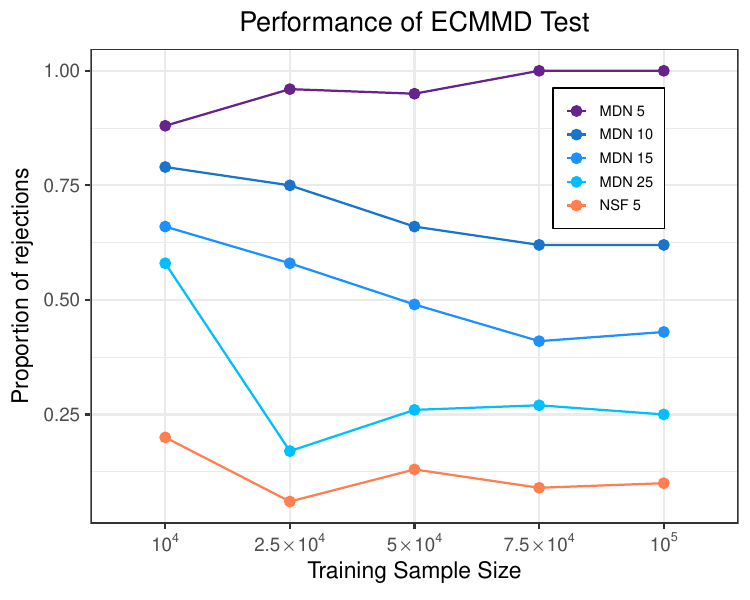}
      \caption*{(b)}
    \end{subfigure}
    \caption{ \small{ Empirical power of the ECMMD test for validating the emulators of the posterior distributions in (a) \texttt{Two Moons} and (b) \texttt{SLCP} examples as a function of training budget. } }
    \label{fig:twomoonsSLCP}
\end{figure} 

The MDN and NSF based posterior approximations are trained using their implementation in the \texttt{sbi} \cite{tejero2020sbi} package. We use a training budget (sample size) of $N\in \{10^2, 5\times 10^2, 10^3, 5\times 10^3, 10^4, 5\times 10^4, 10^5\}$ for the \texttt{Two Moons} example and $N\in \{10^4, 2.5\times 10^4, 5\times 10^4, 7.5\times 10^4, 10^5\}$ for the \texttt{SLCP} example.  For testing the validity of the emulators we consider a test budget of $n=1000$ with $K = 50$ nearest neighbors. For each $1\leq i\leq 1000$, we generate $(\theta_i, X_i)$, where $\theta_i$ is generated from the prior and $X_i$ is generated from the simulator given $\theta_i$. Then with $X_i$ as input to the posterior emulator we generate a sample $\theta_i^\prime$, for $1\leq i\leq 1000$. Based on samples $(\theta_i^\prime, \theta_i, X_i)$, for $1\leq i\leq 1000$, we implement the asymptotic ECMMD based test from \eqref{eq:defphin}. We use a nominal level of $\alpha = 0.05$ and the Gaussian kernel with bandwidth $\lambda = \text{Median}\left\{\left\|\theta_i-\theta_i^\prime\right\|:1\leq i\leq 1000\right\}$. The test is repeated $100$ times (with the same trained model) to estimate the proportion of rejections as a function of the training budget, for different models.  

The results are shown in Figure \ref{fig:twomoonsSLCP}. 
In the \texttt{Two Moons} example, MDNs with $5$ and $7$ mixtures are able to emulate the posterior distribution (with moderate training budget). On the other hand, in the more complex \texttt{SLCP} example, even with $25$ mixtures the MDN is unable to fully capture the complexities of the posterior distribution. In contrast, the NSF model with $5$ transforms is able to emulate the posterior distribution in both examples with moderate training budget. 
This showcases the power of the NSF in emulating complex posterior distribution compared to MDNs. 

\subsubsection{Conditional Density for Photometric Redshift}
\label{sec:densitysimulator}

In this section we illustrate the applicability of the $\EMMD$ measure in testing  emulators for simulating redshifts associated with galaxy images. Specifically, our goal is to validate conditional density estimates of the synthetic `redshift' $X$ associated with photometric or `photo-z' galaxy images $Z$. This example is motivated by the diagnostic studies on conditional density of photometric redshifts in \cite{zhao2021diagnostics,dey2022calibrated,leroy2021md} and presents the $\EMMD$ test as a powerful method for validating photo-z probability densities.

In our experiments, we use \texttt{GalSim}, the open source toolkit for simulating realistic images of astronomical objects introduced in \cite{rowe2015galsim},  to generate $20\times 20$-pixel images of galaxies. The parameters are set as follows: 
 \begin{itemize}
 
 \item The galaxy's rotational angle $\alpha$ with respect to the $x$-axis is chosen from $\text{Unif}\left[-\pi/2,\pi/2\right]$. 
 
\item The parameter $\lambda$, which is the ratio between the minor and major axes of the projection of the elliptical galaxy, is set to be $0.75$ with a small uniform noise generated from $\text{Unif}\left[-0.1,0.1\right]$
added to it. 

\end{itemize} 
\begin{figure}[ht]
    \centering
    \begin{subfigure}{.5\textwidth}
      \centering
      \includegraphics[width=.7\linewidth]{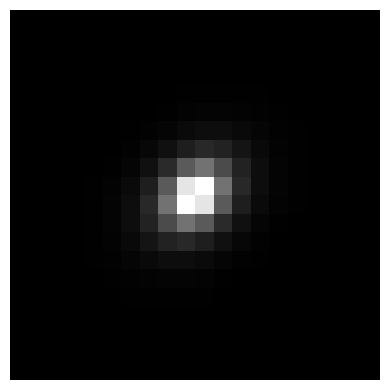}
      \caption*{\small{(a)}} 
    \end{subfigure}%
    \begin{subfigure}{.5\textwidth}
      \centering
      \includegraphics[width=.7\linewidth]{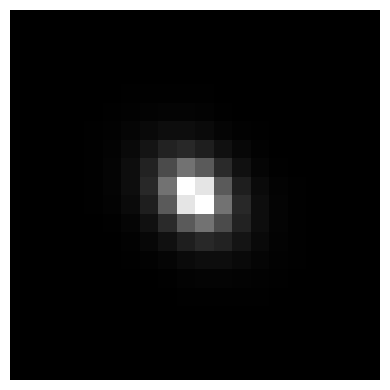}
      \caption*{(b)}
    \end{subfigure}
    \caption{ \small{ Galaxy images generated by \texttt{GalSim} with $\lambda = 0.75$ and (a) $\alpha = -\pi/4$, (b) $\alpha = \pi/4$. } }
    \label{fig:galaxyimages}
\end{figure}
These random parameters $\alpha$ and $\lambda$ are used as input to \texttt{GalSim} to generate a galaxy image $\bm Z \in [0, 1]^{20 \times 20}$. Example of galaxy images with $(\alpha, \lambda) = (-\pi/4,  0.75)$ and $(\alpha, \lambda) = (\pi/4,  0.75)$ are shown in Figure \ref{fig:galaxyimages}.
To simulate the `redshift' $X$ associated with $\bm Z$ we consider the following multimodal setup:
\begin{align}\label{eq:Xgivenalpha}
    X|\alpha\sim
    \begin{cases}
        \frac{1}{2} N \left(2,1\right) + \frac{1}{2} N \left(-2,1\right) & \text{ if }\alpha\geq \theta\\
        \frac{1}{4} N\left(3,0.0625\right) + \frac{1}{4}N\left(\mu,0.0625\right) + \frac{1}{4}N\left(-3,0.0625\right) + \frac{1}{4}N\left(-\mu,0.0625\right) & \text{otherwise,}
    \end{cases}
\end{align}
for fixed parameters $\theta\in \left\{-\pi/3, -\pi/4, -\pi/6, -\pi/12, 0, \pi/12, \pi/6, \pi/4, \pi/3\right\}$ and $\mu\in\{0, 0.5, 1, 1.5\}$. The marginal distribution of $X$ for $\theta \in  \{-\pi/3, 0, \pi/3\}$ and different values of $\mu$ are shown in Figure \ref{fig:marginalX}. 

\begin{figure}
    \centering
    \begin{subfigure}{.32\linewidth}
    \includegraphics[width=\linewidth]{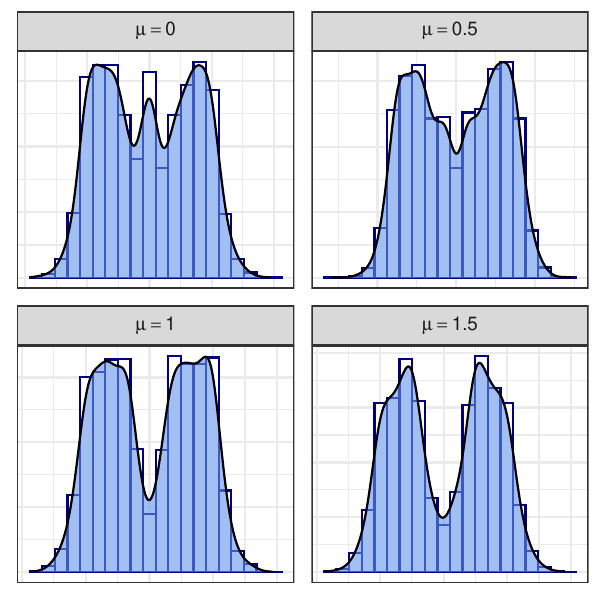}
    \caption*{(a)}
    \end{subfigure}
    \hfil
    \begin{subfigure}{.32\linewidth}
    \includegraphics[width=\linewidth]{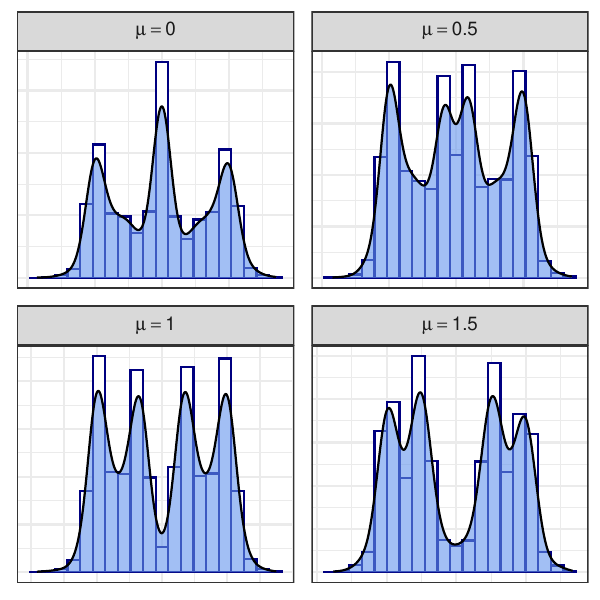}
    \caption*{(b)}
    \end{subfigure}
    \hfil
    \begin{subfigure}{.32\linewidth}
    \includegraphics[width=\linewidth]{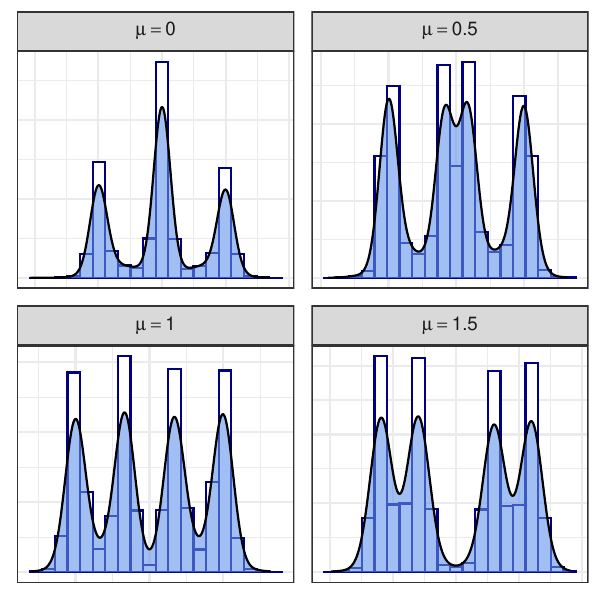}
    \caption*{(c)}
    \end{subfigure}
    \caption{ \small{ Marginal distribution of `redshift' $X$ generated according to \eqref{eq:Xgivenalpha} for (a) $\theta = -\pi/3$, (b) $\theta = 0$, and (c) $\theta = \pi/3$. The density and histogram are generated using $10000$ samples.} }
    \label{fig:marginalX}
\end{figure}
  
In our experiment we use the `redshift' value $X$ and the galaxy image $\bm Z$ as observations, and the parameter $\theta$ remains implicit and hence,  unobserved. Based on the marginal distributions and following the experiments from \cite{zhao2021diagnostics,leroy2021md} we model the conditional distribution of $X|\bm Z$ as a mixture of Gaussians. In particular,  we consider a Gaussian convolutional mixture density network (ConvMDN) \cite{d2018photometric} with $2$ mixtures. This will play the role of the emulator. To fit the ConvMDN we use a convolutional neural network with two convolutional layers interspersed with ReLU layers and having filters of size $3$ with $6$ and $12$ channels,  respectively. Post convolution, after passing through a $2\times 2$ max-pooling layer, a flattening layer transforms the 2D feature maps into a 1D vector, followed by two fully connected layers with $128$ and $10$ neurons and an output layer with $2$ neurons. The means, standard deviation, and weights of the mixtures model are determined using the output layer and the mechanism from \cite[Section 2.2]{d2018photometric}. To train the models we use the negative log-likelihood loss and optimize through the Adam optimizer  \cite{kingma2014adam} with learning rate $10^{-3}$ and parameters $\beta_1 = 0.9$ and $\beta_2 = 0.999$.

For training the ConvMDN, we use $N = 20000$ training samples generated according to the above schematics and to test the accuracy of the trained model we generate $n = 2500$ testing samples. For each $1\leq i\leq 2500$, we 
generate $(X_i, \bm Z_i)$, where $\bm Z_i$ is a galaxy image generated from \texttt{GalSim} and $X_i$ follows \eqref{eq:Xgivenalpha}, with parameters chosen as above. Also, with the galaxy image $\bm Z_i$ as input to the trained ConvMDN model, we generate a sample $Y_i$, for $1\leq i\leq 2500$. To evaluate the ConvMDN model in emulating redshifts, we now test whether or not $X|\bm Z$ (the simulator) has the same distribution as $Y |\bm Z$ (the emulator) using the asymptotic $\EMMD$ test from \eqref{eq:defphin}, based on the samples $(X_i, Y_i, \bm Z_i)$, for $1 \leq i \leq n$. We use a nominal level of $\alpha=0.05$, the number of nearest-neighbors $K= 100$, and the Gaussian kernel $\sfK_{\lambda}(x,y) = \exp\left(-(x-y)^2/\lambda^2\right)$ with bandwidth $\lambda = \text{Median}\{|X_i-Y_i|: 1\leq i\leq 2500\}$. The test is repeated $100$ times (with the same trained model) to estimate the proportion of rejections as a function of $\theta$, for different values of $\mu$.

\begin{figure}[ht]
    \centering
    \begin{minipage}{.5\textwidth}
      \centering
      \includegraphics[width=0.95\linewidth]{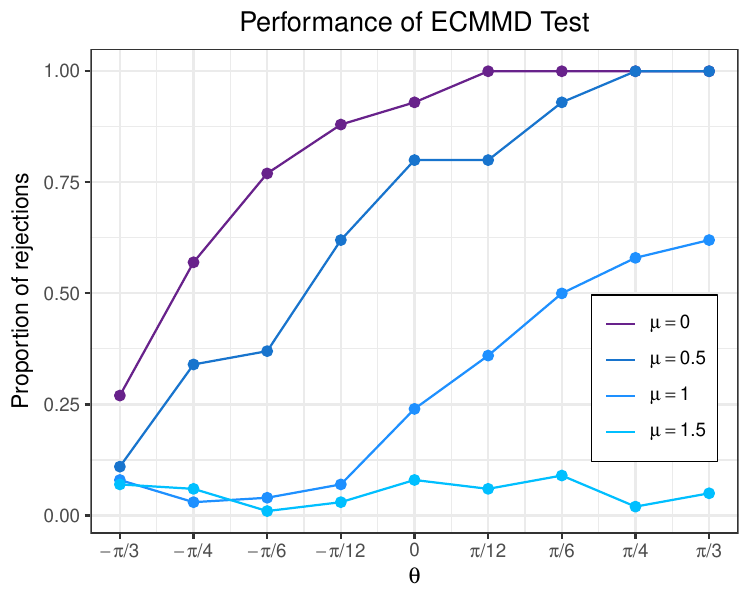} \\ 
      \small{ (a) } 
    \end{minipage}%
    \begin{minipage}{.5\textwidth}
      \centering
      \includegraphics[width=0.95\linewidth]{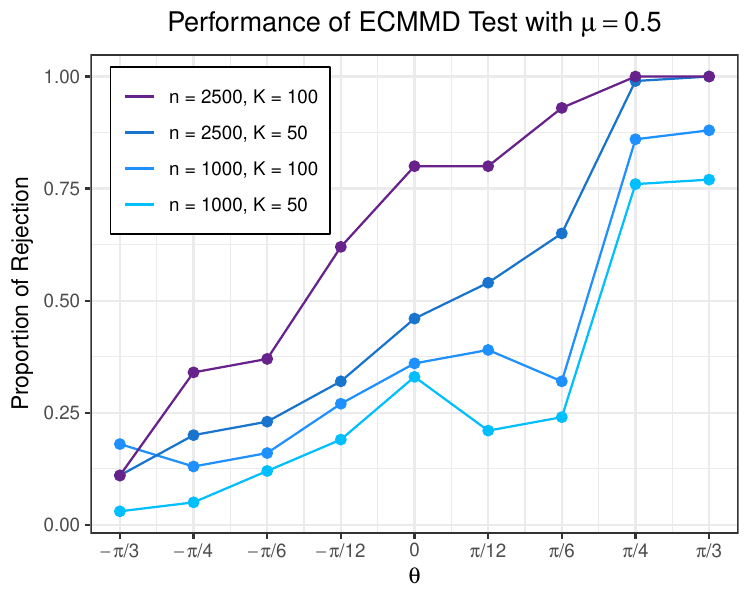} \\ 
      \small{ (b) } 
      \end{minipage}
  \caption{ \small{ Empirical power of the $\EMMD$ test for validating the Gaussian ConvMDN emulator for generating redshifts as a function of $\theta$:  (a) for $n=2500$, $K=100$, and different values of $\mu$, and (b) for $\mu=0.5$ and different values of $n, K$. } }  
      \label{fig:galaxy}
\end{figure} 
The results are shown in Figure \ref{fig:galaxy}. We observe the following: 
\begin{itemize}
\item[$1.$] the power of the $\EMMD$ test is more for larger values of $\theta$; and 
\item[$2.$] the power decreases as $\mu$ increases. 
\end{itemize}
This can be explained by looking at the marginal distribution of $X$. 
Note that the marginal distribution of $X$ has 2 modes for small values of $\theta$ (see Figure \ref{fig:marginalX}(a)), hence, the ConvMDN model provides a good fit in this case. Therefore, the $\EMMD$ test accepts the null hypothesis $P_{X|\bm Z} = P_{Y |\bm Z}$ more often when $\theta$ is small. As $\theta$ increases, the number of modes increase, which leads to more rejections, hence, the increasing trends in the power curves in Figure \ref{fig:galaxy}(a). Further, even though a $2$ component Gaussian mixture seems inadequate for larger values of $\theta$, when $\mu$ is also large, the modes tend to flatten out (see the histograms for $\mu=1.5$ in Figure \ref{fig:marginalX}), which makes it harder to distinguish the true distribution from a $2$-mixture of Gaussian. Hence, the rates of increase of the power curves in Figure \ref{fig:galaxy}(a) are slower for larger values of $\mu$. 
In Figure \ref{fig:galaxy}(b) we fix $\mu=0.5$ and show the proportion of rejections as a function of $\theta$, for different values of $n$ and $K$. As expected, the power improves as $n$ and $K$ increases and the $\EMMD$ test accurately detects the differences between the 2 distributions.
Hence, the $\EMMD$ test successfully validates whether the emulator model is a reasonable surrogate for simulating photometric redshift. This also illustrates how the proposed method can be more generally applicable as a model diagnostic tool within the realm of simulation-based inference.
 
\small

\subsection*{Acknowledgements} B. B. Bhattacharya was supported by NSF CAREER grant DMS 2046393, NSF grant DMS 2113771, and a Sloan Research Fellowship. The authors also thank Lucas Janson for helpful discussions.

\bibliographystyle{abbrvnat}
\bibliography{ref}

\normalsize 

\appendix

\section{Proof of Proposition \ref{ppn:EMMD0} and Proposition \ref{ppn:EMMDexpr}}

\subsection{Proof of Proposition \ref{ppn:EMMD0}} 
\label{sec:conditionalpf}
To begin with suppose $\P_{Z \sim P_Z}\left[ P_{X|Z} =  P_{Y|Z}\right] = 1$. , by recalling \eqref{eq:conditionalmean}, the conditional mean embeddings $\mu_{ P_{X|Z}}$ and $\mu_{ P_{X|Z}}$ are equal on a $P_Z$-almost sure set. Hence, \eqref{eq:defEMMD} implies, $\EMMD^2\left[\cF,  P_{X|Z}, P_{Y|Z}\right] = 0$. 

Now, suppose $\EMMD^2\left[\cF,  P_{X|Z}, P_{Y|Z}\right] = 0$. Then by \eqref{eq:defEMMD},
\begin{align}\label{eq:condmeanembed0}
    \left\|\mu_{ P_{X|Z}} - \mu_{ P_{Y|Z}}\right\|_{\cH} = 0\text{ almost surely } P_Z . 
\end{align} 
Since $\sfK$ is characteristic by Assumption \ref{assumption:K}, recalling the definition of characteristic kernels from \citet{sriperumbudur2008injective} the identity from \eqref{eq:condmeanembed0} implies, $P_{X|Z} =  P_{Y|Z}$, on a $P_Z$-almost sure set, which completes the proof of Proposition \ref{ppn:EMMD0}.

\subsection{Proof of Proposition \ref{ppn:EMMDexpr}}
\label{sec:Kpf}

From the definition of the conditional mean embeddings in \eqref{eq:conditionalmean}, 
\begin{align}\label{eq:EMMDequivKXY}
    \E\left[\left\|\mu_{ P_{X|Z}}(\cdot) - \mu_{ P_{Y|Z}}(\cdot)\right\|_{\cH}^2\right] = \E\left[\left\|\E\left[\sfK(X,\cdot) - \sfK(Y,\cdot)\middle|Z\right]\right\|_{\cH}^2\right].
\end{align}
Now, the construction of $(X,X',Y,Y')$ gives,
\begin{align*}
    \E\left[\left\|\E\left[\sfK(X,\cdot) - \sfK(Y,\cdot)\middle|Z\right]\right\|_{\cH}^2\right] = \E\left[\left\langle\E\left[\sfK(X,\cdot) - \sfK(Y,\cdot)\middle|Z\right],\E\left[\sfK(X',\cdot) - \sfK(Y',\cdot)\middle|Z\right]\right\rangle_{\cH}\right]
\end{align*}
Interchanging the inner product and conditional expectation gives, 
\begin{align}\label{eq:gequiv}
    \E\left[\left\|\E\left[\sfK(X,\cdot) - \sfK(Y,\cdot)\middle|Z\right]\right\|_{\cH}^2\right] 
    & = \E\left[\left\langle\sfK(X,\cdot) - \sfK(Y,\cdot), \sfK(X',\cdot) - \sfK(Y',\cdot)\right\rangle_{\cH}\right]\nonumber\\
    & = \E\left[\sfK(X,X') + \sfK(Y,Y') - \sfK(X,Y') - \sfK(X',Y)\right]. 
\end{align}
This together with the equality in \eqref{eq:EMMDequivKXY} completes the proof of Proposition \ref{ppn:EMMDexpr}.

\section{Proof of Theorem \ref{thm:consistency}}\label{sec:proofofconsistency}

To begin with, for notational convenience, define
\begin{align}\label{eq:defTn}
    T_n = T_n(\bcW_n,\sZ_n) := \EMMD^2[\sfK, \bcW_n, \sZ_n] = \frac{1}{n}\sum_{u=1}^{n}\frac{1}{K}\sum_{v \in N_{G(\sZ_n)}(u)}\sfH(\bW_{u},\bW_{v}) . 
\end{align}
To show Theorem \ref{thm:consistency} we will prove that $T_n$ converges to $\EMMD\left[\cF, P_{X|Z}, P_{Y|Z}\right]$ in $L^2$. The proof is based on the following 2 lemmas: 

\begin{lemma}\label{lemma:consistency2}
    Under the assumptions of Theorem \ref{thm:consistency},
    \begin{align*}
        \E[T_n] \ra\E\left[\left\|g_{Z_{1}}\right\|_{\cH}^2\right] , 
    \end{align*}
    where $g_{Z_1}(\cdot) := \E\left[\sfK(\cdot, X_{1}) - \sfK(\cdot, Y_{1})| Z_{1}\right]$.
\end{lemma}

\begin{lemma}\label{lemma:consistency1}
    Under the assumptions of Theorem \ref{thm:consistency},
    \begin{align}\label{eq:Tnvariance}
        \mathrm{Var}[T_n] = o\left(n^{-\frac{\delta}{4 + \delta}}\right). 
    \end{align}
\end{lemma}

The proofs of Lemma \ref{lemma:consistency2} and Lemma \ref{lemma:consistency1} are given in Appendix \ref{sec:proofofconsistency2} and Appendix \ref{sec:proofofconsistency1}, respectively. Using these results the proof of Theorem \ref{thm:consistency} can be completed as follows: Note that Lemma \ref{lemma:consistency2} and Lemma \ref{lemma:consistency1} together implies,  
\begin{align}\label{eq:TgZ}
    T_n\pto \E\left[\left\|g_{Z_{1}}\right\|_{\cH}^2\right].
\end{align}
Now, consider $(X,X',Y,Y',Z)$ such that $Z\sim P_Z$ and $(X,Y), (X',Y')$ are generated i.i.d. from $ P_{XY|Z}$. Then recalling \eqref{eq:gequiv} gives,
\begin{align}\label{eq:EgEMMD}
    \E\left[\left\|g_{Z_{1}}\right\|_{\cH}^2\right] = \E\left[\sfK(X,X') + \sfK(Y,Y') - \sfK(X,Y') - \sfK(X',Y)\right].
\end{align}
Collecting \eqref{eq:defTn}, \eqref{eq:TgZ}, \eqref{eq:EgEMMD}, and Proposition \ref{ppn:EMMDexpr}, the result in Theorem \ref{thm:consistency} follows.

\subsection{Proof of Lemma \ref{lemma:consistency2}}\label{sec:proofofconsistency2} 

To begin with, note that for all $g_{Z_u} := \E\left[\sfK(\cdot, X_{u}) - \sfK(\cdot, Y_{u})| Z_{u}\right] \in \cH$ is well defined for all $1\leq u\leq n$ by \cite[Theorem 4.1]{park2020measure}. Now, recalling the definitions of $T_n$ and $\sfH$ from \eqref{eq:defTn} and
 \eqref{eq:defH}, respectively, and by exchangeability, 
\begin{align}\label{eq:TnXY}
\E\left[T_{n}\right] 
    = & \E\left[\frac{1}{K}\sum_{u=1}^{n}\left\langle\sfK(\cdot, X_{1}) - \sfK(\cdot, Y_{1}), \sfK(\cdot, X_{u}) - \sfK(\cdot, Y_{u})\right\rangle_{\cH}\I\left\{(1,u)\in E(G(\sZ_n))\right\}\right]  \nonumber \\ 
     = & \E\left[\frac{1}{K}\sum_{u=1}^{n} \E\left[ \left\langle\sfK(\cdot, X_{1}) - \sfK(\cdot, Y_{1}), \sfK(\cdot, X_{u}) - \sfK(\cdot, Y_{u}) \right\rangle_{\cH} | \cF(\sZ_n) \right] \I\left\{(1,u)\in E(G(\sZ_n))\right\}\right] \nonumber \\ 
    = &\  \E\left[\frac{1}{K}\sum_{u=1}^{n} \langle g_{Z_{1}}, g_{Z_{u}}\rangle_{\cH} \I\left\{(1,u)\in E(G(\sZ_n))\right\}\right] , 
\end{align}
where the last step follows by interchanging the inner product and expectation. Now, suppose $N(1)$ is an index selected from the neighbors of vertex $1$ in $G(\sZ_n)$ uniformly at random. Then from \eqref{eq:TnXY}, 
\begin{align}
    \left|\E\left[T_{n}\right] - \E\left[\left\|g_{Z_{1}} \right\|_{\cH}^2\right]\right| 
    &\ \leq \E\left[\frac{1}{K}\sum_{u=1}^{n}\left|\langle g_{Z_{1}}, g_{Z_{u}}\rangle_{\cH} - \|g_{Z_{1}}\|_{\cH}^2\right|\I\left\{(1,u)\in E(G(\sZ_n))\right\}\right] \label{eq:gZH}\\ 
    &\ =\E\left[\frac{1}{K}\sum_{u=1}^{n}\left|\langle g_{Z_{1}}, g_{Z_{u}} - g_{Z_{1}}\rangle_{\cH}\right|\I\left\{(1,u)\in E(G(\sZ_n))\right\}\right]\nonumber\\
    &\ \leq \E\left[\|g_{Z_{1}} \|_{\cH}\|g_{Z_{N(1)}} - g_{Z_{1}} \|_{\cH}\right] \nonumber\\  
    &\ \leq\sqrt{\E\left[\|g_{Z_{1}}\|_{\cH}^2\right]}\sqrt{\E\left[\left\|g_{Z_{N(1)}} - g_{Z_{1}}\right\|_{\cH}^2\right]} \label{eq:Hbarlimitdiffbdd}
\end{align}
where the last 2 steps follows by the Cauchy-Schwarz inequality.  From \cite[Lemma D.2]{deb2020measuring}, 
\begin{align}\label{eq:gZN1Z1}
    \E\left[\left\|g_{Z_{N(1)}} - g_{Z_{1}}\right\|_{\cH}^4\right]\lesssim \E\left[\left\|g_{Z_{N(1)}}\right\|_{\cH}^4\right] + \E\left[\left\|g_{Z_{1}}\right\|_{\cH}^4\right]\lesssim \E\left[\left\|g_{Z_{1}}\right\|_{\cH}^4\right].
\end{align} 

Now, recalling $g_{Z_{1}} = \E\left[\sfK(\cdot, X_{1}) - \sfK(\cdot, Y_{1})| Z_{1}\right]$, observe that 
\begin{align} 
   \E \left[ \|g_{Z_{1}}\|_{\cH}^4 \right ] & = \E \left [ \left\|\E\left[\sfK(\cdot,X_{1}) - \sfK(\cdot, Y_{1})|Z_{1}\right]\right\|_{\cH}^4  \right] \nonumber \\ 
    & \leq \E \left[ \E\left[\left\|\sfK(\cdot, X_{1}) - \sfK(\cdot, Y_{1})\right\|_{\cH}^4|Z_{1}\right] \right] \tag*{ (by Jensen's inequality \cite{inequalitygeneralized}) } \nonumber \\ 
   & \lesssim \E\left[\left\|\sfK(\cdot, X_{1})\right\|_{\cH}^4\right] + \E\left[\left\|\sfK(\cdot, Y_{1})\right\|_{\cH}^4\right] \nonumber \\ 
\label{eq:gZ1}  & = \E\left[\sfK(X_{1},X_{1})^2\right] + \E\left[\sfK(Y_{1},Y_{1})^2\right] <\infty , 
\end{align} 
where the last equality follows by the reproducible property of $\cH$. Applying  \eqref{eq:gZ1} in \eqref{eq:gZN1Z1} implies, $\|g_{Z_{N(1)}} - g_{Z_{1}}\|_{\cH}^2$ is uniformly integrable. Recalling that $K = o(n/\log n)$, by \cite[Lemma D.3]{deb2020measuring}, $$\E\left[\left\|g_{Z_{N(1)}} - g_{Z_{1}}\right\|_{\cH}^2\right] = o(1).$$ The proof of Lemma \ref{lemma:consistency2} is completed by recalling \eqref{eq:Hbarlimitdiffbdd}.

\subsection{Proof of Lemma \ref{lemma:consistency1}}\label{sec:proofofconsistency1} 

The establish the result in \eqref{eq:Tnvariance} we will invoke the Efron-Stein inequality. For this, suppose $(\bW_1', Z_1'), (\bW_2', Z_2'), \ldots, (\bW_n', Z_n')$, where $\bm W_i' = (X_i', Y_i')$, for $1 \leq i \leq n$, are generated i.i.d. from the joint distribution $P_{XYZ}$, independent of the data $(\bW_{1},Z_{1}),\ldots, (\bW_{n},Z_{n})$. Define a new collection 
\begin{align}\label{eq:WZs}
\bcW_{n, s} := \left(\bcW_n \setminus\{\bW_{s}\}\right)\bigcup\{ \bW_{s}'\}\text{ and }\cZ_{n,s} := \left(\sZ_n\setminus\{Z_{s}\}\right)\bigcup\{Z_{s}'\} , 
\end{align} 
where we replace the $s$-th data point with an independent copy. (Recall that $\bcW_n:= \{\bW_i = (X_i,Y_i):1\leq i\leq n \}$ and $\sZ_n := \{Z_1,Z_2,\ldots,Z_n\}$.) The Efron-Stein inequality \cite{efron1981variance} (see also \cite[Chapter 4]{concentration}) then implies, 
\begin{align}\label{eq:EfromStein}
    \mathrm{Var}[T_n] & \leq \sum_{s=1}^{n}\E\left[\left(T_n - T_{n, s}\right)^2\right] , 
\end{align}
where, for $1 \leq s \leq n$, $T_{n,s} := T_n(\bcW_{n,s},\cZ_{n,s})$ is the test statistic as defined in \eqref{eq:defTn} with $\bcW_{n}$ and $\cZ_{n}$ replaced by $\bcW_{n,s}$ and $\cZ_{n,s}$, respectively. 

We now proceed to bound the individual terms in the sum in \eqref{eq:EfromStein}. To this end, recall that $G(\sZ_n)$ is the directed $K$-NN graph constructed using $\cZ_{n}$. For $1 \leq s \leq n$, denote by $G(\sZ_{n, s})$ the directed $K$-NN graph constructed using $\cZ_{n, s}$. We will now decompose $T_n-T_{n,s}$ by filtering out the common terms. For this, we need the following definitions. Throughout we fix $1 \leq s \leq n$. 

\begin{itemize}

\item Denote by $\cA_s$ the set of all ordered pairs $(u, v)$, for $1 \leq u \ne v \leq n$, such that either $u=s$ or $v=s$.

\item Now, define
\begin{align}\label{eq:defCns}
    \cC_{n, s} := \left(E(G(\sZ_n))\bigcap E(G(\sZ_{n, s}))\right)\setminus \cA_{s} , 
\end{align}
which is the collection of edges that are common between $G(\sZ_n)$ and $G(\sZ_{n, s})$ such that none of the end points of the edges are the vertex $s$. 

\item Also, for $1 \leq v \leq n$, define
\begin{align*}
    E_{G(\sZ_n)}(v):= \left\{(v,u): u\in N_{G(\sZ_n)}(v)\right\}, 
\end{align*} 
to be the collection of directed edges in $G(\sZ_n)$ with $Z_v$ as the starting point. 

\item Finally, let 
\begin{align*}
    V_{G(\sZ_n)}:= \left\{v\in [n]: E_{G(\sZ_n)}(v)\bigcap\left(E(G(\sZ_n))\bigcap\left(\cC_{n, s} \right)^{c}\right)\neq \emptyset \right\}
\end{align*}
be the indices which have at least one incident edge in $G(\sZ_n)$ that belongs to $\cC_{n,s}^c$. Similarly, for $G(\sZ_{n, s})$, consider 
\begin{align*}
    V_{G(\sZ_{n, s})}:= \left\{v\in [n]: E_{G(\sZ_{n, s})}(v)\bigcap\left(E(G(\sZ_{n, s}))\bigcap\left(\cC_{n,s}\right)^{c}\right)\neq \emptyset \right\}.
\end{align*}
\end{itemize}

We now rewrite the difference between $T_n$ and $T_{n,s}$ as follows: 
\begin{align*}
   T_n - T_{n, s} = \frac{1}{nK}\left[\sum_{u=1}^{n}\sum_{v \in N_{G(\sZ_n)}(u)}\sfH(\bW_{u}, \bW_{v}) - \sum_{u=1}^{n}\sum_{v\in N_{G(\sZ_{n, s})}(u)}\sfH\left(\bW_{u,s},\bW_{v,s}\right)\right] , 
\end{align*} 
where 
\begin{align}\label{eq:newnotation}
    \bW_{u,s}=
    \begin{cases}
    \bW_u' & \text{if }u=s\\
    \bW_u & \text{otherwise}, 
    \end{cases} \text{ and } 
    Z_{u,s}=
    \begin{cases}
    Z_u' & \text{if }u=s\\
    Z_u & \text{otherwise}. 
    \end{cases}
\end{align}
Note that the edges in $\cC_{n,s}$ (recall \eqref{eq:defCns}) are not effected when $Z_{s}$ is replaced by $Z_{s}'$. Thus, 
\begin{align}\label{eq:diff1stTn}
    & T_n - T_{n, s} \nonumber \\  
    & = \frac{1}{nK} \left[\sum_{u=1}^{n}\sum_{\substack{ v \in N_{G(\sZ_n)}(u)\\ (u, v)\in E(G(\sZ_n))\setminus\cC_{n, s} }} \sfH(\bW_{u},\bW_{v}) - \sum_{u=1}^{n} \sum_{\substack{ v \in N_{G(\sZ_{n, s})}(u) \\ (u, v) \in E(G(\sZ_{n, s}) )\setminus\cC_{n, s} } } \sfH\left(\bW_{u, s}, \bW_{v, s} \right) \right] . 
\end{align}
Next, observe that if $u \in V_{G(\sZ_n)}\bigcap V_{G(\sZ_{n, s})}$ is such that $u\neq s$, then the following hold: 
$$u\in \left(V_{G(\sZ_n)}\right)^{c} \implies E_{G(\sZ_n)}(u)\subseteq \cC_{n,s} \text{ and } u\in \left(V_{G(\sZ_{n, s})}\right)^{c} \implies E_{G(\sZ_{n, s})}(u)\subseteq \cC_{n,s}.$$ In other words, vertices in $\left(V_{G(\sZ_n)}\right)^c$ and $\left(V_{G(\sZ_{n, s})}\right)^c$ contribute all their edges to $\cC_{n,s}$ and, hence, have no contribution to the difference in \eqref{eq:diff1stTn}. Thus,\footnote{
For two sequences $\{a_n\}_{ n \geq 1 }$ and $\{b_n\}_{ n \geq 1 }$, denote $a_n \lesssim_{\square} b_n$ to mean $a_n \leq C(\square) b_n$ for $n$ large enough, where $C(\square) > 0$ is a constant that depends on the subscripted parameters. }
\begin{align}\label{eq:diff2ndTn}
    & T_n - T_{n, s} \nonumber \\ 
    & = \frac{1}{nK}\left[\sum_{u\in V_{G(\sZ_n)}}\sum_{\substack{v\in N_{G(\sZ_n)}(u)\\(u,v)\in E(G(\sZ_n))\setminus\cC_{n,s}}}\sfH(\bW_{u},\bW_{v}) - \sum_{u\in V_{G(\sZ_{n, s})}}\sum_{\substack{v\in N_{G(\sZ_{n, s})}(u)\\(u,v)\in E(G(\sZ_{n, s}))\setminus\cC_{n,s}}}\sfH\left(\bW_{u,s},\bW_{v,s}\right)\right] \nonumber \\ 
    & \leq \frac{1}{nK}\left[\left|\sum_{u\in V_{G(\sZ_n)}}\sum_{\substack{v\in N_{G(\sZ_n)}(u)\\(u,v)\in E(G(\sZ_n))\setminus\cC_{n,s}}}\sfH(\bW_{u},\bW_{v})\right| + \left|\sum_{u\in V_{G(\sZ_{n, s})}}\sum_{\substack{v\in N_{G(\sZ_{n, s})}(u)\\(u,v)\in E(G(\sZ_{n, s}))\setminus\cC_{n,s}}}\sfH\left(\bW_{u,s},\bW_{v,s}\right)\right|\right] \nonumber \\ 
      & \lesssim_{d} \frac{1}{n}\left[\max_{1\leq u\neq v\leq n}\left|\sfH(\bW_{u},\bW_v)\right| + \max_{1\leq u\neq v\leq n}\left|\sfH\left(\bW_{u,s},\bW_{v,s}\right)\right|\right] , 
\end{align} 
where the last step uses the bounds: 
\begin{align}\label{eq:ECn}
    \sum_{u\in V_{G(\sZ_n)}}\left|v\in N_{G(\sZ_n)}(u): (u,v)\in E(G(\sZ_n))\setminus\cC_{n,s}\right| 
    \leq 2\left|E(G(\sZ_n))\setminus\cC_{n,s}\right| \lesssim_d K, 
\end{align}
and, similarly, 
\begin{align}\label{eq:ECns}
    \sum_{u\in V_{G(\sZ_{n, s})}}\left|v\in N_{G(\sZ_{n, s})}(u): (u,v)\in E(G(\sZ_{n, s}))\setminus\cC_{n,s}\right|\leq 2\left|E(G(\sZ_{n, s}))\setminus\cC_{n,s}\right| \lesssim_d K .
\end{align} 
The bound in \eqref{eq:ECn} (and similarly \eqref{eq:ECns}) follows by noting that $E(G(\sZ_n))\setminus\cC_{n,s}$ contains the edges in $G(\sZ_n)$ which are either absent in $G(\sZ_{n,s})$ or have $s$ as one of the endpoints; and from the fact that the  maximum total degree (sum of the out-degree and the in-degree) in a $K$-NN graph is bounded by $c(d) K$, for some constant $c(d) > 0$ depending on the dimension $d$ (see, for example, \cite[Lemma 1]{jaffe2020randomized}). 

Combining \eqref{eq:EfromStein} and the bound from \eqref{eq:diff2ndTn} gives, 
\begin{align}
    \mathrm{Var}[T_n] 
    & \leq \sum_{s=1}^{n}\E\left[\left(T_n - T_{n, s}\right)^2\right]\lesssim_{d}\frac{1}{n}\E\left[\left(\max_{1\leq u\neq v\leq n}\left|\sfH(\bW_u,\bW_v)\right|\right)^2\right] . 
    \label{eq:EfronSteinbd}
\end{align}
Now recalling the definition of $\sfH$ from \eqref{eq:defH}, note that,
\begin{align*}
    \sfH(\bw,\bw') = \left\langle\sfK(x,\cdot) - \sfK(y,\cdot),\sfK(x',\cdot) - \sfK(y,\cdot)\right\rangle_{\cH}.
\end{align*}
Hence, by the Cauchy-Schwartz inequality,
\begin{align*}
    \mathrm{Var}[T_n]  \lesssim_{d} \frac{1}{n}\E\left[\left(\max_{1\leq u\leq n}\left\|\sfK(\cdot, X_{u}) - \sfK(\cdot, Y_{u})\right\|_{\cH}^2\right)^2\right]
\end{align*}
Finally, using the triangle inequality and the reproducing property of $\cH$,
\begin{align}\label{eq:TnETndiffbdd}
    \mathrm{Var}[T_n]
    &\lesssim_{d} \frac{1}{n}\E\left[\max_{1\leq u\leq n}\left\|\sfK(\cdot, X_{u})\right\|_{\cH}^4 + \max_{1\leq u\leq n}\left\|\sfK(\cdot, Y_{u})\right\|_{\cH}^4\right]\nonumber\\
    &\lesssim_{d}\frac{1}{n}\E\left[\max_{1\leq u\leq n}\sfK(X_{u},X_{u})^2 + \max_{1\leq u\leq n}\sfK(Y_{u},Y_{u})^2\right]
\end{align}
Now, for $\varepsilon > 0$, define $\varepsilon_n := \varepsilon n^{\frac{4}{4+\delta}}$, and note that 
\begin{align}
& \E\left[\max_{1\leq u\leq n}\sfK(X_{u},X_{u})^2 \right] \nonumber \\ 
& \leq \varepsilon_n + \E\left[\max_{1\leq u\leq n}\sfK(X_{u},X_{u})^2 \bm 1 \left\{ \max_{1\leq u\leq n} \sfK(X_{u},X_{u})^2 > \varepsilon_n \right\} \right] \nonumber \\ 
& \leq \varepsilon_n + \int_{\varepsilon_n}^{\infty} \P\left(\max_{1\leq u\leq n}\sfK(X_{u},X_{u})^2 \geq t \right) \mathrm dt \nonumber \\ 
& \leq \varepsilon_n + n \int_{\varepsilon_n}^{\infty} \P( \sfK(X_{1},X_{1})^2 \geq t ) \mathrm dt \tag*{ (by union bound) } \nonumber \\ 
& \leq \varepsilon_n + n \int_{\varepsilon_n}^{\infty} \frac{\E[\sfK(X_1,X_1)^{2+\delta}]}{t^{1+ \frac{\delta}{2}}} \tag*{ (by Markov's inequality) } 
\nonumber \\ 
&  \lesssim \varepsilon_n + n \int_{\varepsilon_n}^{\infty} \frac{1}{t^{1+ \frac{\delta}{2}}} \tag*{ (since $\E[\sfK(X_1,X_1)^{2+\delta}] < \infty$ by assumption) } \nonumber \\  
& \lesssim_{\delta} \varepsilon n^{\frac{4}{4+\delta}} + n^{\frac{ 4 - \delta }{4+\delta}} . \nonumber 
\end{align}
Since $\varepsilon$ is arbitrary, this shows 
\begin{align}\label{eq:maximumKX}
\E\left[\max_{1\leq u\leq n}\sfK(X_{u},X_{u})^2 \right]  = o\left(n^{\frac{4}{4+\delta}}\right),
\end{align} 
and similarly, $\E[\max_{1\leq u\leq n}\sfK(Y_{u}, Y_{u})^2 ]  = o(n^{\frac{4}{4+\delta}})$. Hence, from \eqref{eq:TnETndiffbdd} we conclude,
\begin{align}\label{eq:smalloconsistency1}
   \mathrm{Var}[T_n] \lesssim_{d} \frac{1}{n}o\left(n^{\frac{4}{4+\delta}}\right) = o\left(n^{-\frac{\delta}{4 + \delta}}\right) , 
\end{align}
completing the proof of Lemma \ref{lemma:consistency1}.\hfill $\Box$

\section{Proofs for Section \ref{sec:H0}} 
\label{sec:H0pf}

This section is organized as follows: In Section \ref{sec:varpf} we prove 
Proposition \ref{ppn:condexpH0}. Proposition \ref{prop:SnestvarH0} is proved in Section \ref{sec:estimatepf}. In Section \ref{sec:CLTpf} we prove Theorem \ref{thm:CLT} and in Section \ref{sec:consistencyH0pf} we prove Corollary \ref{cor:asymptest}. 

\subsection{Proof of Proposition \ref{ppn:condexpH0}} 
\label{sec:varpf}

In this section, all equalities involving conditional expectations hold almost surely and we omit writing the same.  With that convention, note that 
\begin{align}\label{eq:variancesecondmomentH0}
\mathrm{Var}_{\bH_{0}}[\eta_{n}|\cF(\sZ_n)] = \E_{\bH_{0}}\left[\eta_{n}^2|\cF(\sZ_n)\right],
\end{align}
 since $ \E_{\bH_{0}}\left[\eta_{n}|\cF(\sZ_n)\right] = 0$ by \eqref{eq:EH0}. Hence, 
\begin{align}
    \mathrm{Var}_{\bH_{0}}[\eta_{n}|\cF(\sZ_n)]
    &= \frac{1}{nK}\E_{\bH_{0}}\left[\left(\sum_{ u , v  \in [n] } \sfH(\bW_{u},\bW_{v})\one\left\{(u,v)\in E(G(\sZ_n))\right\}\right)^2 \bigg| \cF(\sZ_n)\right]\nonumber\\
& = S_1 + S_2 + S_3 , 
            \label{eq:condsecondmoment}
        \end{align}
        where, recalling the definition of the 
        function $f$ from Proposition \ref{ppn:condexpH0}, 
        \begin{align}\label{eq:S}
        S_1 & := 
         \frac{1}{nK}\sum_{ u, v \in [n] } \E_{\bH_{0}}[ \sfH^2(\bW_{u},\bW_{v}) |\cF(\sZ_n)] \left(\one\left\{(u,v)\in E(G(\sZ_n))\right\}+ \one\left\{(u,v),(v,u)\in E(G(\sZ_n))\right\}\right)  \nonumber \\ 
         & = \frac{1}{nK}\sum_{ u, v \in [n] } f(Z_u,Z_v)\left(\one\left\{(u,v)\in E(G(\sZ_n))\right\} + \one\left\{(u,v),(v,u)\in E(G(\sZ_n))\right\}\right) \\
            S_2 &:= 
            \frac{1}{nK} \sum_{ \substack{ u, v, u', v' \in [n] \\ |\{u, v\} \cap\{u',v' \}| = 1 }} \E_{\bH_{0}}\left[\sfH(\bW_{u},\bW_{v})\sfH(\bW_{u'},\bW_{v'})|\cF_{n}(\sZ_n)\right]\one\left\{(u, v),(u', v')\in E(G(\sZ_n))\right\}, \nonumber \\ 
             S_3 & := \frac{1}{nK} \sum_{ \substack{ u, v, u', v' \in [n] \\ |\{u, v\} \cap\{u',v' \}| = 0 }} \E_{\bH_{0}} \left[\sfH(\bW_{u},\bW_{v})\sfH(\bW_{u'},\bW_{v'})|\cF_{n}(\sZ_n)\right]\one\left\{(u, v),(u', v')\in E(G(\sZ_n))\right\} .  \nonumber 
        \end{align} 

We begin with $S_2$. For this, let $u, v, u', v' \in [n]$ be such that $|\{u, v\} \cap\{u',v' \}| = 1$. Without loss of generality, assume that $u=u'$ and $v\neq v'$. Then recallng the symmetric property of $\sfH$,
        \begin{align*}
            \E_{\bH_0}\left[\sfH(\bW_{u},\bW_{v})\sfH(\bW_{u'},\bW_{v'}) |\cF_{n}(\sZ_n)\right]
            & = \E_{\bH_0}\left[\sfH(\bW_{v},\bW_{u})\sfH(\bW_{u},\bW_{v'})|Z_{u},Z_{v},Z_{v'}\right].
        \end{align*}
        Note that, under $\bH_{0}$,
        \begin{align}\label{eq:condexp14equal}
            \E_{\bH_0}\left[\sfK(X_{v},X_{u})|X_{u}, Z_{u}, Z_{v}\right] = \E_{\bH_0}\left[\sfK(Y_{v},X_{u})|X_{u}, Z_{u}, Z_{v}\right]
        \end{align}
        and 
        \begin{align}\label{eq:condexp23equal}
            \E_{\bH_0}\left[\sfK(Y_{v},Y_{u})|Y_{u}, Z_{u}, Z_{v}\right] = \E_{\bH_0}\left[\sfK(X_{v},Y_{u})|Y_{u}, Z_{u}, Z_{v}\right].
        \end{align}
Hence, 
        \begin{align}\label{eq:condexpfull0}
             & \E_{\bH_0} \left[\sfH(\bW_{v},\bW_{u})|\bW_{u},Z_{u},Z_{v},Z_{v'}\right] \nonumber \\ 
            & =\E_{\bH_0}\left[\sfK(X_{v},X_{u})|X_{u}, Z_{u}, Z_{v}\right] + \E_{\bH_0}\left[\sfK(Y_{v},Y_{u})|Y_{u}, Z_{u}, Z_{v}\right] \nonumber \\ 
            & \hspace{0.85in} - \E_{\bH_0}\left[\sfK(X_{v},Y_{u})|Y_{u}, Z_{u}, Z_{v}\right] - \E_{\bH_0}\left[\sfK(Y_{v},X_{u})|X_{u}, Z_{u}, Z_{v}\right] . 
        \end{align}
                By \eqref{eq:condexpfull0} and using the tower property of expectation we have,
        \begin{align*}
            \E_{\bH_0}\left[\sfH(\bW_{u},\bW_{v})\sfH(\bW_{u'},\bW_{v'})|\cF(\sZ_n)\right] = 0.
        \end{align*}
       This implies, $S_2 = 0$.

        Next, we consider $S_3$. For this, let $u, v, u', v' \in [n]$ be such that $|\{u, v\} \cap\{u',v' \}| = 0$. Then,
        \begin{align*}
            \E_{\bH_0}\left[\sfH(\bW_{u},\bW_{v})\sfH(\bW_{u'},\bW_{v'})|\cF(\sZ_n)\right] 
            & = \E_{\bH_0}\left[\sfH(\bW_{u},\bW_{v})\sfH(\bW_{u'},\bW_{v'})| Z_{u},Z_{v},Z_{u'},Z_{v'}\right]\\
            & = \E_{\bH_0}\left[\sfH(\bW_{u},\bW_{v})| Z_{u},Z_{v}\right]\E_{\bH_0}\left[\sfH(\bW_{u'},\bW_{v'})| Z_{u'},Z_{v'}\right]
        \end{align*}
        Now under $\bH_{0}$, recalling \eqref{eq:condexpij0} we have,
        \begin{align*}
            \E_{\bH_0}\left[\sfH(\bW_{u},\bW_{v})\sfH(\bW_{u'},\bW_{v'})|\cF_{n}(\sZ_n)\right] = 0 . 
        \end{align*}
       This implies, $S_3 = 0$. 
       
       The proof of Proposition \ref{ppn:condexpH0} is completed by replacing the terms in the RHS of \eqref{eq:condsecondmoment} with $S_1$ from \eqref{eq:S}, $S_2=0$, and $S_3 = 0$. \hfill $\Box$

\subsection{Proof of Proposition \ref{prop:SnestvarH0}}
\label{sec:estimatepf}

The proof of Proposition \ref{prop:SnestvarH0} in based on the following 2 lemmas. First, we show that $\hat{\sigma}_n^2$ is close to $\Var_{\bH_0}[\eta_n|\cF(\sZ_n)]$ in $L^2$ (see Appendix \ref{sec:proofofSn2VarL2} for the proof). 

\begin{lemma}\label{lemma:Sn2VarL2}
    Under the assumptions of Proposition \ref{prop:SnestvarH0},
    \begin{align*}
        \E_{\bH_0}\left[\left(\hat{\sigma}_n^2 - \Var_{\bH_0}[\eta_n|\cF(\sZ_n)]\right)^2\right] = o\left(n^{-\frac{\delta}{8 + \delta}}\right) . 
    \end{align*}
\end{lemma}

Note that the above result shows that the difference of $\hat{\sigma}_n^2$ and $\Var_{\bH_0}[\eta_n|\cF(\sZ_n)]$ are close in $L^2$. To translate this into an approximation in the ratio form as in \eqref{eq:Snapproxvar} we show in the following lemma $\Var_{\bH_0}[\eta_n|\cF(\sZ_n)]$ is asymptotically bounded away from $0$ (see Appendix \ref{eq:varianceprobabilityH0} for the proof).

\begin{lemma}\label{lemma:nepscondvarinfty}
    Under the assumptions of Proposition \ref{prop:SnestvarH0}, for all $\vep>0$,
    \begin{align*}
        \P_{\bH_0}\left[n^\vep\Var_{\bH_0}[\eta_n|\cF(\sZ_n)]> t\right]\ra 1.
    \end{align*}
    for all $t>0$.
\end{lemma}

To complete the proof of Proposition \ref{prop:SnestvarH0} using the above lemmas, first define $\gamma = \frac{\delta}{8 + \delta}$. Note that by Lemma \ref{lemma:nepscondvarinfty} the probability $\P_{\bH_0}\left[n^{\frac{\gamma}{4}}\Var_{\bH_0}\left[\eta_{n}|\cF_n(\cZ_{n})\right]\leq 1\right] = o(1)$. Hence for all $\vep>0$,
\begin{align}
    \small
    \P_{\bH_0}\bigg[\left|\frac{\hat{\sigma}_n^2}{\Var_{\bH_0}\left[\eta_{n}|\cF_n(\cZ_{n})\right]} - 1\right|
    &\geq \vep n^{-\frac{\delta}{32 + 4\delta}}\bigg]\nonumber\\
    & \leq \P_{\bH_0}\left[n^{\frac{\gamma}{4}}\left|\hat{\sigma}_n^2 - \Var_{\bH_0}\left[\eta_{n}|\cF_n(\cZ_{n})\right]\right|\geq\vep n^{-\frac{\delta}{32 + 4\delta}}\right] + o(1)\nonumber\\
    & \leq \frac{n^{\frac{\gamma}{2}}}{\vep^2 n^{-\frac{\delta}{16 + 2\delta}}}\E_{\bH_0}\left[\left(\hat{\sigma}_n^2 - \Var_{\bH_0}\left[\eta_{n}|\cF_n(\cZ_{n})\right]\right)^2\right] + o(1)\nonumber\\
    & = \frac{n^{\gamma}}{\vep^2}\E_{\bH_0}\left[\left(\hat{\sigma}_n^2 - \Var_{\bH_0}\left[\eta_{n}|\cF_n(\cZ_{n})\right]\right)^2\right] + o(1)\ra 0 , \label{eq:completesigmavar}
\end{align}
where the last step uses Lemma \ref{lemma:Sn2VarL2}. This completes the proof of Proposition \ref{prop:SnestvarH0}.

\subsubsection{Proof of Lemma \ref{lemma:Sn2VarL2}}\label{sec:proofofSn2VarL2}

Recalling the definition of $\hat{\sigma}_n^2$ from \eqref{eq:defSn} 
$$\hat{\sigma}_n^2 = S_{n} + \tilde{S}_{n},$$
where,
\begin{align}\label{eq:defSn1}
    S_{n} := \frac{1}{nK}\sum_{u=1}^{n}\sum_{v=1}^{n}\sfH^{2}\left(\bW_{u},\bW_{v}\right)\one\left\{(u,v)\in E(G(\sZ_n))\right\} ,
\end{align}
and 
\begin{align}\label{eq:defSn2}
    \tilde{S}_{n} := \tilde{S}_{n}(\bcW_{n},\cZ_{n}) = \frac{1}{nK}\sum_{u=1}^{n}\sum_{v=1}^{n}\sfH^{2}\left(\bW_{u},\bW_{v}\right)\one\left\{(u,v),(v,u)\in E(G(\sZ_n))\right\}.
\end{align}
Similarly, recalling the definition of $\Var_{\bH_0}[\eta_n|\cF(\sZ_n)]$ from \eqref{eq:condeta2exp} define 
$$\Var_{\bH_0}[\eta_n|\cF(\sZ_n)] = V_{n} + \tilde{V}_{n},$$ 
where 
\begin{align}\label{eq:Vn}
    V_{n} := \frac{1}{nK}\sum_{u=1}^{n}\sum_{v = 1}^{n}f(Z_u,Z_v)\one\left\{(u,v)\in E(G(\sZ_n))\right\},
\end{align}
and 
\begin{align}\label{eq:Vnuv}
    \tilde{V}_{n} := \frac{1}{nK}\sum_{u=1}^{n}\sum_{v = 1}^{n}f(Z_u,Z_v)\one\left\{(u,v),(v,u)\in E(G(\sZ_n))\right\}.
\end{align}
Here, recall that
\begin{align}\label{eq:deff2}
    f(Z,Z')= \E_{\bH_0}\left[\sfH^2(\bW,\bW') |Z,Z'\right] , 
\end{align} 
where $\bW = (X,Y),\bW' = (X',Y')$ and $(X,Y,Z)$, $(X',Y',Z')$ are sampled independently from $ P_{XYZ}$.

To complete the proof of Lemma \ref{lemma:Sn2VarL2} it is enough to show that,
\begin{align}\label{eq:SnVndiffconvg}
    \E_{\bH_0}\left[(S_{n} - V_{n})^2\right] = o\left(n^{-\frac{\delta}{8 + \delta}}\right) \text{ and } \E_{\bH_0}\left[(\tilde{S}_{n} - \tilde{V}_{n})^2\right]  = o\left(n^{-\frac{\delta}{8 + \delta}}\right).
\end{align}
We will show this in Lemma \ref{lemma:V1S1close} and Lemma \ref{lemma:V2S2close}, respectively.  
For notational convenience, as before, we define $\gamma = \frac{\delta}{8 + \delta}$.

\begin{lemma}\label{lemma:V1S1close}
    Suppose the assumptions of Proposition \ref{prop:SnestvarH0} hold. Let $S_n$ and $V_n$ be as defined in \eqref{eq:defSn1} and \eqref{eq:Vn}, respectively. Then 
    \begin{align}\label{eq:SnVn}
        \E_{\bH_0}\left[(S_{n} - V_{n})^2\right] = o\left(n^{-\gamma}\right) . 
    \end{align}
\end{lemma}

\begin{proof}
The proof proceeds by showing that both $S_{n}$ and $V_{n}$ concentrate around $\E_{\bH_0}[S_{n}]$ in $L^2$, that is, 
\begin{align}\label{eq:S_nV_npf}
\E_{\bH_0}\left[\left(S_{n} - \E[S_{n}]\right)^2\right] = o\left(n^{-\gamma}\right)  \text{ and } \E_{\bH_0}\left[\left(V_{n} - \E[S_{n}]\right)^2\right] = o\left(n^{-\gamma}\right) . 
\end{align}
From \eqref{eq:S_nV_npf} the result in \eqref{eq:SnVn} is immediate. 

First, we show that $S_{n}$ concentrates around $\E_{\bH_0}[S_{n}]$ in $L^2$. To begin with, notice that $S_{n}$ (recall \eqref{eq:defSn1}) has the exactly the same form as $T_n$ (that is, $\emmdstat$) in \eqref{eq:defTn} with $\sfH$ replaced by $\sfH^2$. 
Hence, following the proof of \eqref{eq:EfronSteinbd} in Section \ref{sec:proofofconsistency1} with $\sfH$ replaced by $\sfH^2$ gives,
\begin{align}\label{eq:centerSnbdd}
    \E_{\bH_0}\left[\left(S_{n} - \E_{\bH_0}\left[S_{n}\right]\right)^2\right]\lesssim_{d}\frac{1}{n}\E_{\bH_0}\left[\left(\max_{1\leq u\neq v\leq n}\sfH^2(\bW_{u},\bW_{v})\right)^2\right] . 
\end{align}
Now, recalling \eqref{eq:defH}, $\sfH$ can be rewritten as: 
\begin{align}\label{eq:Hasinprod}
    \sfH(\bw,\bw') = \left\langle\sfK(x,\cdot) - \sfK(y,\cdot),\sfK(x',\cdot) - \sfK(y,\cdot)\right\rangle.
\end{align}
Hence, by Cauchy Schwarz inequality and the bound $\|a + b\|_{\cH}^m\leq 2^{m-1}\left(\|a\|_{\cH}^m + \|b\|_{\cH}^{m}\right)$, 
\begin{align}
    \E_{\bH_0}\left[\left(\max_{1\leq u\neq v\leq n}\sfH^2(\bW_{u},\bW_{v})\right)^2\right] & \leq \E_{\bH_0}\left[\left(\max_{1\leq u\leq n}\left\|\sfK(\cdot, X_{u}) - \sfK(\cdot, Y_{u})\right\|_{\cH}^4\right)^2\right] \nonumber \\ 
& \lesssim \E_{\bH_0}\left[\max_{1\leq u\leq n}\|\sfK(\cdot,X_{u})\|_{\cH}^8 + \max_{1\leq u\leq n}\|\sfK(\cdot,Y_{u})\|_{\cH}^8\right]\nonumber\\
    & = \E_{\bH_0}\left[\max_{1\leq u\leq n}\sfK(X_{u},X_{u})^4 + \max_{1\leq u\leq n}\sfK(Y_{u},Y_{u})^4\right] , \label{eq:4thexpbd}
\end{align} 
where the last equality follows from the reproducing property of $\sfK$. Now, since by assumption $\E_{\bH_0}\left[\sfK(X_{1},X_{1})^{4 + \delta}\right]<\infty$ and $\E_{\bH_0}\left[\sfK(Y_{1},Y_{1})^{4 + \delta}\right]<\infty$, for some $\delta>0$, by arguments similar to \eqref{eq:maximumKX} it follows that 
\begin{align*}
\E_{\bH_0}\left[\max_{1\leq u\leq n}\sfK(X_{u},X_{u})^4 \right]  = o\left(n^{\frac{8}{8+\delta}}\right) \text{ and } \E_{\bH_0}\left[\max_{1\leq u\leq n}\sfK(Y_{u}, Y_{u})^4 \right]  = o\left(n^{\frac{8}{8+\delta}}\right) . 
\end{align*}
Hence, combining the above with \eqref{eq:centerSnbdd}  and \eqref{eq:4thexpbd} we conclude,
\begin{align}\label{eq:onconvg}
    \E_{\bH_0}\left[\left(S_{n} - \E[S_{n}]\right)^2\right]\lesssim_{d}\frac{1}{n}\E_{\bH_0}\left[\left(\max_{1\leq u\neq v\leq n}\sfH^2(\bW_{u},\bW_{v})\right)^2\right] = o\left(n^{-\gamma}\right) ,
\end{align}
establishing the concentration of $S_n$ as in \eqref{eq:S_nV_npf}.

Next, we proceed to show that $V_{n}$ (recall \eqref{eq:Vn}) also concentrates around $\E[S_{n}]$. By a conditional expectation argument and recalling \eqref{eq:deff2} it is easy to see that,
\begin{align}\label{eq:VnSnZ}
    \E_{\bH_0}[V_{n}] = \E_{\bH_0}[S_{n}] . 
\end{align}
Once again, following the proof of \eqref{eq:EfronSteinbd} in  Section \ref{sec:proofofconsistency1} gives, 
\begin{align*}
    \E_{\bH_0}\left[\left(V_{n} - \E_{\bH_0}[S_{n}]\right)^2\right] = \E_{\bH_0}\left[\left(V_{n} - \E_{\bH_0}\left[V_{n}\right]\right)^2\right]
    & \lesssim_{d} \frac{1}{n}\E_{\bH_0}\left[\left(\max_{1\leq u\neq v\leq n}f(Z_{u},Z_{v})\right)^2\right].
\end{align*}
Recalling the definition of $f$ from \eqref{eq:deff2} and applying Jensen's inequality shows, 
\begin{align}\label{eq:f22bound}
    \E_{\bH_0}\left[\left(\max_{1\leq u\neq v\leq n}f(Z_{u},Z_{v})\right)^2\right] & \leq \E_{\bH_0}\left[\E_{\bH_0}\left[\max_{1\leq u\neq v\leq n}\sfH(\bW_{u},\bW_{v})^4|\cF(\sZ_n)\right]\right] \nonumber \\ 
    & = \E_{\bH_0}\left[\max_{1\leq u\neq v\leq n}\sfH^{4}(\bW_{u},\bW_{v})\right] . 
\end{align}
Finally, following arguments as in \eqref{eq:4thexpbd} and \eqref{eq:onconvg} we conclude,
\begin{align}\label{eq:Vn1ESn1close}
    \E_{\bH_0}\left[\left(V_{n} - \E_{\bH_0}[S_{n}]\right)^2\right] = o(n^{-\gamma}) . 
\end{align}
This proves the concentration of $V_n$ as in \eqref{eq:S_nV_npf} and completes the proof of Lemma \ref{lemma:V1S1close}. 
\end{proof}

Now, we proceed to show the $\tilde{S}_{n}$ and $\tilde{V}_{n}$ are close in $L^2$ as in stated in \eqref{eq:SnVndiffconvg}. The proof is similar is Lemma \ref{lemma:consistency1}, so we omit some details. 

\begin{lemma}\label{lemma:V2S2close}
    Under the assumptions of Proposition \ref{prop:SnestvarH0}, $S_n$ and $V_n$ as defined in \eqref{eq:defSn2} and \eqref{eq:Vnuv}, respectively, 
    \begin{align*}
        \E_{\bH_0}\left[(\tilde{S}_{n} - \tilde{V}_{n})^2\right] = o\left(n^{-\gamma}\right)
    \end{align*}
\end{lemma}

\begin{proof}
As in the previous lemma, the proof involves showing that both $\tilde{S}_{n}$ and $\tilde{V}_{n}$ are close to $\E_{\bH_0}[\tilde{S}_{n}]$ in $L^2$. 
We begin with by analyzing $\tilde{S}_{n}$. As in Lemma \ref{lemma:consistency1}, we will apply the Efron-Stein inequality. For this, suppose $(\bW_1', Z_1'), (\bW_2', Z_2'), \ldots, (\bW_n', Z_n')$, where $\bm W_i' = (X_i', Y_i')$, for $1 \leq i \leq n$, are generated i.i.d. from the joint distribution $P_{XYZ}$, independent of the data $(\bW_{1},Z_{1}),\ldots, (\bW_{n},Z_{n})$. As in \eqref{eq:WZs},  define a new collection 
\begin{align*}
\bcW_{n, s} := \left(\bcW_n \setminus\{\bW_{s}\}\right)\bigcup\{ \bW_{s}'\}\text{ and }\cZ_{n,s} := \left(\sZ_n\setminus\{Z_{s}\}\right)\bigcup\{Z_{s}'\} , 
\end{align*} 
where we replace the $s$-th data point with an independent copy. (Recall that $\bcW_n:= \{\bW_i = (X_i,Y_i):1\leq i\leq n \}$ and $\sZ_n := \{Z_1,Z_2,\ldots,Z_n\}$.) Now, recalling \eqref{eq:defSn2}, consider,
    \begin{align*}
        \tilde{S}_{n, s} = \tilde{S}_{n}(\bcW_{n,s},\cZ_{n,s}) , 
    \end{align*}
    that is, $\tilde{S}_{n, s}$ is defined as in \eqref{eq:defSn2} evaluated using the samples $\bcW_{n,s}$ and $\cZ_{n,s}$ (instead of $\bcW_{n}$ and $\cZ_{n}$). Now, by arguments similar to the proof of \eqref{eq:diff2ndTn} it can be shown that 
    \begin{align*}
        |\tilde{S}_{n} - \tilde{S}_{n, s}|\lesssim_{d}\frac{1}{n}\left[\max_{1\leq u\neq v\leq n}\left|\sfH^2(\bW_{u},\bW_v)\right| + \max_{1\leq u\neq v\leq n}\left|\sfH^2\left(\bW_{u,s},\bW_{v,s}\right)\right|\right] , 
    \end{align*} 
    where $\bW_{u,s}$ is defined in \eqref{eq:newnotation}. The Efron-Stein inequality and \eqref{eq:onconvg} implies, 
        \begin{align}\label{eq:onvareq}
        \E_{\bH_0}\left[\left(\tilde{S}_{n} - \E_{\bH_0}[\tilde{S}_{n}]\right)\right]\lesssim_{d}\frac{1}{n}\E_{\bH_0}\left[\left(\max_{1\leq u\neq v\leq n}\sfH^2(\bW_{u},\bW_{v})\right)^2\right] = o(n^{-\gamma}) . 
    \end{align}
    This shows that $\tilde{S}_{n}$ concentrates around $\E_{\bH_0}[\tilde{S}_{n}]$. 
    
    Next, we will show that $\tilde{V}_{n}$ concentrates around $\E_{\bH_0}[\tilde{S}_{n}]$. For this, repeating the arguments in the proof of \eqref{eq:onvareq} with $\tilde{S}_{n}$ replaced by $\tilde{V}_{n}$ gives, 
    \begin{align*}
        \E_{\bH_0}\left[\left(\tilde{V}_{n} - \E_{\bH_0}\left[\tilde{V}_{n}\right]\right)^2\right]\lesssim_{d}\frac{1}{n}\E_{\bH_0}\left[\left(\max_{1\leq u\neq v\leq n}f^2(Z_u,Z_v)\right)^2\right] . 
    \end{align*}
    Further, the tower property of expectation shows that $\E_{\bH_0}[\tilde{V}_{n}] = \E_{\bH_0}[\tilde{S}_{n}]$, which implies,
    \begin{align*}
        \E_{\bH_0}\left[\left(\tilde{V}_{n} - \E_{\bH_0}\left[\tilde{S}_{n}\right]\right)^2\right]
        & \lesssim_{d}\frac{1}{n}\E_{\bH_0}\left[\left(\max_{1\leq u\neq v\leq n}f^2(Z_u,Z_v)\right)^2\right]\\
        &\leq \frac{1}{n}\E_{\bH_0}\left[\max_{1\leq u\neq v\leq n}\sfH^{4}(\bW_{u},\bW_{v})\right] = o(n^{-\gamma}) , 
    \end{align*}
    by \eqref{eq:f22bound} and \eqref{eq:onvareq}. This shows that $\tilde{V}_{n}$ concentrates around $\E_{\bH_0}[\tilde{S}_{n}]$, which completes the proof of Lemma \ref{lemma:V2S2close}. 
\end{proof}

The proof of Lemma \ref{lemma:Sn2VarL2} is now completed by recalling \eqref{eq:SnVndiffconvg} and the results in Lemma \ref{lemma:V1S1close} and Lemma \ref{lemma:V2S2close}.

\subsubsection{Proof of Lemma \ref{lemma:nepscondvarinfty}} 
\label{eq:varianceprobabilityH0}

Using \eqref{eq:variancesecondmomentH0} and recalling the expression of $\E_{\bH_{0}}\left[\eta_{n}^2\middle|\cF(\sZ_n)\right]$ from \eqref{eq:condeta2exp} notice, 
\begin{align*}
    \Var_{\bH_{0}}\left[\eta_n\middle|\cF(\sZ_n)\right] & = \E_{\bH_{0}}\left[\eta_{n}^2\middle|\cF(\sZ_n)\right] \nonumber \\ 
    & \geq \frac{1}{nK}\sum_{u=1}^{n}\sum_{v = 1}^{n}\E_{\bH_0}\left[\sfH^2(\bW_u,\bW_v)\middle|\cF(\sZ_n)\right]\one\left\{(u,v)\in E(G(\sZ_n))\right\} . 
\end{align*}
Therefore, to complete the proof of Lemma \ref{lemma:nepscondvarinfty} it is now enough to show that there exists $c>0$ such that,
\begin{align}\label{eq:EH2condpositive}
    \frac{1}{nK}\sum_{u=1}^{n}\sum_{v = 1}^{n}\E_{\bH_0}\left[\sfH^2(\bW_u,\bW_v)\middle|\cF(\sZ_n)\right]\one\left\{(u,v)\in E(G(\sZ_n))\right\}\pto c . 
\end{align} 
Recalling \eqref{eq:VnSnZ} and \eqref{eq:Vn1ESn1close}, to prove \eqref{eq:EH2condpositive} it suffices to show that 
    \begin{align*}
        \E_{\bH_0}\left[\frac{1}{nK}\sum_{u=1}^{n}\sum_{v=1}^{n}\sfH^{2}\left(\bW_{u},\bW_{v}\right)\one\left\{(u,v)\in E(G(\sZ_n))\right\}\right]\ra c. 
    \end{align*}
The proof is now completed by invoking Lemma \ref{lemma:condH2convg} and Lemma \ref{lemma:varlimitnon0}. 

\subsection{Proof of Theorem \ref{thm:CLT}}
\label{sec:CLTpf}

To begin with define,
    \begin{align}\label{eq:defVu}
        V_{u} := \frac{1}{\sqrt{nK}}\sum_{v\in N_{G(\sZ_n)}(u)}\sfH\left(\bW_{u},\bW_{v}\right), \text{ for all }1\leq u\leq n.
    \end{align}
    Then $\eta_n$ (recall \eqref{eq:defetan}) can be written as: 
    \begin{align*}
        \eta_{n} = \sum_{u=1}^{n}V_{u} . 
    \end{align*} 
    We will prove the CLT of  $\eta_{n}$ using Stein's method based on dependency graphs \cite{chen2004normal}. For this, we need to construct a dependency graph for the random variables $\{V_{u}:1\leq u\leq n\}$. To this end, denote by $G^*(\sZ_n)$ the undirected simple graph obtained from the the $K$-NN graph $G(\sZ_n)$, that is, 
    we remove the directions from the edges and if for a pair of vertices there are directed edges in both directions, we keep only an undirected edge between them. Then we can construct a dependency graph $\cG_n = (V(\cG_n), E(\cG_n))$ as follow: $V(\cG_n) = \{V_{u}:1\leq u\leq n\} $ and there is an edge between $V_u$ and $V_v$, for $1 \leq u \ne v \leq n$, if and only if there exists a path of length $\leq 2$ between $u$ and $v$ in the graph $G^*(\sZ_n)$. Note that if $A, B \in [n]$ is such that there are no edges from the set $\{V_u\}_{u \in A}$ to the set $\{V_u\}_{u \in B}$ in the graph  $\mathcal{G}_n$ , then 2 collections $\{V_u\}_{u \in A}$ and  $\{V_u\}_{u \in B}$ are independent of each other given $\cF(\cZ_{n})$. 

Denote by $\Delta$ the maximum degree of a vertex in the dependency graph $\cG_n$. Since the maximum degree of the graph $G^*(\sZ_n)$ is bounded by $c(d) K$, for some constant $c(d) > 0$ depending on the dimension (see \cite[Lemma 1]{jaffe2020randomized}), it follows that $\Delta \leq C(d) K^2$, for some constant $C(d) > 0$. Now, applying the Stein's method error bound in \cite[Theorem 2.7]{chen2004normal}, recalling \eqref{eq:EH0}, and denoting $\sigma_n^2:=\Var_{\bH_0}\left[\eta_{n}|\cF(\cZ_{n})\right]$ we have,
    \begin{align*}
        \sup_{z\in \R}\left|\P_{\bH_0}\left[\frac{\eta_{n}}{\sigma_n}\leq z \middle|\cF(\cZ_{n})\right] - \Phi(z)\right|\lesssim_{d} \frac{K^{20}}{\sigma_n^3} \E_{\bH_0}\left[\sum_{i=1}^{n}|V_{i}|^3 \middle|\cF(\cZ_{n})\right] , 
     \end{align*}
    almost surely. Using the tower property of conditional expectation gives,     \begin{align}
        \sup_{z\in \R}\left|\P_{\bH_0}\left[\frac{\eta_{n}}{\sigma_n}\leq z\right] - \Phi(z)\right|
        & \lesssim_{d} \E_{\bH_0}\left[ \frac{K^{20}}{\sigma_n^3} \E_{\bH_0}\left[\sum_{i=1}^{n}|V_{i}|^3 \middle|\cF(\cZ_{n})\right]  \right] . \label{eq:bddonsteindiff}
    \end{align} 
Note that by Lemma \ref{lemma:nepscondvarinfty}, 
    \begin{align}
        \E_{\bH_0}\Bigg[ \frac{K^{20}}{\sigma_n^3} \E_{\bH_0}\bigg[\sum_{u=1}^{n}|V_{u}|^3\bigg|
        &\cF(\cZ_{n})\bigg] \Bigg]\nonumber\\
        & \leq \E_{\bH_0}\left[ \frac{K^{20}}{\sigma_n^3} \E_{\bH_0}\left[\sum_{u=1}^{n}|V_{u}|^3|\cF(\cZ_{n})\right] \one\left\{n^{\vep} \sigma_n^2 \geq 1\right\} \right] + o(1) \nonumber \\ 
        & \leq  K^{20}n^{\frac{3\vep-1}{2}}\E_{\bH_0}\left[\sqrt{n}\E_{\bH_0}\left[\sum_{u=1}^{n}|V_{u}|^3|\cF(\cZ_{n})\right]\right] + o(1) . 
        \label{eq:bddExpmin2}
    \end{align} 
   The following provides an upper bound on the RHS of from \eqref{eq:bddExpmin2}.

    \begin{lemma}\label{lemma:bddV3}
        Under the conditions of Theorem \ref{thm:CLT}, 
        \begin{align*}
            \E_{\bH_0}\left[\sqrt{n}\E_{\bH_0}\left[\sum_{u=1}^{n}|V_{u}|^3 \middle|\cF(\cZ_{n})\right]\right]\lesssim_{\sfK} K^{\frac{3}{2}}.
        \end{align*}
    \end{lemma}
    
    The proof of Lemma \ref{lemma:bddV3} is given in Appendix \ref{sec:Vexpectationpf}. First we complete the proof of Theorem \ref{thm:CLT} using this lemma. To this end, combining \eqref{eq:bddonsteindiff}, \eqref{eq:bddExpmin2}, and Lemma \ref{lemma:bddV3} gives, 
      \begin{align}\label{eq:etacltH0}
         \sup_{z\in \R}\left|\P_{\bH_0}\left[\frac{\eta_{n}}{\sigma_n}\leq z\right] - \Phi(z)\right| \lesssim_{d,\sfK} K^{\frac{43}{2}}n^{\frac{3\vep-1}{2}} + o(1) = o(1) , 
    \end{align} 
    where the final equality follows by choosing $\vep<\frac{1}{132}$.     
To complete the proof of Theorem \ref{thm:CLT} we have to replace $\sigma_n^2$ by $\hat{\sigma}_n^2$ in \eqref{eq:etacltH0}. This follows by \eqref{eq:completesigmavar}.    
    
    \subsubsection{Proof of Lemma \ref{lemma:bddV3}} 
    \label{sec:Vexpectationpf}
    
    Recall the definition of $V_u$, for $1 \leq u \leq n$, from \eqref{eq:defVu}. Then by H\"older's inequality, 
    \begin{align}\label{eq:bddonsumV}
        \sqrt{n}\sum_{u=1}^{n}|V_{u}|^3 & \leq \frac{1}{nK^{\frac{3}{2}}}\sum_{u=1}^{n}\left(\sum_{v\in N_{G(\sZ_n)}(u)}|\sfH(\bW_{u},\bW_{v})|\right)^3 \nonumber \\ 
        & \leq \frac{\sqrt{K}}{n}\sum_{u=1}^{n}\sum_{v\in N_{G(\sZ_n)}(u)}|\sfH(\bW_{u},\bW_{v})|^3 . 
    \end{align}
    Using the representation of $\sfH$ in \eqref{eq:Hasinprod} and the Cauchy-Schwartz inequality,
    \begin{align*}
        & \sum_{u=1}^{n}\sum_{v\in N_{G(\sZ_n)}(u)}
        |\sfH(\bW_{i},\bW_{j})|^3 \nonumber \\ 
        & \leq \sum_{u=1}^{n}\sum_{v\in N_{G(\sZ_n)}(u)}\left\|\sfK(X_{u},\cdot) - \sfK(Y_{u},\cdot)\right\|_{\cH}^{3}\left\|\sfK(X_{v},\cdot) - \sfK(Y_{v},\cdot)\right\|_{\cH}^{3} \nonumber \\ 
   & \lesssim \sum_{u=1}^{n}\sum_{v\in N_{G(\sZ_n)}(u)}\left(\left\|\sfK(X_{u},\cdot)\right\|_{\cH}^{3} + \left\|\sfK(Y_{u},\cdot)\right\|_{\cH}^{3}\right)\left(\left\|\sfK(X_{v},\cdot)\right\|_{\cH}^{3} + \left\|\sfK(Y_{v},\cdot)\right\|_{\cH}^{3}\right) , 
    \end{align*} 
    where the last step uses the inequality $\|a+b\|_{\cH}^m\leq 2^{m-1}(\|a\|_{\cH}^m + \|b\|_{\cH}^m)$, for $m \geq 1$. Finally, recalling the reproducing property of $\sfK$ gives, 
        \begin{align}\label{eq:bddonsumh3}
        & \sum_{u=1}^{n} \sum_{v\in N_{G(\sZ_n)}(u)}
        |\sfH(\bW_{i},\bW_{j})|^3\nonumber\\
        &\lesssim \sum_{u=1}^{n}\sum_{v\in N_{G(\sZ_n)}(u)}\left(\sfK(X_{u},X_{u})^{\frac{3}{2}} + \sfK(Y_{u},Y_{u})^{\frac{3}{2}}\right)\left(\sfK(X_{v},X_{v})^{\frac{3}{2}} + \sfK(Y_{v},Y_{v})^{\frac{3}{2}}\right) . 
    \end{align}
    
    Now, for $(X,Y,Z)\sim P_{XYZ}$ define,
    \begin{align}\label{eq:defbetafn}
        \beta(Z) := \E_{\bH_0}\left[\sfK(X,X)^{\frac{3}{2}} + \sfK(Y,Y)^{\frac{3}{2}} |Z\right].
    \end{align}
    By the moment assumption on $\sfK$ and from \cite[Lemma 1.13]{kallenberg1997foundations} it follows that $\beta:\cZ\ra\R$ is a well defined  measurable function. Hence, from \eqref{eq:bddonsumV} and \eqref{eq:bddonsumh3}, 
    \begin{align} 
        \E_{\bH_0}\left[\sqrt{n}\E_{\bH_0}\left[\sum_{u=1}^{n}|V_{u}|^3|\cF(\cZ_{n})\right]\right] & \lesssim\frac{\sqrt{K}}{n}\E_{\bH_0}\left[\sum_{u=1}^{n}\sum_{v\in N_{G(\sZ_n)}(u)}\beta(Z_{u})\beta(Z_{v})\right] \nonumber \\ 
        &\lesssim K^{\frac{3}{2}}\E_{\bH_0}\left[\frac{1}{nK}\sum_{u=1}^{n}\sum_{v\in N_{G(\sZ_n)}(u)}\beta(Z_{u})\beta(Z_{v})\right] \nonumber \\ 
    & = K^{\frac{3}{2}}\E_{\bH_0}\left[\frac{1}{K}\sum_{v\in N_{G(\sZ_n)}(1)}\beta(Z_{1})\beta(Z_{v})\right] \nonumber \\ 
    & = K^{\frac{3}{2}}\E_{\bH_0}\left[\beta(Z_{1})\beta(Z_{N(1)})\right] , 
        \label{eq:CLTproofuniformN1}
    \end{align} 
    where $N(1)$ is a neighbor of the vertex 1 in the graph $G(\sZ_n)$ chosen uniformly at random. By Cauchy-Schwarz inequality and \cite[Lemma D.2]{deb2020measuring}, 
    \begin{align}\label{eq:CLTproofCauchySch}
        K^{\frac{3}{2}}\E_{\bH_0}\left[\beta(Z_{1})\beta(Z_{N(1)})\right] & \leq K^{\frac{3}{2}}\left(\E_{\bH_0}\left[\beta(Z_{1})^2\right]\right)^{\frac{1}{2}}\left(\E_{\bH_0}\left[\beta\left(Z_{N(1)}\right)^2\right]\right)^{\frac{1}{2}} \nonumber \\ 
        & \lesssim K^{\frac{3}{2}}\E_{\bH_0}\left[\beta(Z_{1})^2\right] . 
    \end{align}
    Combining \eqref{eq:CLTproofuniformN1} and \eqref{eq:CLTproofCauchySch} gives, 
      \begin{align}\label{eq:1stbddbeta}
        \E_{\bH_0}\left[\sqrt{n}\E_{\bH_0}\left[\sum_{u=1}^{n}|V_{u}|^3|\cF(\cZ_{n})\right]\right]\lesssim K^{\frac{3}{2}}\E_{\bH_0}\left[\beta(Z_{1})^2\right].
    \end{align}
    To complete the proof, recalling the definition of the function $\beta$ from \eqref{eq:defbetafn} note that,
    \begin{align}\label{eq:2ndbddbeta}
        \E_{\bH_0}\left[\beta(Z_{1})^2\right] 
        & = \E_{\bH_0}\left[\left(\E_{\bH_0}\left[\sfK(X_{1},X_{1})^{\frac{3}{2}} + \sfK(Y_{1},Y_{1})^{\frac{3}{2}}|Z_{1}\right]\right)^2\right] \nonumber \\ 
        & \lesssim \E_{\bH_0}\left[\sfK(X_{1},X_{1})^3 + \sfK(Y_{1},Y_{1})^3\right]<\infty , 
    \end{align}
    by the moment assumption on $\sfK$. The proof of Lemma \ref{lemma:bddV3} is now completed by substituting the bound from \eqref{eq:2ndbddbeta} in \eqref{eq:1stbddbeta}.  \hfill $\Box$ 

\subsection{Proof of Corollary \ref{cor:asymptest}} 
\label{sec:consistencyH0pf} 

From \eqref{eq:defetan} recall,
\begin{align*}
    \eta_{n}=\sqrt{nK}\emmdstatsq , 
\end{align*} 
where $\bcW_n = \{\bW_i = (X_i,Y_i):1\leq i\leq n\}$ and $\sZ_n = \{Z_1,\ldots,Z_n\}$. By Theorem \ref{thm:consistency} we know that under $\bH_1$,
\begin{align*}
    \emmdstatsq \pto \EMMD\left[\cF, P_{X|Z}, P_{Y|Z}\right]>0 , 
\end{align*}
where the positivity of the limit follows from Proposition \ref{ppn:EMMD0}. Now, recalling $\phi_n$ from \eqref{eq:defphin}, note that to prove Corollary \ref{cor:asymptest} it is now enough to show that under $\bH_1$, $\hat{\sigma}_n = O_P(1)$. Towards this, from \eqref{eq:defSn},
\begin{align*}
    \hat{\sigma}_n^2 & = \frac{1}{nK}\sum_{u=1}^{n}\sum_{v=1}^{n}\sfH^{2}\left(\bW_{u},\bW_{v}\right)\left(\one\left\{(u,v)\in E(G(\sZ_n))\right\}+\one\left\{(u,v),(v,u)\in E(G(\sZ_n))\right\}\right) \nonumber \\ 
    & \leq \frac{2}{nK}\sum_{u=1}^{n}\sum_{v=1}^{n}\sfH^{2}\left(\bW_{u},\bW_{v}\right)\one\left\{(u,v)\in E(G(\sZ_n))\right\}. 
\end{align*} 
We know from \eqref{eq:EH2condconvg} and \eqref{eq:exphZfinite} that 
\begin{align*}
    \E_{\bH_1}\left[ \frac{1}{nK}\sum_{u=1}^{n}\sum_{v=1}^{n}\sfH^{2}\left(\bW_{u},\bW_{v}\right)\one\left\{(u,v)\in E(G(\sZ_n))\right\}\right] = O(1). 
\end{align*} 

This implies, $\hat{\sigma}_n = O_{P}(1)$ under $\bH_1$ and completes the proof of Corollary \ref{cor:asymptest}.

\section{Proofs from Section \ref{sec:finitetest} }
\label{sec:finitepf}

This section is organized as follows: In Section \ref{sec:H0hypothesispf} we prove Proposition \ref{ppn:finitetest}. Theorem \ref{thm:consistencyDn} is proved in Section \ref{sec:proofofDnconsistency} and Theorem \ref{thm:DnCLT} is  proved in Section \ref{sec:proofofDnClT}.

\subsection{Proof of Proposition \ref{ppn:finitetest}}
\label{sec:H0hypothesispf}

Recall that under $\bH_0$, $X$ and $Y$ have the same distribution conditioned on $Z$. Now, recall the finite sample test  described in Section \ref{sec:finitetest}. Observe that conditional on $\sZ_n$, the samples $(Y_1,\ldots,Y_n)$ are independent and identically distributed as the  
samples $\{X_1^{(m)},\ldots, X_n^{(m)}\}$, for all $1 \leq m\leq M+1$. 
As a result, the collection $\{ \eta_{n}^{(m)}, 1 \leq m\leq M+1 \}$ (as defined in \eqref{eq:defetanm}) are exchangeable conditioned on $\sZ_n$. Thus, recalling \eqref{eq:Mfinite},  
\begin{align*}
    \P_{\bH_0}[ \tilde{\phi}_{n, M} = 1|\cF(\sZ_n)] = \P_{\bH_0}\left[p_M\leq \alpha|\cF(\sZ_n)\right]=\frac{\lfloor (M+1)\alpha\rfloor }{M+1}\leq\alpha.
\end{align*}
Taking expectations on both sides show the result in Proposition \ref{ppn:finitetest} (1).

To prove the consistency of the test $\tilde{\phi}_{n,M}$ in Proposition \ref{ppn:finitetest} (2) we start by observing that under $\bH_1$,
\begin{align*}
    \frac{\eta_{n}^{(m)}}{\sqrt{nK}}\pto 0 \text{ for all }1\leq m\leq M\text{ and }\frac{\eta_{n}^{(M+1)}}{\sqrt{nK}}\pto \EMMD\left[\cF, P_{X|Z}, P_{Y|Z}\right]>0,
\end{align*}
where the last inequality follows from Proposition \ref{ppn:EMMD0} and the convergences follow from Theorem \ref{thm:consistency}. Consequently, $\one\bigg \{\left|\eta_{n}^{(m)}\right|\geq \left|\eta_{n}^{(M+1)}\right|\bigg\} = o_{P}(1)$, for all $1\leq m\leq M$. Hence,
\begin{align*}
    p_M = \frac{1}{M+1}\left[1 + \sum_{m=1}^{M}\one\bigg\{\left|\eta_{n}^{(m)}\right|\geq \left|\eta_{n}^{(M+1)}\right|\bigg\}\right] \pto \frac{1}{M+1}.
\end{align*}
The proof is now completed by choosing $M > \frac{1}{\alpha}-1$. 

\subsection{Proof of Theorem \ref{thm:consistencyDn}}\label{sec:proofofDnconsistency}
The proof is similar to the proof of Theorem \ref{thm:consistency} in Appendix \ref{sec:proofofconsistency}. To begin with, define
$$\tilde{\bW}_{u} = \left(X_{u}^{(1)},\ldots, X_{u}^{(M_n)}, Y_u \right),$$ for $1\leq u\leq n$. Then $D_n$ in \eqref{eq:defDn} can be expressed as: 
\begin{align}\label{eq:defDnHbar}
    D_n = \frac{1}{nK}\sum_{u=1}^{n}\sum_{v\in N_{G(\sZ_n)}(u)}\bar{\sfH}\left(\tilde{\bW}_{u},\tilde{\bW}_{v}\right),
\end{align}
where 
\begin{align}\label{eq:defbarH}
    \bar{\sfH}\left(\tilde{\bW}_{u},\tilde{\bW}_{v}\right) = \frac{1}{M_n}\sum_{m=1}^{M_n}\sfH\left(\bW_{u}^{(m)},\bW_{v}^{(m)}\right), \quad \text{ for all }1\leq u,v\leq n ,
\end{align}
and $\bW_{u}^{(m)} = (X_u^{(m)}, Y_{u})$, for $1\leq m\leq M_n$. The expression \eqref{eq:defDnHbar} shows that $D_n$  has the exact same form as the estimate $\emmdstatsq$ in \eqref{eq:estEMMD} with $\sfH$ replaced by $\bar{\sfH}$. Hence, following the combinatorial arguments from Section \ref{sec:proofofconsistency1}, in particular the proof of \eqref{eq:EfronSteinbd} with $T_n$ replaced by $D_n$ shows,
\begin{align*}
    \E\left[\left(D_n - \E D_n\right)^2\right]
    &\lesssim_d\frac{1}{n}\E\left[\left(\max_{1\leq u\neq v\leq n}\left|\bar{\sfH}\left(\tilde{\bW}_{u},\tilde{\bW}_{v}\right)\right|\right)^2\right]\\
    &\leq\frac{1}{nM_n}\sum_{m=1}^{M_n}\E\left[\left(\max_{1\leq u\neq v\leq n}\left|\sfH\left(\bW_{u}^{(m)},\bW_{v}^{(m)}\right)\right|\right)^2\right] , 
\end{align*}
where the last inequality follows by the definition of $\bar{\sfH}$ from \eqref{eq:defbarH} and the Cauchy-Schwarz inequality. By the sampling scheme from Section \ref{sec:Derandom}, for any $1\leq m\leq M_n$, the collection $\{(X_{u}^{(m)},Y_{u}, Z_{u} \}_{1 \leq u \leq n}$ are generated i.i.d. from $P_{XYZ}$. Then,
\begin{align*}
    \E\left[\left(D_n - \E D_n\right)^2\right]
    &\lesssim_d\frac{1}{n}\E\left[\left(\max_{1\leq u\neq v\leq n}\left|\sfH\left(\bW_{u}^{(1)}, \bW_{v}^{(1)} \right)\right|\right)^2\right].
\end{align*}
Note that the above upper bound is same as that \eqref{eq:EfronSteinbd}. Hence proceeding as in \eqref{eq:TnETndiffbdd} and \eqref{eq:smalloconsistency1} we conclude,
\begin{align*}
    \E\left[\left(D_n - \E D_n\right)^2\right]  = o\left(n^{-\frac{\delta}{4+\delta}}\right).
\end{align*}
Now, recalling \eqref{eq:EgEMMD}, to complete the proof of Theorem \ref{thm:consistencyDn} it is enough to show that,
\begin{align*}
    \E[D_n]\ra\E\left[\left\|g_{Z_1} \right\|_{\cH}^2\right] , 
\end{align*}
where $g_{Z_1}  = \E\left[\sfK(\cdot,X_1) - \sfK(\cdot,Y_1)\middle|Z_1\right]$ is defined in Lemma \ref{lemma:consistency2}.
For this, exchangeability and the definition of $D_n$ in \eqref{eq:defDnHbar} gives,
\begin{align*}
    \left|\E[D_n] - \E\left[\left\|g_{Z_1} \right\|_{\cH}^2\right]\right| = \left|\E\left[\frac{1}{K}\sum_{v=1}^{n}\bar{\sfH}\left( \tilde{\bW}_1, \tilde{\bW}_{v}\right)\one\left\{(1,v)\in E\left(\sZ_n\right)\right\}\right] - \E\left[\left\|g_{Z_1}(\cdot)\right\|_{\cH}^2\right]\right| . 
\end{align*}
Now, observe that for $1\leq v\leq n$,
\begin{align*}
    \E\left[\bar{\sfH}\left( \tilde{\bW}_{1}, \tilde{\bW}_{v}\right)\middle|Z_1,Z_v\right] = \frac{1}{M_n}\sum_{m=1}^{n}\E\left[\sfH\left(\bW_{1}^{(m)}, \bW_{v}^{(m)}\right)\middle|Z_1,Z_v\right] = \left\langle g_{Z_1}, g_{Z_v} \right\rangle_{\cH} . 
\end{align*}
where the last equality follows by observing that the collection $\{(X_u^{(m)}, Y_{u}, Z_{u})\}_{1 \leq u \leq n}$ are generated i.i.d. from $P_{XYZ}$, for any $1\leq m\leq M_n$, and the definition of $\sfH$.
Hence,
\begin{align*}
    \left|\E[D_n] - \E\left[\left\|g_{Z_1}(\cdot)\right\|_{\cH}^2\right]\right| \leq \E\left[\frac{1}{K}\sum_{u=1}^{n}\left|\langle g_{Z_{1}}, g_{Z_{u}}\rangle_{\cH} - \|g_{Z_{1}}\|_{\cH}^2\right|\I\left\{(1,u)\in E(G(\sZ_n))\right\}\right] . 
\end{align*}
Notice that the upper bound in the RHS above is same as that in  \eqref{eq:gZH}. Hence, the proof of Theorem \ref{thm:consistencyDn} is now completed by following the subsequent arguments from the proof of Lemma \ref{lemma:consistency2} in Section \ref{sec:proofofconsistency2}.

\subsection{Proof of Theorem \ref{thm:DnCLT}}\label{sec:proofofDnClT}

To begin with, observe from \eqref{eq:condexpij0} that $\E_{\bH_0}\left[D_n\middle|\cF(\sZ_n)\right] = 0$. In the following lemma we compute the conditional variance of $D_n$ under $\bH_0$. 

\begin{lemma}\label{lemma:varH0Dn}
Denote by $\tau_n^2:= \Var_{\bH_0}\left[\sqrt{nK}D_n\middle|\cF(\sZ_n)\right]$, the conditional variance of $\sqrt{nK}D_n$ given $\cF(\sZ_n)$  under $\bH_0$. Then, 
    \begin{align*}
       \tau_n^2 = \frac{1}{nK}\sum_{ 1 \leq u\neq v \leq n} \ell_n(Z_u,Z_v) \left( \one\left\{(u,v)\in E(G(\sZ_n))\right\} + \one\left\{(u,v),(v,u)\in E(G(\sZ_n))\right\} \right),
    \end{align*}
    where $\ell_n(Z_u,Z_v) := \E_{\bH_0}[\bar{\sfH}^2(\tilde{\bW}_{u},\tilde{\bW}_{v}) |Z_u,Z_v]$, with $\bar{\sfH}(\tilde{\bW}_{u},\tilde{\bW}_{v})$ defined in \eqref{eq:defbarH}, for all $1\leq u\neq v\leq n$.
\end{lemma} 

\begin{proof} 
Recalling definition of $\bar{H}$ from \eqref{eq:defbarH} gives, 
\begin{align*}
    \E_{\bH_0}\left[\bar{\sfH}\left(\tilde{\bW}_{u},\tilde{\bW}_{v}\right)\middle|\tilde{\bW}_{v},Z_u,Z_v\right] = \frac{1}{M_n}\sum_{m=1}^{M_n}\E_{\bH_0}\left[\sfH\left(\bW_{u}^{(m)},\bW_{v}^{(m)}\right)\middle|\bW_{v}^{(m)},Z_u,Z_v\right] = 0,
\end{align*}
where the last equality follows from \eqref{eq:condexp14equal}, \eqref{eq:condexp23equal}, and \eqref{eq:condexpfull0}. The proof of Lemma \ref{lemma:varH0Dn} can now be completed by repeating the proof of  Proposition \ref{ppn:condexpH0} from Appendix \ref{sec:varpf} with $\eta_n$ replaced by $\sqrt{nK}D_n$. 
\end{proof}

Next, we show that $\hat{\tau}_n^2$ as defined in \eqref{eq:estimateM} is a  consistent estimate $\tau_n^2$. The proof is given in Appendix \ref{sec:estimateMpf}.

\begin{lemma}\label{lemma:SnapproxvarDn}
    Suppose Assumption \ref{assumption:K} and Assumption \ref{assumption:Kn} hold. Moreover, assume $(X,Y,Z)\sim P_{XYZ}$ is such that $\P\left[X\neq Y\right]>0$ and $\int \sfK(x,x)^{4+\delta}\mathrm{d} P_X(x)<\infty$, $\int \sfK(x,x)^{4+\delta}\mathrm{d} P_Y(x)<\infty$, for some $\delta>0$. Then under $\bH_0$, with $K = o(n/\log n)$ and $M_n\ra\infty$ as $n\ra\infty$,
    \begin{align*}
        \left|\frac{\hat{\tau}_n^2}{\tau_n^2}-1\right| = o_P\left(n^{-\frac{\delta}{32 + 4\delta}}\right) , 
    \end{align*}
    where $\hat{\tau}_n^2$ and $\tau_n^2$ are defined in \eqref{eq:estimateM} and Lemma \ref{lemma:varH0Dn}, respectively. 
\end{lemma}

With the above lemmas, we now ready to complete with the proof of Theorem \ref{thm:DnCLT}. Once again, drawing parallels with $\eta_n$ (recall \eqref{eq:defetan}) and $\sqrt{nK}D_n$ define,
\begin{align*}
    \tilde{V}_{u} = \frac{1}{\sqrt{nK}}\sum_{v\in N_{G(\sZ_n)}(u)}\bar{\sfH}\left(\tilde{\bW}_{u},\tilde{\bW}_{v}\right),\ 1\leq u\leq n.
\end{align*}
Following the arguments from Section \ref{sec:CLTpf} with $V_u$ replaced by $\tilde{V}_{u}$, for $1\leq u\leq n$, and using \eqref{eq:epsVarDn} we get, 
\begin{align*}
    \sup_{z\in \R}\left|\P_{\bH_0}\left[\frac{\sqrt{nK}D_n}{\tau_n}\leq z\right] - \Phi(z)\right|
    & \lesssim_{d} K^{20}n^{\frac{3\vep-1}{2}}\E_{\bH_0} \left[\sqrt{n}\E_{\bH_0} \left[\sum_{u=1}^{n}| \tilde{V}_{u}|^3 \bigg|\cF(\cZ_{n})\right]\right] + o(1) , 
\end{align*}
for any $\vep<\frac{1}{132}$. 
Hence, to complete the proof of Theorem \ref{thm:DnCLT} it is enough to show that,
\begin{align}\label{eq:VZ}
    \E_{\bH_0}\left[\sqrt{n}\E_{\bH_0}\left[\sum_{u=1}^{n}| \tilde{V}_u |^3 \bigg |\cF(\cZ_{n})\right]\right]\lesssim_{\sfK}K^{\frac{3}{2}}.
\end{align}
Towards this, by \eqref{eq:bddonsumV} and H\"{o}lder's inequality observe that, 
\begin{align*}
    \sqrt{n}\sum_{u=1}^{n} | \tilde{V}_u |^3
    &\leq \frac{\sqrt{K}}{n}\sum_{u=1}^{n}\sum_{v\in N_{G(\sZ_n)}(u)}\left|\bar{\sfH}\left(\tilde{\bW}_{u},\tilde{\bW}_{v}\right)\right|^3\\
    &\leq \frac{\sqrt{K}}{n}\sum_{u=1}^{n}\sum_{v\in N_{G(\sZ_n)}(u)}\frac{1}{M_n}\sum_{m=1}^{M_n}\left|\sfH\left(\bW_{u}^{(m)},\bW_{v}^{(m)}\right)\right|^3.
\end{align*}
Then 
\begin{align*}
    & \E_{\bH_0}\left[\sqrt{n}\E_{\bH_0}\left[\sum_{u=1}^{n}| \tilde{V}_{u}|^3 \bigg| \cF(\cZ_{n})\right]\right] \nonumber \\ 
    & \leq \frac{\sqrt{K}}{nM_n}\sum_{m=1}^{M_n}\E_{\bH_0}\left[\E_{\bH_0}\left[\sum_{u=1}^{n}\sum_{v\in N_G(\sZ_n)(u)}\left|\sfH\left(\bW_{u}^{(m)},\bW_{v}^{(m)}\right)\right|^3\middle|\cF(\cZ_{n})\right]\right] . 
\end{align*}
Now, recalling \eqref{eq:bddonsumh3} and the definition of $\beta(\cdot)$ from \eqref{eq:defbetafn} we get,
\begin{align*}
    \E_{\bH_0}\left[\sqrt{n}\E_{\bH_0}\left[\sum_{u=1}^{n}| \tilde{V}_{u}|^3|\cF(\cZ_{n})\right]\right]
    &\leq \frac{\sqrt{K}}{nM_n}\sum_{m=1}^{M_n}\E_{\bH_0}\left[\sum_{u=1}^{n}\sum_{v\in N_G(\sZ_n)(u)}\beta(Z_u)\beta(Z_v)\right]\\
    & = \frac{\sqrt{K}}{n}\E_{\bH_0}\left[\sum_{u=1}^{n}\sum_{v\in N_G(\sZ_n)(u)}\beta(Z_u)\beta(Z_v)\right]\lesssim K^{\frac{3}{2}},
\end{align*}
where the last inequality follows from \eqref{eq:CLTproofuniformN1}, \eqref{eq:CLTproofCauchySch}, \eqref{eq:1stbddbeta}, and \eqref{eq:2ndbddbeta}. This completes the proof of \eqref{eq:VZ}. Collecting the trail of inequalities we now have,
\begin{align}\label{eq:supdiff}
    \sup_{z\in \R}\left|\P_{\bH_0}\left[\frac{\sqrt{nK}D_n}{\tau_n}\leq z\right] - \Phi(z)\right|\lesssim_{d,\sfK} K^{\frac{43}{2}}n^{\frac{3\vep-1}{2}} = o(1)
\end{align}
where the final equality follows by recalling $K = o(n^{1/44})$ and $\vep<\frac{1}{132}$. Replacing $\tau_n$ with $\hat{\tau}_n$ in \eqref{eq:supdiff} by Lemma \ref{lemma:SnapproxvarDn}, gives the result in Theorem \ref{thm:DnCLT}. \hfill $\Box$

\subsubsection{Proof of Lemma \ref{lemma:SnapproxvarDn}} 
\label{sec:estimateMpf}

Drawing parallels between $\hat{\sigma}_n^2$ (recall \eqref{eq:defSn}) and $\tau_n^2$; and between $\eta_n$ (recall \eqref{eq:condeta2exp}) and $\sqrt{nK}D_n$, we will proceed along similar lines as the proof of Proposition \ref{prop:SnestvarH0} from Section \ref{sec:estimatepf}. In particular, as in Lemma \ref{lemma:Sn2VarL2} we first show that,
\begin{align}\label{eq:SndiffvarDn}
    \E_{\bH_0}\left[\left(\hat{\tau}_n^2 - \tau_n^2 \right)^2\right] = o\left(n^{-\frac{\delta}{8 + \delta}}\right).
\end{align}
From the proof of Lemma \ref{lemma:Sn2VarL2} from Section \ref{sec:proofofSn2VarL2}, note that \eqref{eq:SndiffvarDn} holds as long as,
\begin{align*}
    \frac{1}{n}\E_{\bH_0}\left[\max_{1\leq u\neq v\leq n}\bar{\sfH}\left(\tilde{\bW}_{u},\tilde{\bW}_{v}\right)^4\right] = o\left(n^{-\frac{\delta}{8 + \delta}}\right).
\end{align*}
Towards this, using H\"older's inequality gives,  
\begin{align*}
    \left|\bar{\sfH}\left(\tilde{\bW}_{u},\tilde{\bW}_{v}\right)\right| \leq \frac{1}{M_n}\left(\sum_{m=1}^{M_n}\left|\sfH\left(\bW_{u}^{(m)},\bW_{v}^{(m)}\right)\right|^4\right)^\frac{1}{4}M_n^{\frac{3}{4}}.
\end{align*}
Hence, recalling that for any $1\leq m\leq M_n$, the collection 
$\{ ( X_1^{(m)},Y_{1}, Z_{1} ), \ldots, (X_n^{(m)}, Y_{n}, Z_n ) \}$ are generated i.i.d. from $P_{XYZ}$, we get,
\begin{align*}
    \frac{1}{n}\E_{\bH_0}\left[\max_{1\leq u\neq v\leq n}\bar{\sfH}\left(\tilde{\bW}_{u},\tilde{\bW}_{v}\right)^4\right]
    &\leq \frac{1}{nM_n}\sum_{m=1}^{M_n}\E_{\bH_0}\left[\max_{1\leq u\neq v\leq n}\left|\sfH\left(\bW_{u}^{(m)},\bW_{v}^{(m)}\right)\right|^4\right]\\
    & = \frac{1}{n}\E_{\bH_0}\left[\max_{1\leq u\neq v\leq n}\left|\sfH\left(\bW_{u}^{(1)},\bW_{v}^{(1)}\right)\right|^4\right] = o\left(n^{-\frac{\delta}{8 + \delta}}\right),
\end{align*}
where the last equality follows by \eqref{eq:onconvg}. Having shown \eqref{eq:SndiffvarDn}, notice that to complete the proof of Lemma \ref{lemma:SnapproxvarDn} we have to translate \eqref{eq:SndiffvarDn} in to an approximation in terms of the ratio of $\hat{\tau}_n^2$ and $\tau_n^2$. This will be done as in Lemma \ref{lemma:nepscondvarinfty} by proving the following: For all $\vep>0$,
\begin{align}\label{eq:epsVarDn}
    \P_{\bH_0}\left[n^\vep\Var_{\bH_0}\left[\sqrt{nK}D_n\middle|\cF(\sZ_n)\right]>t\right] \ra 1 , 
\end{align}
whenever $t>0$. Again, following the proof of Lemma \ref{lemma:nepscondvarinfty}, to prove \eqref{eq:epsVarDn} it suffices to show that there exists $c>0$ such that 
\begin{align}\label{eq:Hbar2toc}
    \E_{\bH_0}\left[\frac{1}{nK}\sum_{u=1}^{n}\sum_{v\in N_{G(\sZ_n)}(u)}\bar{\sfH}^2(\tilde{\bW}_{u},\tilde{\bW}_{v})\right]\ra c.
\end{align}
To evaluate the expectation on the LHS, recall the sampling scheme from Section \ref{sec:Derandom} to get, 
\begin{align}\label{eq:expandHbar}
    \E_{\bH_0}\bigg[\frac{1}{nK}\sum_{u=1}^{n}\sum_{v\in N_{G(\sZ_n)}(u)}
    &\bar{\sfH}^2(\tilde{\bW}_{u},\tilde{\bW}_{v})\bigg] = \frac{1}{M_n} B_1 + \frac{M_n-1}{M_n} B_2 , 
    \end{align}
    where   
    \begin{align}
    B_1 & : = \E_{\bH_0}\left[\frac{1}{nK}\sum_{u=1}^{n}\sum_{v\in N_{G(\sZ_n)}(u)}\sfH^2(\bW_{u}^{(1)},\bW_{v}^{(1)})\right] , \nonumber \\ 
    B_2 & : = \E_{\bH_0}\left[\frac{1}{nK}\sum_{u=1}^{n}\sum_{v\in N_{G(\sZ_n)}(u)}\sfH(\bW_{u}^{(1)},\bW_{v}^{(1)})\sfH(\bW_{u}^{(2)},\bW_{v}^{(2)})\right] . \nonumber  
\end{align} 
Note, from Lemma \ref{lemma:condH2convg}, in particular \eqref{eq:EH2condconvg} and \eqref{eq:exphZfinite},  $B_1/M_n = o(1)$. Also, by arguments similar to the proof of Lemma \ref{lemma:condH2convg}, $B_2 \rightarrow \E_{\bH_0}[\|\bar{h}(Z)\|_{\cH\otimes\cH}^2]$,     
where $\bar{h}$ is defined as: 
\begin{align*}
    & \bar{h}(Z) \nonumber \\ 
    & := \E_{\bH_0}\left[\sfK(X,\cdot)\otimes\sfK(X,\cdot) - \sfK(X,\cdot)\otimes\sfK(Y^\prime,\cdot) - \sfK(Y,\cdot)\otimes\sfK(X,\cdot) + \sfK(Y,\cdot)\otimes\sfK(Y^\prime,\cdot)|Z\right]  , 
\end{align*} 
with $X\sim P_{X|Z = Z}$ and $ Y,Y^\prime$ are i.i.d. $P_{Y|Z = Z}$ independent of $X|Z$.
Thus, taking $M_n\ra\infty$ in the RHS of \eqref{eq:expandHbar} gives, 
$$\E_{\bH_0}\bigg[\frac{1}{nK}\sum_{u=1}^{n}\sum_{v\in N_{G(\sZ_n)}(u)}
    \bar{\sfH}^2(\tilde{\bW}_{u},\tilde{\bW}_{v})\bigg] \rightarrow \E_{\bH_0}\left[\|\bar{h}(Z)\|_{\cH\otimes\cH}^2\right].$$ 
    Hence, to complete the proof of \eqref{eq:Hbar2toc} (and, as a result,  \eqref{eq:epsVarDn}), it is now enough to show that $\E_{\bH_0}\left[\|\bar{h}(Z)\|_{\cH\otimes\cH}^2\right]\neq 0$. For the sake of contradiction assume that  $\E_{\bH_0}\left[\|\bar{h}(Z)\|_{\cH\otimes\cH}^2\right] = 0$. Then, by arguments similar to the proof of Lemma \ref{lemma:varlimitnon0} we conclude that for $(X,Y,Z)\sim P_{XYZ}$, $X = Y$ almost surely. Recall that by construction that $P_{XYZ} = P_{X|Z}\times P_{Y|Z}\times P_{Z}$, hence from the arguments in the proof of \cite[Theorem 2.1]{deb2020measuring} we arrive at a contradiction. 
    
The proof of Lemma \ref{lemma:SnapproxvarDn} now follows by the arguments in \eqref{eq:completesigmavar} together with \eqref{eq:SndiffvarDn} and \eqref{eq:epsVarDn}. \hfill $\Box$ 

\section{Technical Results}

In this section we collect some technical about tensor product spaces. The following results hold under both $\bH_0$ and $\bH_1$,  hence we refrain from explicitly mentioning the same. 

\begin{definition}\cite[Section 4.6]{berlinet2011reproducing}. 
The tensor product space $\cH\otimes\cH$  is the collection of functions $f\otimes g:\cX\times\cX\ra\R$, where $f, g\in \cH$ and $(f\otimes g)(x,x') := f(x)g(x')$. The inner product in the space $\cH\otimes\cH$ is defined as follows: For $f_1, f_2, g_1, g_2 \in \cH$, 
\begin{align*}
    \left\langle f_{1}\otimes g_{1},f_{2}\otimes g_{2}\right\rangle_{\cH\otimes\cH} = \left\langle f_{1}, f_{2}\right\rangle_{\cH}\left\langle g_{1}, g_{2}\right\rangle_{\cH}.
\end{align*}
\end{definition}
\begin{lemma}\label{lemma:condH2convg} 
   For $(X,Y,Z)\sim P_{XYZ}$, define $h_{Z} \in \cH\otimes\cH$ as follows:
    \begin{align*}
        h_{Z} := \E\left[\sfK(X,\cdot)\otimes\sfK(X,\cdot) - \sfK(X,\cdot)\otimes\sfK(Y,\cdot) - \sfK(Y,\cdot)\otimes\sfK(X,\cdot) + \sfK(Y,\cdot)\otimes\sfK(Y,\cdot)|Z\right] . 
    \end{align*} 
    Suppose Assumption \ref{assumption:K} and Assumption \ref{assumption:Kn} holds. Moreover, suppose the kernel $\sfK$ satisfies \begin{align}\label{eq:momentassumpXY}
        \int\sfK(x,x)^{4+\delta}\rmd P_{X}(x)<\infty \text{ and }\int\sfK(y,y)^{4+\delta}\rmd P_{Y}(y)<\infty,
    \end{align}
    for some $\delta>0$. Then with $K = o(n/\log n)$, as $n \rightarrow \infty$
    \begin{align}\label{eq:EH2condconvg}
        \E\left[\frac{1}{nK}\sum_{u=1}^{n}\sum_{v=1}^{n}\sfH^{2}\left(\bW_{u},\bW_{v}\right)\one\left\{(u,v)\in E(G(\sZ_n))\right\}\right]\ra \E\left[\left\|h_{Z_{1}}\right\|_{\cH\otimes\cH}^2\right]
    \end{align}
    where $\bW_i = (X_i, Y_i)$ for all $1\leq i\leq n$, $\sZ_n = \{Z_1,\ldots, Z_n\}$ and $\{(X_i, Y_i, Z_i), 1\leq i\leq n\}$ are generated independently from $P_{XYZ}$.
\end{lemma}
\begin{proof}
For $\bw = (x,y)\in \cX\times\cX$ define,
    \begin{align*}
        \sfL_{\bw} := \sfK(x,\cdot)\otimes\sfK(x,\cdot) - \sfK(x,\cdot)\otimes\sfK(y,\cdot) - \sfK(y,\cdot)\otimes\sfK(x,\cdot) + \sfK(y,\cdot)\otimes\sfK(y,\cdot).
    \end{align*}
    By definition, $\sfL_{\bw} \in \cH\otimes\cH$ and from \eqref{eq:Hasinprod} it follows that,
    \begin{align*}
        \sfH^{2}(\bw,\bw') = \left\langle\sfL_{\bw},\sfL_{\bw'} \right\rangle
    \end{align*}
    for all $\bw, \bw' \in \cX\times\cX$. Now, following the arguments in the proof of \eqref{eq:Hbarlimitdiffbdd} it can shown that, 
    \begin{align}\label{eq:hbarvarlimitidff}
        & \left| \E\left[\frac{1}{nK}\sum_{u=1}^{n}\sum_{v=1}^{n}\sfH^{2}\left(\bW_{u},\bW_{v}\right)\one\left\{(u,v)\in E(G(\sZ_n))\right\}\right]  - \E\left[\left\|h_{Z_1}\right\|_{\cH\otimes\cH}^2\right]\right| \nonumber \\ 
       & \leq \sqrt{\E\left[\left\|h_{Z_1}\right\|_{\cH\otimes\cH}^2\right]}\sqrt{\E\left[\left\|h_{Z_{N(1)}} - h_{Z_{1}}\right\|_{\cH\otimes\cH}^2\right]} , 
    \end{align}
    where $N(1)$ is a uniformly selected index from the neighbors of vertex $1$ in the graph $G(\sZ_n)$. From \cite[Lemma D.2]{deb2020measuring},
    \begin{align}\label{eq:hZN1}
        \E\left[\left\|h\left(Z_{N(1)}\right) - h_{Z_{1}}\right\|_{\cH\otimes\cH}^4\right]\lesssim \E\left[\left\|h_{Z_{1}}\right\|_{\cH\otimes\cH}^4\right].
    \end{align}
    Also, as in the proof of \eqref{eq:gZ1},
    \begin{align}
        \E\left[\left\|h_{Z_{1}}\right\|_{\cH\otimes\cH}^4\right]\lesssim
        &\ \E\left[\sfK(X_{1},X_{1})^4 + 2\sfK(X_{1},X_{1})^2\sfK(Y_{1},Y_{1})^2 + \sfK(Y_{1},Y_{1})^4\right]<\infty , \label{eq:exphZfinite}
    \end{align}
    where the finiteness follows from the moment assumption in \eqref{eq:momentassumpXY}. Combining \eqref{eq:hZN1} and \eqref{eq:exphZfinite} shows that $\|h_{Z_{N(1)}} - h_{Z_{1}}\|_{\cH\otimes\cH}^2$ is uniformly integrable, Hence, by \cite[Lemma D.3]{deb2020measuring} $$\left\|h_{Z_{N(1)}} - h_{Z_{1}} \right\|_{\cH\otimes\cH}^2 \pto 0.$$ Using this and \eqref{eq:exphZfinite} in \eqref{eq:hbarvarlimitidff} completes the proof of Lemma \ref{lemma:condH2convg}. 
\end{proof}
In the following lemma we prove that the limit from \eqref{eq:EH2condconvg} is strictly positive.
\begin{lemma}\label{lemma:varlimitnon0}
    Suppose the assumptions in Lemma \ref{lemma:condH2convg} hold. Also, assume that $(X,Y,Z)\sim P_{XYZ}$ is such that $P_{XY}\left[X\neq Y\right]>0$. Then for $h_Z \in \cH\otimes\cH$ as defined in Lemma \ref{lemma:condH2convg},
    \begin{align*}
        \E\left[\left\|h_{Z}\right\|_{\cH\otimes\cH}^2\right] > 0.
    \end{align*}
\end{lemma} 
\begin{proof} The proof will proceed by contradiction. For this, suppose, if possible, $\E[\|h_{Z}\|_{\cH\otimes\cH}^2] = 0$. Then for any $r\in \cH$ by \cite[Lemma 3.3]{park2020measure} we have,
    \begin{align}\label{eq:expinnerprod}
        \E\left[\left|\left\langle r\otimes r, h_{Z}\right\rangle_{\cH\otimes\cH} \right|\right] = \E\left[\left|r(X)^2 - 2r(X)r(Y) + r(Y)^2\right|\right] = \E\left[(r(X) - r(Y))^2\right]
    \end{align}
    By the Cauchy-Schwarz inequality,
    \begin{align*}
        \E\left[\left|\left\langle r\otimes r, h_{Z}\right\rangle_{\cH\otimes\cH} \right|\right]\leq \left\|r\otimes r\right\|_{\cH\otimes\cH}\sqrt{\E\left[\left\|h_{Z}\right\|_{\cH\otimes\cH}^2\right]} = 0 .
    \end{align*}
    Then \eqref{eq:expinnerprod} implies, $r(X) = r(Y)$ almost surely for all $r\in \cH$. Recalling the seperability of $\cH$ we can extend the above conclusion to $r(X) = r(Y)$ for all $r\in \cH$ almost surely.  
    Then, for all $x\in \cX$ choosing $r(\cdot) = \sfK(x,\cdot)$ gives,
    \begin{align*}
        \sfK(x,X) = \sfK(x,Y)\ \text{ for all } x\in \cX\text{ almost surely.}
    \end{align*}
    Hence, $\sfK(X,\cdot) = \sfK(Y,\cdot)$ almost surely, and by  \cite[Proposition 14]{sejdinovic2013equivalence} we conclude that $X = Y$ almost surely. This contradicts the assumption $P_{XY}(X\neq Y)>0$ and completes the proof of Lemma \ref{lemma:varlimitnon0}. 
\end{proof}

\section{Reliability Diagrams and Recalibration}
\label{sec:rdr}

This section is organized as follows. In Section \ref{sec:calibrationfigures} we provide a brief of review reliability diagrams and describe how the expected calibration error (ECE) is computed in the experiments in Section \ref{sec:imagecalibration}. In Section \ref{sec:recalibration} we discuss the recalibration method based on isotonic regression.

\subsection{Reliability Diagrams and Computation of ECE} 
\label{sec:calibrationfigures}

Suppose $Y\in \{0,1\}$ is a binary response associated with feature vector $Z\in \cZ$. Let $f:\cZ\ra [0,1]$ be a predicitve model, such that $f(Z)$ estimates $\P\left[Y=1\middle| Z\right]$ (notice that for each $z\in \cZ$, $f(z)$ defines a distribution over $\{0,1\}$). Furthemore, we assume access to test samples $\{(Y_i, Z_i): 1\leq i\leq n\}$ to construct reliability diagrams and estimate the ECE.

\subsubsection{Reliability Diagrams}\label{sec:reliability}

Reliability diagrams are visual tools to represent calibration of predictive model from finite samples (see \cite{degroot1983comparison,niculescu2005predicting,gweon2019reliable,brocker2007increasing}). In this section we follow the construction from \citet{gweon2019reliable} and \citet{brocker2007increasing} to briefly describe reliability diagrams in the context of calibration for binary predictive models. Towards this, recall from \eqref{eq:classcalibration} that a predictive model $f$ is said to be calibrated if and only if,
\begin{align}\label{eq:binarycal}
    \P\left[Y=1|f(Z)\right] = f(Z)\text{ almost surely }P_Z.
\end{align}
The continuous nature of the conditioning random variable makes it hard to have a visual depiction of the above identity from finite samples. Thus, one adopts a binning approach to approximate conditional probabilities in \eqref{eq:binarycal}. To this end, fix $M \geq 2$ and for $m\in [M] := \{1, 2, \ldots, M\}$ let $I_m = \left(\frac{m-1}{M}, \frac{m}{M}\right]$. Then by \eqref{eq:binarycal} for a perfectly calibrated model,
\begin{align}\label{eq:condexpreliability}
    \P\left[Y = 1|f(Z)\in I_m\right] = \E\left[f(Z)|f(Z)\in I_m\right]\text{ for all }, 1\leq m\leq M.
\end{align} 
We can now estimate the quantities in  \eqref{eq:condexpreliability} from the test samples as follows: For all $m\in [M]$, let $B_m$ be the set of indices of test samples whose prediction falls into the interval $I_m$. Observe that,
\begin{align}\label{eq:LandRBm}
    L(B_m) = \frac{1}{|B_m|}\sum_{i\in B_m}\one\left\{Y_i=1\right\}\text{ and }R(B_m) = \frac{1}{|B_m|}\sum_{i\in B_m}f(Z_i)
\end{align}
are natural estimates of $\P\left[Y = 1|f(Z)\in I_m\right]$ and $\E\left[f(Z)|f(Z)\in I_m\right]$, respectively. Then the reliability diagram is the curve joining the points $\{(L(B_m),R(B_m)):1\leq m\leq M\}$. For calibrated predictive models $f$ we expect the curve to be close to the diagonal line. 

\subsubsection{Computation of ECE} 
\label{sec:calibrationerror}

The Expected Calibration Error (ECE) of a predictor $f$ \cite{guo2017calibration, naeini2015obtaining}, also known as the mean calibration error \cite{frank2015regression}, is defined as,
\begin{align}\label{eq:defECE}
    \text{ECE}(f) = \E_{Z\sim P_{Z}}\left[\left|\P\left[Y=1\middle|f(Z)\right] - f(Z)\right|\right],
\end{align}
which is the expectation of the absolute difference of both sides in the identity from \eqref{eq:binarycal}. By definition, a predictive model $f$ is calibrated if and only if $\text{ECE}(f) = 0$. (In \eqref{eq:defECE} define the ECE with the $\ell_1$ norm, but other variants using the $\ell_p$ norm, for $p\geq 1$, are also studied (see \cite{lee2023,arrieta2022metrics}).) To estimate ECE from test samples we once again follow the binning strategy from Section \ref{sec:reliability}. Following \cite{lee2023}, first note that $\text{ECE}(f)$ can be approximated by piecewise averaging as follows:
\begin{align*}
    \text{ECE}(f)
    & \geq \sum_{m=1}^{M}\P(f(Z)\in I_m)\left|\E\left[\P\left[Y=1\middle|f(Z)\right] - f(Z)\middle|f(Z)\in I_m\right]\right|\\
    & = \sum_{m=1}^{M}\P\left(f(Z)\in I_m\right)\left|\P\left[Y=1\middle|f(Z)\in I_m\right] - \E\left[f(Z)\middle|f(Z)\in I_m\right]\right|.
\end{align*}
Recalling \eqref{eq:LandRBm} the RHS above can be estimated as,
\begin{align*}
    \widehat{\text{ECE}}(f) = \sum_{m=1}^{M}\frac{|B_m|}{n}\left|L(B_m) - R(B_m)\right|,
\end{align*}
which is used as an approximate measure of $\text{ECE}(f)$ from test samples. 

\subsection{Recalibration Based on Isotonic Regression}
\label{sec:recalibration} 

 Isotonic regression is a well-known approach for recalibrating the prediction probabilities of a binary classifier (see \cite{guo2017calibration,berta2024classifier,zadrozny2002transforming,zhong2013accurate} and references therein). Here, we briefly recall this method. For this, suppose in a binary classification problem a predictive model outputs potentially miscalibrated probabilities $\hat p_1,\ldots,\hat p_n$ for the positive class in a test set of size $n$. To recalibrate the probabilities one learns a non-decreasing, piecewise constant transformation $\hat q_i = \mathsf{q}(\hat p_i)$ that minimizes the squared error loss $\sum_{i=1}^n\left(\mathsf{q}(\hat p_i) -Y_i\right)^2$, where $Y_i\in \{0,1\}$ is the observed response of the $i$-th test sample. Leveraging the piecewise constant condition on $\mathsf{q}$, isotonic regression considers the following optimization problem:
\begin{align*}
    \min_{\overset{M,\theta_1,\ldots,\theta_M}{a_1,\ldots,a_{M+1}}}\sum_{m=1}^M \sum_{i=1}^n \mathbf{1}\{a_m\leq \hat p_i<a_{m+1}\}(\theta_m-y_i)^2 , 
\end{align*}
such that $0=a_1\leq a_1\leq \ldots\leq a_{M+1}=1$ and $\theta_1\leq \theta_1\leq \ldots\leq \theta_M$. Here, $M$ is the number of intervals, $\{a_i:1\leq i\leq M\}$ are the endpoints of the intervals and 
$\theta_1,\ldots,\theta_M$ are the function values.

\section{Additional Simulations} 
\label{sec:additionalsimulations}

In this section we present some additional simulation results.  In Appendix \ref{sec:additionalclassification} we present additional simulations for calibration tests for classification. In Appendix \ref{sec:computation} we compare the running times of the SKCE and the $\EMMD$ tests. Additional results for testing calibration using the CIFAR-10 dataset are given in Appendix \ref{sec:imagecalibration3}. 

\subsection{Additional Simulations Calibration Tests for Classification}
\label{sec:additionalclassification}

Consider the same setting as in Section \ref{sec:classification} with sample size $n=75$. Figure \ref{fig:classcalibration2} shows the empirical Type I error and power over 500 repetitions. The findings from the results are summarized in Section \ref{sec:classification}. 
    
  \begin{figure}[ht]
    \begin{subfigure}{0.5\textwidth}
        \centering 
        \small
        \centering\begingroup\fontsize{11}{20}\selectfont

\begin{tabular}{cccccc}
\multicolumn{6}{c}{Type I Error} \\
\toprule
\multirow{2}{*}{$\rho$} &\multirow{2}{*}{SKCE} & \multicolumn{2}{c}{Asymptotic} & \multicolumn{2}{c}{Derandomized}\\
\cmidrule(l{3pt}r{3pt}){3-6}
& & 15 NN & 25 NN &  15 NN & 25 NN\\
\midrule
\cellcolor{gray!6}{0.1} & \cellcolor{gray!6}{0.064} & \cellcolor{gray!6}{0.050} & \cellcolor{gray!6}{0.054} & \cellcolor{gray!6}{0.048} & \cellcolor{gray!6}{0.060}\\
0.2 & 0.074 & 0.038 & 0.052 & 0.062 & 0.052\\
\cellcolor{gray!6}{0.3} & \cellcolor{gray!6}{0.048} & \cellcolor{gray!6}{0.050} & \cellcolor{gray!6}{0.066} & \cellcolor{gray!6}{0.040} & \cellcolor{gray!6}{0.050}\\
0.4 & 0.064 & 0.046 & 0.052 & 0.044 & 0.048\\
\cellcolor{gray!6}{0.5} & \cellcolor{gray!6}{0.066} & \cellcolor{gray!6}{0.048} & \cellcolor{gray!6}{0.046} & \cellcolor{gray!6}{0.030} & \cellcolor{gray!6}{0.056}\\
\bottomrule
\end{tabular}
\endgroup{}

        \caption*{(a)}
    \end{subfigure}%
    \begin{subfigure}{0.5\textwidth}
        \centering
         \includegraphics[height=6.25cm]{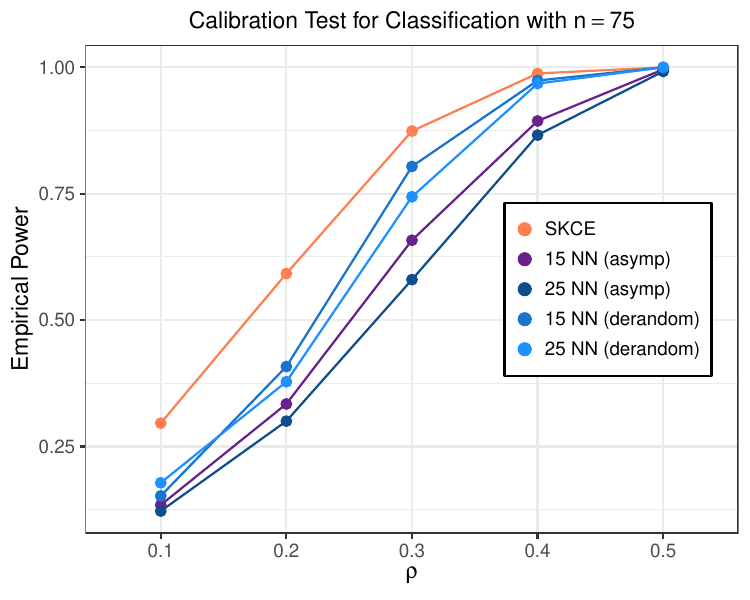} \\
             \small{ (b) }
    \end{subfigure}
    \caption{\small Calibration tests for classification: (a) Type I error and (b) empirical power for $n=75$, as a function of the signal strength $\rho$. } 
    \label{fig:classcalibration1}
\end{figure}

\subsection{Time Comparison Between the SKCE and $\EMMD$ Tests} 
\label{sec:computation}

In this section we report the running times of the SKCE and $\EMMD$ tests for the simulations settings described in Section \ref{sec:classification} and Section \ref{sec:regression}. Figure \ref{fig:computationcalibration}(a) shows the box plots of the running times (over 500 repetitions) of the different tests for the classification setup described in Section \ref{sec:classification} with sample size $n=100$. Figure \ref{fig:computationcalibration}(b) shows the box plots of the running times of the different tests  for the regression setup described in  \eqref{eq:YZ} with sample size $n=75$. In both cases the $\EMMD$ tests are computationally much faster than the SKCE tests. 

\begin{figure}[!ht]
    \begin{subfigure}{0.5\textwidth}
        \centering
        \includegraphics[height=4.75cm]{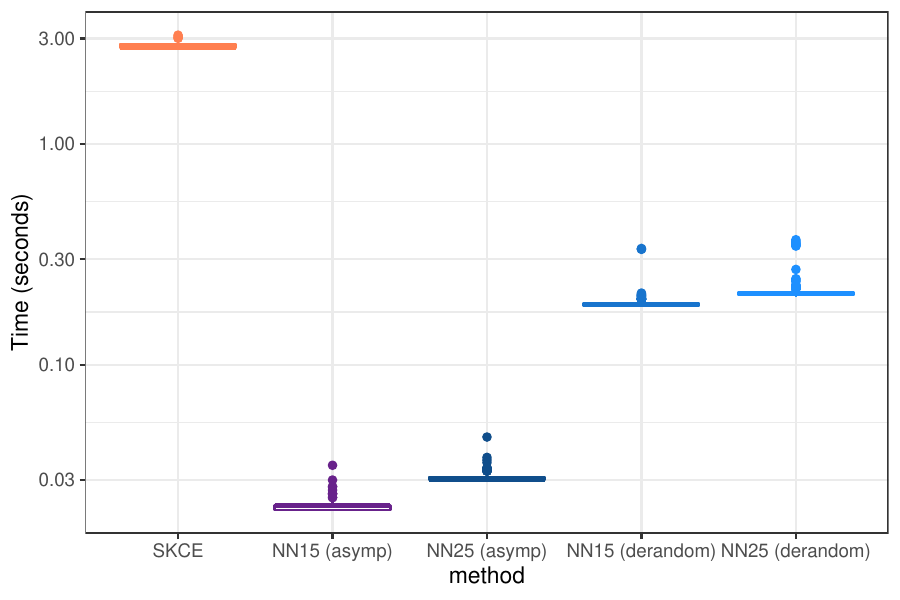} \\ 
        \small{ (a) }
    \end{subfigure}%
    \begin{subfigure}{0.5\textwidth}
        \centering
        \includegraphics[height=4.75cm]{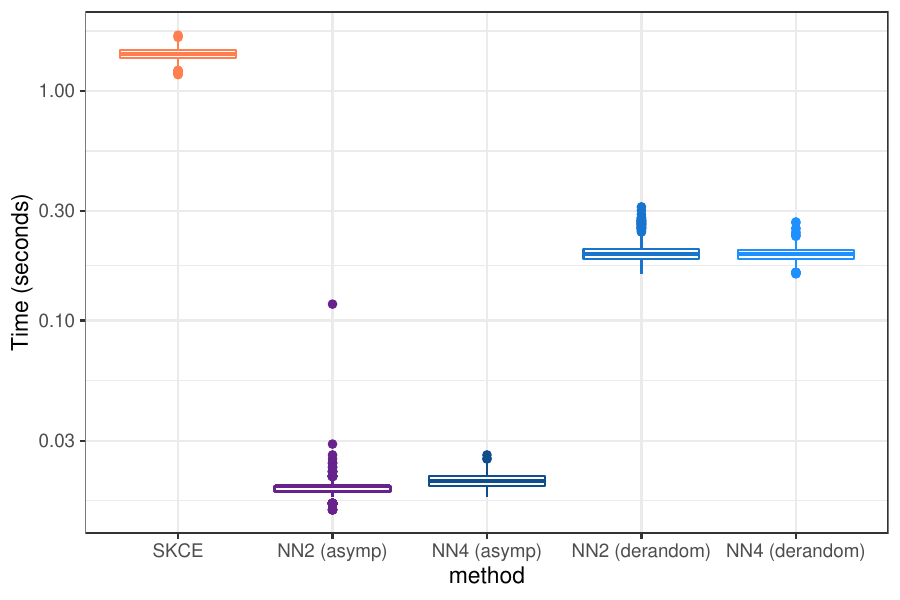} \\ 
        \small{ (b) }
    \end{subfigure}
    \caption{ \small{ Running times of the different tests for (a) the classification setup described Section \ref{sec:classification} with sample size $n=100$ and (b) the regression setup described in \eqref{eq:YZ} with sample size $n=75$. } } 
    \label{fig:computationcalibration}
\end{figure}

\subsection{Additional Results on Calibration Tests for Real Data}
\label{sec:imagecalibration3}

In this section we show the reliability diagrams and results for the $\EMMD$ based calibration tests (before and after recalibration)
for the following  3 pairs of images: 
\begin{align*}\label{eq:additionalimagepairs}
   \{\text{cat, deer}\}, \{\text{cat, frog}\}, \{\text{cat, horse}\}.
\end{align*} 
The experimental setup is same as that in Section \ref{sec:imagecalibration}. 

\begin{figure}[!ht]
    \begin{subfigure}{0.45\textwidth}
        \centering
        \includegraphics[scale=0.55]{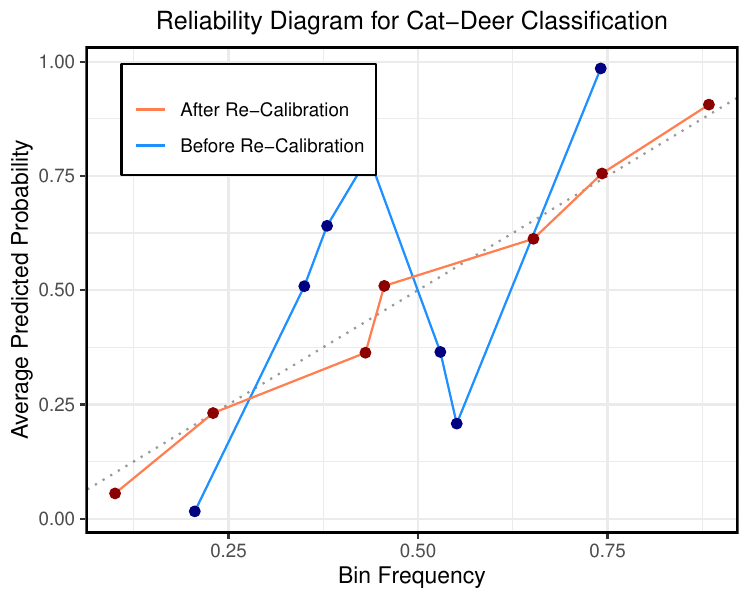}
        \caption*{(a)}
    \end{subfigure}%
    \begin{subfigure}{0.55\textwidth}
        \centering 
        \small
        \begin{tabular}{cccc}
\toprule
\multicolumn{1}{c}{ } & \multicolumn{1}{c}{ } & \multicolumn{2}{c}{Rejection Proportion} \\
\cmidrule(l{3pt}r{3pt}){3-4}
Test & \textit{K} & \multicolumn{1}{p{2.5cm}}{\centering Before\\ Re-Calibration} & \multicolumn{1}{p{2.5cm}}{\centering After\\ Re-Calibration}\\
\midrule
\cellcolor{gray!6}{FS} & \cellcolor{gray!6}{60} & \cellcolor{gray!6}{1.00} & \cellcolor{gray!6}{0.09}\\
FS & 80 & 1.00 & 0.10\\
\cellcolor{gray!6}{FS} & \cellcolor{gray!6}{100} & \cellcolor{gray!6}{1.00} & \cellcolor{gray!6}{0.06}\\
FS & 120 & 1.00 & 0.04\\
\midrule
\cellcolor{gray!6}{Asymp} & \cellcolor{gray!6}{60} & \cellcolor{gray!6}{1.00} & \cellcolor{gray!6}{0.14}\\
Asymp & 80 & 1.00 & 0.09\\
\cellcolor{gray!6}{Asymp} & \cellcolor{gray!6}{100} & \cellcolor{gray!6}{1.00} & \cellcolor{gray!6}{0.09}\\
Asymp & 120 & 1.00 & 0.06\\
\midrule
\cellcolor{gray!6}{ECE} & \cellcolor{gray!6}{} & \cellcolor{gray!6}{0.48} & \cellcolor{gray!6}{0.18}\\
\bottomrule
\end{tabular}

        \caption*{(b)}
    \end{subfigure}
    \caption{ \small{ Results for cat-deer classification: (a) reliability plot before recalibration and after recalibration, and (b) $p$-values of the $\EMMD$ test for different values of $K$ and ECE before and after recalibration. } }
    \label{fig:reliability-calibration-3} 
\end{figure}

\begin{figure}[ht]
    \begin{subfigure}{0.45\textwidth}
        \centering
        \includegraphics[scale=0.55]{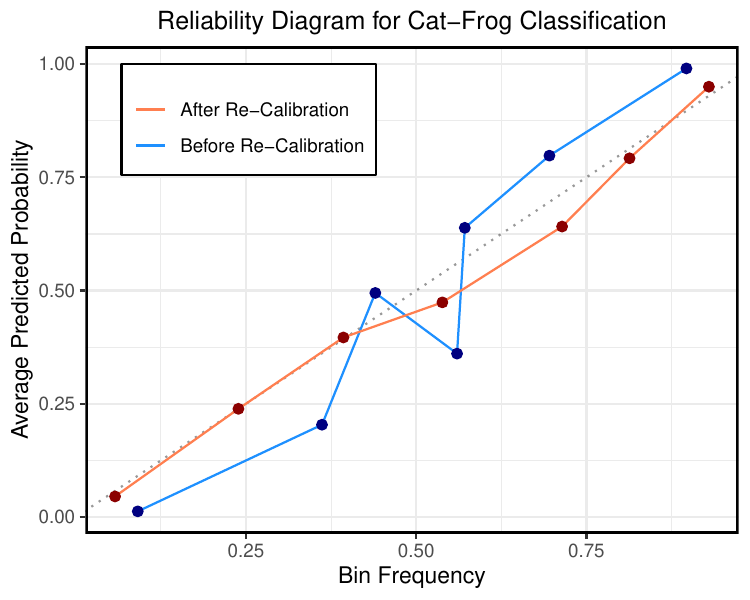}
        \caption*{(a)}
    \end{subfigure}%
    \begin{subfigure}{0.55\textwidth}
        \centering 
        \small
        \begin{tabular}{cccc}
\toprule
\multicolumn{1}{c}{ } & \multicolumn{1}{c}{ } & \multicolumn{2}{c}{Rejection Proportion} \\
\cmidrule(l{3pt}r{3pt}){3-4}
Test & \textit{K} & \multicolumn{1}{p{2.5cm}}{\centering Before\\ Re-Calibration} & \multicolumn{1}{p{2.5cm}}{\centering After\\ Re-Calibration}\\
\midrule
\cellcolor{gray!6}{FS} & \cellcolor{gray!6}{60} & \cellcolor{gray!6}{1.00} & \cellcolor{gray!6}{0.02}\\
FS & 80 & 1.00 & 0.04\\
\cellcolor{gray!6}{FS} & \cellcolor{gray!6}{100} & \cellcolor{gray!6}{1.00} & \cellcolor{gray!6}{0.05}\\
FS & 120 & 1.00 & 0.09\\
\midrule
\cellcolor{gray!6}{Asymp} & \cellcolor{gray!6}{60} & \cellcolor{gray!6}{1.00} & \cellcolor{gray!6}{0.03}\\
Asymp & 80 & 1.00 & 0.02\\
\cellcolor{gray!6}{Asymp} & \cellcolor{gray!6}{100} & \cellcolor{gray!6}{1.00} & \cellcolor{gray!6}{0.06}\\
Asymp & 120 & 1.00 & 0.06\\
\midrule
\cellcolor{gray!6}{ECE} & \cellcolor{gray!6}{} & \cellcolor{gray!6}{0.23} & \cellcolor{gray!6}{0.14}\\
\bottomrule
\end{tabular}

        \caption*{(b)}
    \end{subfigure}
    \caption{ \small{ Results for cat-frog classification: (a) reliability plot before recalibration and after recalibration, and (b) $p$-values of the $\EMMD$ test for different values of $K$ and ECE before and after recalibration. } }
    \label{fig:reliability-calibration-4} 
\end{figure}

\begin{figure}[ht]
    \begin{subfigure}{0.45\textwidth}
        \centering
        \includegraphics[scale=0.55]{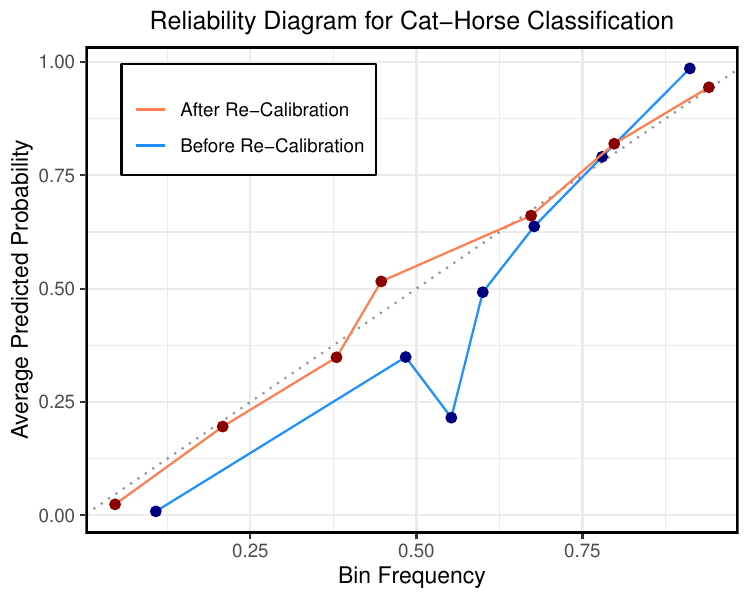}
        \caption*{(a)}
    \end{subfigure}%
    \begin{subfigure}{0.55\textwidth}
        \centering 
        \small
        \begin{tabular}{cccc}
\toprule
\multicolumn{1}{c}{ } & \multicolumn{1}{c}{ } & \multicolumn{2}{c}{Rejection Proportion} \\
\cmidrule(l{3pt}r{3pt}){3-4}
Test & \textit{K} & \multicolumn{1}{p{2.5cm}}{\centering Before\\ Re-Calibration} & \multicolumn{1}{p{2.5cm}}{\centering After\\ Re-Calibration}\\
\midrule
\cellcolor{gray!6}{FS} & \cellcolor{gray!6}{60} & \cellcolor{gray!6}{1.00} & \cellcolor{gray!6}{0.14}\\
FS & 80 & 1.00 & 0.18\\
\cellcolor{gray!6}{FS} & \cellcolor{gray!6}{100} & \cellcolor{gray!6}{1.00} & \cellcolor{gray!6}{0.13}\\
FS & 120 & 1.00 & 0.10\\
\midrule
\cellcolor{gray!6}{Asymp} & \cellcolor{gray!6}{60} & \cellcolor{gray!6}{1.00} & \cellcolor{gray!6}{0.18}\\
Asymp & 80 & 1.00 & 0.21\\
\cellcolor{gray!6}{Asymp} & \cellcolor{gray!6}{100} & \cellcolor{gray!6}{1.00} & \cellcolor{gray!6}{0.19}\\
Asymp & 120 & 1.00 & 0.14\\
\midrule
\cellcolor{gray!6}{ECE} & \cellcolor{gray!6}{} & \cellcolor{gray!6}{0.23} & \cellcolor{gray!6}{0.13}\\
\bottomrule
\end{tabular}

        \caption*{(b)}
    \end{subfigure}
    \caption{ \small{ Results for cat-horse classification: (a) reliability plot before recalibration and after recalibration, and (b) $p$-values of the $\EMMD$ test for different values of $K$ and ECE before and after recalibration. } }
    \label{fig:reliability-calibration-5} 
\end{figure}

\end{document}